\documentclass[manyauthors]{fundam}
\usepackage{hyperref,aliascnt,import}
\hypersetup{hypertexnames=false}
\usepackage[capitalise,nameinlink]{cleveref}
\usepackage{amssymb,thm-restate}
\usepackage{stmaryrd}
\usepackage{mathtools}
\usepackage{xcolor}
\usepackage{todonotes}
\usepackage{expl3}
\usepackage[electronic,notion]{knowledge}
\usepackage{mathcommand}
\usepackage{xspace}
\usepackage{paralist}
\hypersetup{hidelinks,breaklinks}

\knowledgedirective{math notion}{notion}

\IfKnowledgePaperModeTF{}{
  \knowledgestyle{intro notion}{color=black, emphasize}
  \knowledgestyle{notion}{color=black}
}

\newaliascnt{theoremcounter}{definition}
\newtheorem{theoremOK}[theoremcounter]{Theorem}
\newaliascnt{factcounter}{definition}
\newtheorem{factOK}[factcounter]{Fact}
\crefname{factOK}{Fact}{Facts}
\newaliascnt{lemmacounter}{definition}
\newtheorem{lemmaOK}[lemmacounter]{Lemma}
\crefname{lemmaOK}{Lemma}{Lemmas}  
\newaliascnt{examplecounter}{definition}
\newtheorem{exampleOK}[examplecounter]{Example}
\newaliascnt{remarkcounter}{definition}
\newtheorem{remarkOK}[remarkcounter]{Remark}
\newaliascnt{corollarycounter}{definition}
\newtheorem{corollaryOK}[corollarycounter]{Corollary}
\newaliascnt{claimcounter}{definition}
\newtheorem{claimOK}[claimcounter]{Claim}

\crefname{figure}{Figure}{Figures}
\crefname{formula}{Formula}{Formulas}
\creflabelformat{formula}{#2\textup{(#1)}#3}
\crefname{equality}{Equality}{Equalities}
\creflabelformat{equality}{#2\textup{(#1)}#3}
\crefname{equalities}{Equalities}{Equalities}
\creflabelformat{equalities}{#2\textup{(#1)}#3}
\crefname{equation}{Rule}{Rules}

\allowdisplaybreaks

\newcommand{\mso}{{\upshape MSO}\xspace}
\newcommand{\unbound}{\ensuremath{\mathsf U}\xspace}
\newcommand{\msou}{{\upshape MSO+}\unbound}
\newcommand{\wmsou}{{\upshape WMSO+}\unbound}

\newcommand{\ignore}[1]{}

\knowledge\Alphabet{math notion}
\knowledge\AlphabetXx{math notion} 
\knowledge\Counters{math notion}

\newcommand{\EXP}[1]{#1\textsf{-EXP}}

\newcommand{\coNEXP}[1]{\textsf{co-}#1\textsf{-NEXP}}

\newrobustcmd{\set}[1]{\left\{{#1}\right\}}
\newcommand{\setof}[2]{\set{#1 \;|\; #2}}

\newrobustcmd{\tuple}[1]{\langle#1\rangle}

\newrobustcmd{\sem}[1]{\kl[\sem]{\llbracket}#1\kl[\sem]{\rrbracket}}
\knowledge\sem{math notion}

\newrobustcmd{\semdt}[2]{\kl[\semdt]{\llbracket}#2\kl[\semdt]{\rrbracket}_{#1}}
\knowledge\semdt{math notion}

\newrobustcmd{\ctree}[1]{\kl[\ctree]{\llparenthesis}#1\kl[\ctree]{\rrparenthesis}}
\knowledge\ctree{notion}

\newrobustcmd{\lang}[1]{\kl[\lang]{\Ll}(#1)}
\knowledge\lang{math notion}

\newrobustcmd{\STRErobust}{\kl[\STRE]{\textsf{STRE}}}
\newcommand{\STRE}{\STRErobust\xspace}
\knowledge\STRE{math notion}

\newrobustcmd{\SUProbust}{\kl[\SUP]{\textsf{SUP}}}
\newcommand{\SUP}{\SUProbust\xspace}
\knowledge\SUP{math notion}

\newrobustcmd{\makeshallow}[1]{\mathsf{mkshallow}(#1)}
\newrobustcmd{\collect}[1]{\mathsf{collect}(#1)}
\newrobustcmd{\concat}{\mathop{+\!\!+}}

\newrobustcmd{\Alphabet}{\kl[\Alphabet]{\mathbb{A}}}
\newrobustcmd{\AlphabetXx}{\kl[\AlphabetXx]{\mathbb{A}_\Xx}}
\newrobustcmd{\UnivAlphabet}{\kl[\Alphabet]{\mathbb{U}}}
\newrobustcmd{\Counters}{\kl[\Counters]{\mathbb{C}}}
\newrobustcmd{\Nat}{\mathbb{N}}
\newrobustcmd{\Aa}{\mathcal{A}}
\newrobustcmd{\Bb}{\mathcal{B}}
\newrobustcmd{\Cc}{\mathcal{C}}
\newrobustcmd{\Class}{\protect{\ensuremath{\mathfrak{C}}}}
\newrobustcmd{\Gg}{\mathcal{G}}
\newrobustcmd{\Hh}{\mathcal{H}}
\newrobustcmd{\Ii}{\mathcal{I}}
\newrobustcmd{\Ll}{\mathcal{L}}
\newrobustcmd{\Mm}{\mathcal{M}}
\newrobustcmd{\Nn}{\mathcal{N}}
\newrobustcmd{\Rr}{\mathcal{R}}
\newrobustcmd{\Ss}{\mathcal{S}}
\newrobustcmd{\Xx}{\mathcal{X}}

\newcommand{\down}{\downarrow}
\newcommand{\downwardclosure}[1]{{#1}{\kl[\downwardclosure]\down}}
\knowledge\downwardclosure{notion}
\newcommand{\embedsinto}{\mathrel{\kl[\embedsinto]{\sqsubseteq}}}
\knowledge\embedsinto{notion}

\newcommand{\hole}{\kl[\hole]{\scalebox{0.6}{$\square$}}}
\knowledge\hole{notion}
\newcommand{\opt}[1]{{#1}^{\kl[\opt]{?} \!}}
\knowledge\opt{notion}
\newcommand{\iter}[1]{{#1}^{\kl[\iter]* \!}}
\knowledge\iter{notion}

\newcommand{\replace}[2]{#1[#2]}

\newcommand{\alt}{\;|\;}
\newcommand{\pr}{\kl[\pr]{\mathit{pr}}}
 \knowledge\pr{math notion}
\newcommand{\rootLab}{\mathit{root}}

\newrobustcmd{\rank}{\kl[\rank]{\mathit{rank}}}
\knowledge\rank{math notion}
\newrobustcmd{\otyp}{\kl[\otyp]{\mathsf{o}}}
\knowledge\otyp{math notion}
\newrobustcmd{\arr}{\mathbin{\kl[\arr]{\to}}}
\knowledge\arr{math notion}
\newrobustcmd{\ord}{\kl[\ord]{\mathit{ord}}}
\knowledge\ord{math notion}
\newrobustcmd{\BT}[1]{\kl[\BT]{\mathsf{BT}({#1})}}
\knowledge\BT{math notion}

\newrobustcmd{\lamdots}{.\cdots{}.}

\newrobustcmd{\tconst}[1]{\kl[\tconst]{\ensuremath{\overline{#1}}}}
\knowledge\tconst{math notion}

\newrobustcmd{\tvar}[1]{\kl[\tvar]{\ensuremath{\overline{#1}}}}
\knowledge\tvar{math notion}

\newrobustcmd{\tlambda}[1]{\kl[\tlambda]{\ensuremath{\overline{\LaTeXlambda #1}}}}
\newrobustcmd{\tlx}{\tlambda x}
\knowledge\tlambda{math notion}

\newrobustcmd{\tat}{\kl[\tat]{@}}
\knowledge\tat{math notion}

\newcommand{\ensuremathsf}[1]{\ensuremath{\mathsf{#1}}}
\newrobustcmd{\leta}{\ensuremathsf{a}}
\newrobustcmd{\letb}{\ensuremathsf{b}}
\newrobustcmd{\letbone}{\ensuremathsf{b_1}}
\newrobustcmd{\letbtwo}{\ensuremathsf{b_2}}
\newrobustcmd{\letc}{\ensuremathsf{c}}
\newrobustcmd{\letcone}{\ensuremathsf{c_1}}
\newrobustcmd{\letctwo}{\ensuremathsf{c_2}}
\newrobustcmd{\lete}{\ensuremathsf{e}}
\newrobustcmd{\letq}{\ensuremathsf{q}}
\newrobustcmd{\stap}{\ensuremathsf{p}}
\newrobustcmd{\varN}{\ensuremathsf{N}}
\newrobustcmd{\varf}{\ensuremathsf{f}}
\newrobustcmd{\varg}{\ensuremathsf{g}}
\newrobustcmd{\vart}{\ensuremathsf{t}}
\newrobustcmd{\varx}{\ensuremathsf{x}}
\newrobustcmd{\vary}{\ensuremathsf{y}}
\newrobustcmd{\varz}{\ensuremathsf{z}}
\newrobustcmd{\nonA}{\ensuremathsf{A}}
\newrobustcmd{\nonS}{\ensuremathsf{S}}

\newmathcommand\applyto{\kl[\applyto]{\textcolor{white}{{\cdot}}}}
 \knowledge\applyto{notion}

\newrobustcmd{\mar}{\kl[\mar]{\mathit{mar}}}
\knowledge\mar{math notion}
\newcommand{\subst}[3]{{#1[#2 / #3]}}
\newrobustcmd{\reify}[1]{\kl[\reify]{\ensuremath{#1^\bullet}}}
\knowledge\reify{math notion}
\newcommand{\reducesto}{\to_\beta}
\newcommand{\reducesplus}{\to_{\beta^+}}
\newcommand{\headreducesto}{\to_{\mathsf{h}\beta}}
\newcommand{\limbp}[1]{\mathsf{BT}^+(#1)}
\newrobustcmd{\NT}[1]{\kl[\NT]{\Ll({#1})}}
\knowledge\NT{math notion}

\newrobustcmd{\ndd}{\kl[\ndd]{\ensuremath{\mathsf{nd}}}}
\knowledge\ndd{math notion}
\newrobustcmd{\su}{\mathsf{succ}}
\newrobustcmd{\size}[1]{\lvert #1\rvert}
\newrobustcmd{\successor}[1]{\to_{#1}}
\newrobustcmd{\tsuccessor}[1]{\to^*_{#1}}
\newrobustcmd{\simul}{\sqsupseteq}   
\newrobustcmd{\tsimul}{\ensuremath{\mathrel{\kl[\tsimul]{\sqsupseteq}}}}
\knowledge\tsimul{math notion}

\newrobustcmd{\parent}{\kl[\parent]{\mathit{par}}}
\knowledge\parent{math notion}
\newrobustcmd{\child}{\kl[\child]{\mathit{ch}}}
\knowledge\child{math notion}
\newrobustcmd{\Dirs}{\kl[\Dirs]{\mathit{Dirs}}}
\knowledge\Dirs{math notion}

\newrobustcmd{\val}{\kl[\val]{\mathit{val}}}
\knowledge\val{math notion}
\newrobustcmd{\inc}{\kl[\inc]{\mathtt{i}}}
\knowledge\inc{math notion}
\newrobustcmd{\reset}{\kl[\reset]{\mathtt{r}}}
\knowledge\reset{math notion}
\newrobustcmd{\e}{\kl[\e]{\varepsilon}}
\knowledge\e{math notion}
\newrobustcmd{\muN}{\kl[\muN]{\mu}^{\kl[\muN]N}}
\knowledge\muN{math notion}
\newrobustcmd{\dup}{\kl[\dup]{\uparrow}}
\knowledge\dup{math notion}
\newrobustcmd{\ddown}{\kl[\ddown]{\downarrow}}
\knowledge\ddown{math notion}
\newrobustcmd{\de}{\kl[\de]{\circlearrowleft}}
\knowledge\de{math notion}
\newrobustcmd{\act}{c}

\newrobustcmd{\Rup}{\kl[\Rup]{\Uparrow}}
\knowledge\Rup{math notion}
\newrobustcmd{\Rdown}{\kl[\Rdown]{\Downarrow}}
\knowledge\Rdown{math notion}

\newrobustcmd{\HORS}{\kl{HORS}\xspace}
\knowledge{HORS}{notion}

\newcommand{\tree}[1]{\mathsf{tree}(#1)}

\newrobustcmd{\gameon}[2]{(#1, #2)}

\let\LaTeXlambda\lambda
\renewrobustcmd\lambda{\kl[\lambda]{\LaTeXlambda}}
\knowledge{\lambda}{math notion}
\newrobustcmd\exptype[1]{^{\kl[\exptype]{#1}}}
\knowledge\exptype{math notion}

\LoopCommands{\Gamma\Lambda\delta}[#1][LaTeX#1]
    {\renewmathcommand#2{\kl[#2]{#3}}
    }
\LoopCommands{\Gamma\Lambda\delta}[#1]
    {\knowledge#2{math notion}}

\LoopCommands{}[#1]
    {\declaremathcommand#2{\kl[#2]{#1}}
     \knowledge#2{math notion}}

\newif\ifstartedinmathmode
\newcommand*{\st}{\relax\ifmmode\startedinmathmodetrue\else\startedinmathmodefalse\fi\ifstartedinmathmode{\,.\,}\else{\protect{s.t.\xspace~}}\fi}

\makeatletter
\def\itemizename{itemize}
\def\enumeratename{enumerate}
\def\ifinlist{%
  \ifx\@currenvir\itemizename
    \expandafter\@firstoftwo
  \else
    \ifx\@currenvir\enumeratename
      \expandafter\@firstoftwo
    \else
      \expandafter\@secondoftwo
    \fi
  \fi}
\makeatother

\renewcommand{\leq}{\leqslant}
\renewcommand{\geq}{\geqslant}

\newcommand{\cl}{\mathit{cl}}

\newenvironment{proofof}[1]
  {\trivlist\PRstyle\item[]{\bfseries Proof of #1:}\newline}{\QED\endtrivlist}

\knowledgeconfigure{diagnose line=true}

\newcommand\synonym[1]{\knowledge{#1}{synonym}}

\knowledge{ranked alphabet}{notion}
  \synonym{alphabet}
\knowledge{rank}{notion}
  \synonym{ranks}
\knowledge{type}{notion}
  \synonym{types}

\begin{scope}\label{type}
   \knowledge{bad}{notion}
   \knowledge{order}{notion}
   \knowledge{homogeneous}{notion}
     \synonym{homogeneity}
\end{scope}
\knowledge{homogeneous type}{link=homogeneous,link scope=type}
\knowledge{homogeneous types}{link=homogeneous,link scope=type}

\knowledge {lambda-term}{notion}
   \synonym {lambda-terms}

\knowledge{lambda-term of type}{notion}
  \synonym{lambda-terms of type}
  \synonym{of type}
  \synonym{has type}
  \synonym{type of}
  \synonym{types of}

\knowledge {alpha-conversion}{notion}

\knowledge {beta-reduction}{notion}
  \synonym{beta-reductions}
  \synonym{beta-reduces}
  \synonym{reduction}

\knowledge {superficially safe}{notion}

\knowledge {lambda-binder}{notion}
  \synonym {lambda-binders}
  
\knowledge {subterm}{notion}
  \synonym {subterms}

\knowledge{application}{notion}
  \synonym{applications}
  \synonym{application subterm}

\knowledge{infinite application}{notion}

\knowledge {applicative term}{notion}

\knowledge {substitution}{notion}
  \synonym{substituted}
  \synonym{substituting}
  \synonym{substitutions}
  \synonym{capture-avoiding substitution}
  \synonym{capture-avoiding substitutions}

\knowledge {closed}{notion}

\begin{scope}\label{lambda}
   \knowledge {order}{notion}
   \knowledge {variable}[variables|Variables]{notion}
   \knowledge {free variable}[free|free variables|Free variables]{notion}
   \knowledge {represented}{notion}
   \knowledge {safe}{notion}
     \synonym{safety}
   \knowledge{normalizing}{notion}
\end{scope}

\knowledge{lambda-tree}{notion}
  \synonym{lambda-trees}
\knowledge{first-order lambda-term}{notion}
  \synonym{first-order}

\begin{scope}\label{scheme}
   \knowledge {safe}{notion}
     \synonym{safety}
   \knowledge {order}{notion}
   \knowledge{normalizing}{notion}
   \knowledge{recognize}{notion}
     \synonym{recognizer}
     \synonym{recognizes}
     \synonym{recognizers}
     \synonym{recognized}
     \synonym{language recognized}
     \synonym{languages recognized}
   \knowledge{homogeneous}{notion}
     \synonym{homogeneity}
\end{scope}

\knowledge {safe \HORS}{link=safe, link scope = scheme}
  \synonym{safe recursion scheme}
  \synonym {safe recursion schemes}
  \synonym {Safe recursion schemes}

\knowledge{variable of type}{link=variable,link scope=lambda}
\knowledge{variables of type}{link=variable,link scope=lambda}
\knowledge{variable of type~$\alpha$}{link=variable,link scope=lambda}
\knowledge{variables of type~$\alpha$}{link=variable,link scope=lambda}
\knowledge{variable}{scope=scheme, link=variable,link scope=lambda}
\knowledge{variables}{scope=scheme, link=variable,link scope=lambda}
    
\knowledge {(higher-order) recursion scheme}{notion}
  \synonym {recursion scheme}
  \synonym {recursion schemes}
  \synonym {Recursion schemes}
  \synonym {(higher-order, deterministic) recursion scheme}
  \synonym {(unsafe) recursion schemes}

\knowledge {nonterminal}{notion}
  \synonym {nonterminals}
  \synonym {Nonterminals}
 
\knowledge {initial nonterminal}{notion}

\knowledge {rule}{notion}
  \synonym{rules}

\knowledge{downward closure}{notion}
  \synonym{downward closures}
  \synonym{Downward closures}
  \synonym{downward-closed}
  \synonym{downward closed}

\knowledge{$\Sigma$-diagonal}{notion}
  \synonym{$\Sigma$-diagonal language}

\knowledge{homeomorphically embedded}{notion}
  \synonym{$S$ homeomorphically embeds into $T$}
  \synonym{homeomorphically embeds}
  \synonym{embedding well-quasi order}
  \synonym{embed}
  \synonym{embeds}
  \synonym{embedded}
  
\knowledge{$n$-large with respect to $\Sigma$}{notion}

\knowledge{MSO-reflection property}{notion}
  \synonym{MSO-reflection}

\knowledge{simple tree regular expression}{notion}
  \synonym{simple tree regular expressions}
  \synonym{Simple tree regular expressions}
  
\knowledge{tree}[trees|Trees]{notion}

\knowledge{context}{notion}
  \synonym{contexts}
  \synonym{Contexts}

\knowledge{hole}[holes|Holes]{notion}

\knowledge{node}{notion}
  \synonym{nodes}
\knowledge{leaf}{notion}
  \synonym{leaves}
\knowledge{root}{notion}
  \synonym{roots}
\knowledge{child}{notion}
  \synonym{children}
\knowledge{parent}{notion}
\knowledge{branch}{notion}
  \synonym{branches}
\knowledge{subtree}{notion}
  \synonym{subtrees}

\knowledge {regular tree}{notion}
  \synonym{regular}

\knowledge {B\"ohm tree}{notion}
  \synonym{B\"ohm tree generated}
  \synonym{Böhm tree}
  \synonym{Böhm trees}

\knowledge{tree generated}{notion}
  \synonym{tree generated by}
  \synonym{trees generated by}
  \synonym{generated}
  \synonym{generates}
  \synonym{generated by}
  \synonym{generating}
  
\knowledge{versatile tree}{notion}
  \synonym{versatile trees}
  \synonym{Versatile trees}
   
\begin{scope}\label{logic}
   \knowledge{free variable}[free|free variables|Free variables]{notion}
   \knowledge{accept}[accepts|accepted]{notion}
   \knowledge{$n$-accepts}[$n$-accept|$n$-accepting]{notion}
\end{scope}

\knowledge{\size X <\varN}{notion}
\knowledge{\varN}{notion}

\knowledge {MSO}{notion}

\knowledge{monadic variable}{notion}
  \synonym{monadic variables}
  \synonym{monadic second-order variable}
  \synonym{monadic second-order variables}
  \synonym{set variable}
  \synonym{set variables}

\knowledge{sentence}{notion}
  \synonym{\mso sentence}
  
\knowledge{first-order variable}{notion}
  \synonym{first-order variables}
  \synonym{First-order variables}

\knowledge {CMSO}{notion}

\knowledge {cost monadic logic}[Cost monadic logic]{notion}

\knowledge {WCMSO}{notion}
   \synonym{weak cost monadic logic}
   \synonym{weak cost monadic second-order logic}
   \synonym{Weak cost monadic logic}
   
\knowledge {QWCMSO}{notion}

\knowledge {quasi-weak cost monadic logic}{notion}

\knowledge {numeric variable}{notion}

\knowledge{product}[products|Products]{notion}

\knowledge{pre-product}{notion}
  \synonym{pre-products}

\knowledge{pure product}[pure products|Pure products]{notion}

\knowledge{diversified pure product}{notion}
  \synonym{diversified pure products}
  \synonym{Diversified pure products}
  \synonym{diversified}
  
\knowledge{$n$-large with respect to $P$}{notion}
  \synonym{$n$-large}

\knowledge{diagonal problem}[Diagonal problems|diagonal problems]{notion}
   \synonym{diagonal problem for trees}
\knowledge{simultaneous unboundedness problem}{notion}
\knowledge{universality problem}{notion}
\knowledge{boundedness problem}{notion}
\knowledge{model-checking problem}{notion}

\knowledge {B-automaton}{notion}
  \synonym {B-automata}
  \synonym {alternating B-automata}
  \synonym {alternating B-automaton}
  \synonym {(alternating, two-way) B-automaton}
  \synonym {two-way B-automaton}
  \synonym {two-way automata}
  \synonym {alternating two-way B-automaton}

\knowledge{weak B-automaton}{notion}
  \synonym{weak B-automata}

\knowledge{quasi-weak B-automaton}{notion}
  \synonym{quasi-weak B-automata}
  
\knowledge {one-way}{notion}
  \synonym {alternating one-way B-automaton}
  \synonym {one-way B-automaton}
  \synonym {one-way B-automata}
  \synonym {one-way automata}
  
\knowledge{state}[states|States]{notion}

\knowledge{initial state}{notion}

\knowledge{incremented}{notion}
\knowledge{reset}{notion}
\knowledge{unchanged}{notion}
\knowledge{counter action}{notion}
  \synonym{counter actions}
\knowledge{counter}[counters|Counters]{notion}

\knowledge{cost automata}{notion}

\knowledge{value of a sequence}{notion}
  \synonym{value}
  \synonym{counter values}

\begin{scope}\label{aut}
  \knowledge{$n$-accepted}{notion}
    \synonym{$n$-accepts}
    \synonym{$\alpha(n)$-accepts}
    \synonym{$n$-accepting}
  \knowledge{accepts}{notion}
    \synonym{accepted}
    \synonym{acceptance}
  \knowledge{recognize}[recognized|recognizes|recognizing]{notion}
  \knowledge{transition}[transition function|transitions|Transitions]{notion}
\end{scope}

\knowledge{priority}[priorities|Priorities]{notion}

\begin{scope}\label{ftt}
  \knowledge{control state}[control states|state]{notion}
  \knowledge{initial state}{notion}
  \knowledge{transition rule}[transition rules|rule|rules|Rules]{notion}
  \knowledge{linear}{notion}
\end{scope}

\knowledge{finite tree transducer}[finite tree transucers|Finite tree transducers|FTT|linear FTT|transducer]{notion}

\knowledge{linear FTT transduction}{notion}
  \synonym{linear FTT transductions}
  \synonym{FTT transduction}
  \synonym{Linear FTT transductions}

\begin{scope}\label{game}
  \knowledge{$n$-wins}{notion}
    \synonym{$n$-win}
\end{scope}

\knowledge {play}{notion}

\knowledge {$n$-winning play}{notion}

\knowledge {$n$-winning strategy}{notion}
  \synonym {$n$-winning strategies}
  
\knowledge {strategy}{notion}

\knowledge {compatible}{notion}

\knowledge {$n$-wins the game}{notion}

\begin{scope}\label{lambda-tree}
  \knowledge{normalizing}{notion}
\end{scope}
\knowledge{$(\Xx,s)$-successor}{notion}
  \synonym{$(\Xx,s)$-successors}
  \synonym{$(\Xx,2)$-successor}
  \synonym{successor}
  \synonym{successors}
  \synonym{successor steps}
\knowledge{$(\Xx,s)$-maximal path}{notion}
  \synonym{$(\Xx,s)$-maximal paths}
  \synonym{$(\Xx,2)$-maximal path}
  
\knowledge{$(\Xx,s)$-derived tree}{notion}
  \synonym{$(\Xx,2)$-derived tree}
  \synonym{derived tree}
  \synonym{derived trees}

\knowledge{iterator}{notion}
  \synonym{iterators}
  \synonym{Iterators}
  
\knowledge{linear iterator}{notion}
  \synonym{linear iterators}
  \synonym{Linear iterators}
  
\knowledge{linear context}{notion}
  \synonym{linear contexts}
  \synonym{Linear contexts}
  
\knowledge{full iterator}{notion}

\knowledge{full context}{notion}
  \synonym{full contexts}
  \synonym{Full contexts}
  
\knowledge{equivalent STREs}{notion}

\knowledge{cut}{notion}
  \synonym{cuts}
\knowledge{cut variable}{notion}
  \synonym{cut variables}
\knowledge{order-0 cut}{notion}
  \synonym{order-0 cuts}
\knowledge{bad}{notion}
  \synonym{bad type}
\knowledge{maximal arity}{notion}
\knowledge{input lambda-term}{notion}
  \synonym{input lambda-terms}
\knowledge{reification}{notion}
  \synonym{reification operation}
  \synonym{reified}
  \synonym{Reification}
\knowledge{represents}{notion}
\knowledge{beta-reductions of positive order}{notion}
\knowledge{limit}{notion}
\knowledge{head beta-reduces}{notion}
  \synonym{head beta-reductions}
\knowledge{configuration}{notion}
  \synonym{configurations}
\knowledge{valid}{notion}
  \synonym{valid configuration}
  \synonym{valid configurations}
\knowledge{weak simulation}{notion}
\knowledge{agrees}{notion}
\knowledge{letter}{notion}
  \synonym{letters}

\begin{document}

\setcounter{page}{127}
\publyear{22}
\papernumber{2145}
\volume{188}
\issue{3}


    \finalVersionForARXIV

\title{Cost Automata, Safe Schemes, and Downward Closures}

\author{David Barozzini\thanks{Author supported by the National Science Centre, Poland (grant no.\ 2016/22/E/ST6/00041).}\\
	Institute of Informatics \\
	University of Warsaw \\ Warsaw, Poland\\
	dbarozzini{@}mimuw.edu.pl
\and	Lorenzo Clemente\thanksas{1}\\
	Institute of Informatics \\
	University of Warsaw \\ Warsaw, Poland\\
	clementelorenzo@gmail.com
\and	Thomas Colcombet\thanks{Author supported by the European Research Council (ERC) under the European
                 Union’s Horizon 2020 research and innovation programme (grant agreement No.670624),
		         and the DeLTA ANR project (ANR-16-CE40-0007).}\\
	IRIF-CNRS-Université de Paris\\
	Paris, France\\
	thomas.colcombet@irif.fr
\and	Paweł Parys\thanksas{1}\thanks{Address for correspondence: Institute of Informatics, University of Warsaw,
                        Warsaw,  Poland.  \newline \newline
                    \vspace*{-6mm}{\scriptsize{Received September 2021; \ accepted March  2023.}}}
    \\
	Institute of Informatics \\
	University of Warsaw \\ Warsaw, Poland\\
	parys{@}mimuw.edu.pl}

\maketitle

\runninghead{D. Barozzini et al.}{Cost Automata, Safe Schemes, and Downward Closures}

\vspace{-5ex}
\begin{abstract}
	In this work we prove decidability of the model-checking problem for safe recursion schemes against properties defined by alternating B-automata.
	We then exploit this result to show how to compute downward closures of languages of finite trees recognized by safe recursion schemes.

	Higher-order recursion schemes are an expressive formalism used to define languages of finite and infinite ranked trees by means of fixed points of lambda terms.
	They extend regular and context-free grammars, and are equivalent in expressive power to the simply typed $\lambda Y$-calculus and collapsible pushdown automata.
	Safety in a syntactic restriction which limits their expressive power.

	The class of alternating B-automata is an extension of alternating parity automata over infinite trees;
	it enhances them with counting features that can be used to describe boundedness properties.
\end{abstract}



\section{Introduction}

	Higher-order functions are nowadays widely used not only in functional programming languages such as Haskell
	and the OCAML family,
	but also in mainstream languages such as Java, JavaScript, Python, and C++.
	\emph{\kl{Recursion schemes}} are faithful and algorithmically manageable abstractions of the control flow of higher-order programs~\cite{KobayashiPrograms}.
	A deterministic recursion scheme normalizes into a possibly infinite B\"ohm tree,
	and in this respect recursion schemes can equivalently be presented as simply-typed lambda-terms
	using a higher-order fixpoint combinator $Y$~\cite{schemes-lY}.
	There are also (nontrivial) inter-reductions between recursion schemes and the equi-expressive
	formalisms of collapsible higher-order pushdown automata~\cite{collapsible-equivalence}
	and ordered tree-pushdown automata~\cite{Ordered-Tree-Pushdown}.
	In another semantics, also used in this paper,
	nondeterminstic recursion schemes are recognizers of languages of finite trees,
	and in this view they are also known as higher-order OI grammars~\cite{Damm82,Kobele:Salvati:IC:2015},
	generalising indexed grammars~\cite{AhoIndexed} (which are recursion schemes of order two)
	and ordered multi-pushdown automata~\cite{OrderedMultiPushdown}.
	
	The most celebrated algorithmic result in the analysis of recursion schemes
	is decidability of the model-checking problem against properties expressed in monadic second-order logic (\mso):
	given a recursion scheme $\Gg$ and an \kl{\mso sentence} $\varphi$,
	one can decide whether the B\"ohm tree generated by $\Gg$ satisfies $\varphi$~\cite{Ong:LICS:2006}.
	This fundamental result has been reproved several times,
	that is, using collapsible higher-order pushdown automata~\cite{collapsible},
	intersection types~\cite{KobayashiOngtypes},
	Krivine machines~\cite{KrivineWS},
	order-reducing transformations~\cite{order-red},
	and it has been extended in diverse directions such as
	global model checking~\cite{globalMC},
	logical reflection~\cite{reflection},
	effective selection~\cite{selection},
	and a transfer theorem via models of lambda-calculus~\cite{ModelSM}.
	When the input property is given as an \mso formula,
	the model-checking problem is non-elementary already for trees of order $0$ (regular trees)~\cite{Stockmeyer:PhD:1974};
	when the input property is presented as a parity tree automaton
	(which is equi-expressive with \mso on trees, but less succinct),
	the \mso model-checking problem for recursion schemes of order $n$
	is complete for $n$-fold exponential time~\cite{Ong:LICS:2006}.
	Despite these hardness results,
	the model-checking problem can be solved efficiently on multiple nontrivial examples,
	thanks to the development of several recursion-scheme model checkers~%
	\cite{KobayashiPrograms, HorSat, GTRecS, TravMC2, PrefaceTool}.

\paragraph{Unboundedness problems I: Diagonal problem and downward closures.}
	Recently, an increasing interest has arisen for model checking quantitative properties going beyond the expressive power of \mso.
	The \emph{\kl{diagonal problem}} is an example of a quantitative property not expressible in \mso.
	Over words, the problem asks, for a given set of letters $\Sigma$ and a language of finite words $\Ll$,
	whether for every $n \in \Nat$ there is a word in $\Ll$ where every letter from $\Sigma$ occurs at least $n$ times.
	The \kl{diagonal problem} for languages of finite words recognized by recursion schemes 
	is decidable~%
	\cite{ClementeParysSalvatiWalukiewicz:LICS:2016, diagonal-safe, Parys:FSTTCS:2017}.

	The class of languages of finite words recognized by \kl{recursion schemes} form a so-called \emph{full trio} (i.e., it is closed under regular transductions)
	and for full trios decidability of the \kl{diagonal problem} has interesting algorithmic consequences,
	such as computability of downward closures~\cite{Zetzsche:ICALP:2015,CzerwinskiMartensRooijenZeitounZetzsche:DMTCS:2017}
	and decidability of separability by piecewise testable languages~\cite{sep-piecewise-test}.
	
	The problem of computing downward closures is an important problem in its own right.
	The \emph{\kl{downward closure}} of a language $\Ll$ of finite trees
	is the set $\downwardclosure \Ll$ of all trees that can be \kl{homeomorphically embedded} into some tree in $\Ll$.
	By Higman's lemma~\cite{Higman}, the embedding relation on finite ranked trees is a well quasi-order.
	Consequently, the downward closure $\downwardclosure \Ll$ of an arbitrary set of trees $\Ll$ is always a regular language.
	The downward closure of a language offers a nontrivial regular abstraction thereof:
	even though the actual count of letters is lost,
	their limit properties are preserved,
	as well as their order of appearance.
	We say that the \kl{downward closure} is \emph{computable}
	when a finite automaton for $\downwardclosure\Ll$ can be effectively constructed (which is not true in general).
	Downward closures are computable for a wide class of languages of finite words
	such as those recognized by
	context-free grammars~\cite{BachmeierLuttenbergerSchlund:2015,Courcelle:1991,vanLeeuwen:1978},
	Petri nets~\cite{DownwardPN:ICALP:2010},
	stacked counter automata~\cite{Zetzsche:STACS:2015},
	context-free FIFO rewriting systems and 0L-systems~\cite{AbdullaBoassonBouajjani:ICALP:2001},
	second-order pushdown automata~\cite{Zetzsche:ICALP:2015},
	higher-order pushdown automata~\cite{diagonal-safe},
	and (possibly unsafe) recursion schemes over words~\cite{ClementeParysSalvatiWalukiewicz:LICS:2016}.
	Over finite trees, it is known that downward closures are computable
	for the class of regular tree languages \cite{Goubault-Larrecq:Schmitz:ICALP:2016}.
	We are not aware of such computability results for other classes of languages of finite trees.

\paragraph{Unboundedness problems II: B-automata.}

	In another line of research,
	\kl{B-automata}, and among them \emph{\kl{alternating B-automata}}, have been put forward as a quantitative extension to
	\mso~\cite{DBLP:conf/lics/BojanczykC06,DBLP:conf/icalp/Colcombet09,reg-cost-finite-trees,reg-cost-infinite-words,weak-cost,quasi-weak-automata}.
	They extend alternating automata over infinite trees~\cite[Chapter 9]{AutomataLogicsInfiniteGames:2002}
	by nonnegative integer counters that can be incremented or reset to zero.
	The extra counters do not constrain the availability of transitions during a run
	(unlike in other superficially similar models, such as counter machines),
	but are used in order to define the acceptance condition:
	an infinite tree is \emph{\kl(aut){$n$-accepted}} if~$n$ is a bound
	on the values taken by the counters during an accepting run of the automaton over it.

\AP	The \intro{universality problem} consists in deciding whether for every tree there is a bound~$n$ for which it is \kl(aut){$n$-accepted}.
\AP	The \intro{boundedness problem} asks whether there exists a bound~$n$ for which all trees are \kl(aut){$n$-accepted}.
	These two problems are closely related. Their decidability is an important open problem in the field,
	and proving decidability of the \kl{boundedness problem}
	would solve the long standing nondeterministic Mostowski index problem~\cite{ColcombetLoding:ICALP:2008}.
	However, though open in general, the \kl{boundedness problem} is known to be decidable over finite words~\cite{DBLP:conf/icalp/Colcombet09},
	finite trees~\cite{reg-cost-finite-trees},
	infinite words~\cite{reg-cost-infinite-words},
	as well as over infinite trees for its weak~\cite{weak-cost} and the more general quasi-weak~\cite{quasi-weak-automata} variant.

\AP	Another expressive formalism for unboundedness properties beyond \mso is \msou,
	which extends \mso by a novel quantifier ``$\unbound X.\varphi$''~\cite{BojanczykU}
	stating that there exist arbitrarily large finite sets $X$ satisfying $\varphi$.
	This logic is incomparable with B-automata.
	The model-checking problem of recursion schemes against its weak fragment \wmsou,
	where monadic second-order quantifiers are restricted to finite sets, is decidable~\cite{wmsou-schemes}.

\paragraph{Contributions.}

	Our first contribution is decidability of the \kl{model-checking problem} of properties expressed by \kl{alternating B-automata}
	for an expressive class of recursion schemes called \emph{\kl{safe recursion schemes}}.
	As generators of infinite trees,
	\kl{safe recursion schemes} are equivalent to higher-order pushdown automata without the collapse operation~\cite{easy-trees}
	and are strictly less expressive than general \kl{(unsafe) recursion schemes}~\cite[Theorem 1.1]{ho-new}.
	Here, the \kl{model-checking problem} asks whether a concrete infinite tree
	(the \kl{B\"ohm tree generated} by a \kl{safe recursion scheme}) is accepted by the \kl{B-automaton} for some bound.
	This problem happens to be significantly simpler than the universality/boundedness problems described above.
	The proof goes by reducing the \kl(scheme){order} of the \kl{safe recursion scheme}
	similarly as done by Knapik, Niwiński, and Urzyczyn~\cite{easy-trees} to show decidability of the \mso model-checking problem,
	at the expense of making the property automaton two-way.
	We then rely on the fact that two-way alternating B-automata
	can effectively be converted to equivalent one-way alternating B-automata~\cite{quasi-weak-logic}.
	Our result is incomparable with the seminal decidability result of Ong~\cite{Ong:LICS:2006},
	since \begin{inparaenum}[(1)]\item alternating B-automata are strictly more expressive than \mso,
	however \item we obtain it under the more restrictive safety assumption. \end{inparaenum}
	Whether the safety assumption can be dropped while preserving decidability of the model-checking problem
	against \kl{B-automata} properties,
	thus strictly extending Ong's result to the more general setting of boundedness properties, remains open.
	
	Our second contribution is to define the following generalization of the diagonal problem from words to trees:
	given a language of finite trees $\Ll$ and a set of letters $\Sigma$,
	decide whether for every $n \in \Nat$ there is a tree $T\in\Ll$ such that every letter from $\Sigma$ occurs at least $n$ times on every branch of $T$.
	This generalization is designed in order to reduce computation of
	\kl{downward closures} to the \kl{diagonal problem}, in the same fashion as for finite words.
	Our proof strategy is to represent downward-closed sets of trees $\downwardclosure \Ll$ by \kl{simple tree regular expressions},
	which are a subclass of regular expressions for finite trees~%
	\cite{Goubault-Larrecq:Schmitz:ICALP:2016,FinkelGoubaultLarrecq:Completions:STACS:2009}.
	By further analysing and simplifying the structure of these expressions,
	computation of the downward closure can be reduced to finitely many instances of the diagonal problem.
	Unlike in the case of finite words,
	we do not know whether for full trios of finite trees there exists a converse reduction
	from the diagonal problem to the problem of computing downward closures.

	Our third contribution is decidability of the diagonal problem for languages of finite trees recognized by safe recursion schemes
	(and thus computability of downward closures of those languages).
	The diagonal problem can directly be expressed in a logic called \emph{\kl{weak cost monadic second-order logic}} (\kl{WCMSO})~\cite{weak-cost},
	which extends weak \mso with atomic formulas of the form $\size X < \varN$
	stating that the cardinality of the \kl{monadic variable} $X$ is smaller than $\varN$.
	Since \kl{WCMSO} can be translated to \kl{alternating B-automata}~\cite{weak-cost},
	the diagonal problem reduces to the \kl{model-checking problem} of \kl{safe recursion schemes} against \kl{alternating B-automata},
	which we have shown decidable in the first part.
	Note that it seems difficult to express the \kl{diagonal problem} using \kl{alternating B-automata} directly,
	and indeed the fact that alternating B-automata can express all \kl{WCMSO} properties is nontrivial.
	It is worth stressing that this connection between these two unboundedness problems (the diagonal problem and model-checking of B-automata) is new and has not been observed before.
	
	This paper is based on a conference paper~\cite{conference-version}, showing the same results; we add here missing proofs and some examples.

\paragraph{Outline.}

	In \cref{sec:preliminaries}, we define \kl{recursion schemes} and \kl{B-automata}.
	In \cref{sec:model-checking}, we present our first result,
	namely decidability of model checking of safe recursion schemes against B-automata.
	In \cref{sec:ideals}, we introduce the \kl{diagonal problem}, and we show how it can be used to compute \kl{downward closures}.
	In \cref{sec:downward-closure}, we solve the diagonal problem for safe recursion schemes.
	We conclude in \cref{sec:conclusions} with some open problems.


\section{Preliminaries}
\label{sec:preliminaries}

\paragraph{Recursion schemes.}

\AP	A \intro{ranked alphabet} is a (usually finite) set $\Alphabet$ of \intro{letters},
	together with a function $\intro*\rank\colon\Alphabet\to\Nat$, 
	assigning a \intro{rank} to every \kl{letter}.
	When we define trees below, we require that a \kl{node} labeled by a \kl{letter} $a$ has exactly $\rank(a)$ children.
	In the sequel, we usually assume some fixed finite \kl{ranked alphabet} $\intro*\Alphabet$ that contains a distinguished letter $\bot$ of rank $0$.

\AP	The set of \intro[type]{(simple) types} is constructed
	from a unique ground type $\intro*\otyp$ using a binary operation $\intro*\arr$; 
	namely $\otyp$ is a type, and if $\alpha$ and $\beta$ are types, so is $\alpha\arr\beta$.
	By convention, $\arr$ associates to the right, that is, $\alpha\arr\beta\arr\gamma$ is understood as $\alpha\arr(\beta\arr\gamma)$.
	A \kl{type} $\otyp\arr\dots\arr\otyp$ with $k$ occurrences of $\arr$ is also written as $\otyp^k\arr o$.
\AP	The \intro(type){order} of a \kl{type} $\alpha$, denoted $\intro*\ord(\alpha)$ is defined by induction: 
	$\ord(\otyp)=0$ and $\ord(\alpha_1\arr\dots\arr\alpha_k\arr\otyp)=\max_i(\ord(\alpha_i))+1$ for $k\geq 1$.
	
\AP	We coinductively define both \intro{lambda-terms} and a two-argument relation ``$M$ is a \intro{lambda-term of type} $\alpha$'' as follows:%
	\footnote{Cf.~the works~\cite{Berarducci-Dezani:TCS:1999,KennawayKlopSleepDeVries:TCS:1997}
	for analogous definitions in the literature on infinite lambda calculus.
	Note that we use \kl{letters} (constants) from a \kl{ranked alphabet}, which is a minor modification that suits our needs.}
	\begin{compactitem}
	\itemAP	a \kl{letter} $a\in\Alphabet$ is a \kl{lambda-term of type} $\otyp^{\rank(a)}\arr\otyp$;
	\itemAP	for every \kl{type} $\alpha$ there is a countable set $\{x,y,\dots\}$ of \intro{variables of type~$\alpha$} which can be used as \kl{lambda-terms of type} $\alpha$;
	\itemAP	if $M$ is a \kl{lambda-term of type} $\beta$ and $x$ a \kl{variable of type} $\alpha$, then $\intro*\lambda x.M$
	is a \kl{lambda-term of type} $\alpha\arr\beta$; this construction is called a \intro{lambda-binder};
	\itemAP	if $M$ is a \kl{lambda-term of type} $\alpha\arr\beta$, and $N$ is a \kl{lambda-term of type} $\alpha$, then $M\intro*\applyto N$ is a \kl{lambda-term of type} $\beta$, called an \intro{application}.
	\end{compactitem}
	Note that this definition is coinductive, meaning that lambda-terms may be infinite.
\AP	As usual, we identify \kl{lambda-terms} up to
	\intro{alpha-conversion} (i.e., renaming of bound \kl(lambda){variables}).
	Notice that, according to our definition, every \kl{lambda-term}
	(and in particular every \kl(lambda){variable}) has a particular \kl{type} associated with it.
	We use here the standard notions of \intro(lambda){free variable}, \intro{subterm},
	(capture-avoiding) \intro{substitution}, and \intro{beta-reduction}.
\AP	A \intro{closed} \kl{lambda-term} does not have \kl(lambda){free variables}.
	For a \kl{lambda-term} $M$ \kl{of type} $\alpha$, the \intro(lambda){order} of $M$, denoted $\ord(M)$, is defined as $\ord(\alpha)$.
\AP	A \kl{lambda-term} $M$ is a \intro{first-order lambda-term}\label{page:first-order}
	if every \kl{subterm} of $M$ (including $M$ itself) has \kl(lambda){order} at most $1$
	and every \kl(lambda){free variable} of $M$ has \kl(lambda){order} $0$.
\AP	An \intro{applicative term} is a \kl{lambda-term} not containing \kl{lambda-binders} (it contains only \kl{letters}, \kl{applications}, and \kl(lambda){variables}).

\AP	A \kl{lambda-term} $M$ is \intro{superficially safe} if all \kl(lambda){free variables} $x$ thereof satisfy $\ord(x) \geq \ord(M)$.
	A \kl{lambda-term} $M$ is \intro(lambda){safe} if
	for every \kl{subterm} thereof of the form $K\applyto L$
	(i.e., an \kl{application}),
	the \kl{subterm} $L$ is \kl{superficially safe}.\footnote{%
		Some definitions of safe lambda-terms add the following requirement:
		if $K\applyto L$ is a \kl{subterm} of $M$, and $K$ is not an \kl{application}, then also $K$ is required to be \kl{superficially safe}~\cite{schemes-lY,safe}.
		This does not change anything when it comes to \kl(scheme){safety} of $\Lambda(\Gg)$ for a \kl{recursion scheme} $\Gg$:
		if $K\applyto L$ is a \kl{subterm} of $\Lambda(\Gg)$, and $K$ is not an \kl{application}, then $K$ is either \kl{closed} or a \kl(lambda){variable},
		so it is always \kl{superficially safe}.}
	For example, if $\leta,\varx,\varx'$ are \kl{of type} $\otyp$ and $\vary,\vary'$ are \kl{of type} $\otyp\arr\otyp$,
	then the \kl{lambda-term} $(\lambda\vary.\leta)\applyto(\lambda\varx.\vary'\applyto\leta)$ is \kl(lambda){safe},
	but the \kl{lambda-term} $(\lambda\vary.\leta)\applyto(\lambda\varx.\varx')$ is not \kl(lambda){safe}:
	$\varx'$ is an \kl(lambda){order}-$0$ \kl(lambda){free variable} in the \kl(lambda){order}-$1$ \kl{subterm} $(\lambda\varx.\varx')$ located on the argument position of an \kl{application}.
	Intuitively, \kl(lambda){safety} is a syntactic restriction that guarantees that (under appropriate assumptions)
	there is no need to rename bound \kl(lambda){variables} when performing \kl{substitution},
	since variable capture is guaranteed not to happen for safe \kl{lambda-terms}.
	This simplifies the analysis of \kl{lambda-terms}, and allows constructions by induction on the \kl(scheme){order}, as done in Knapik et al.~\cite{easy-trees}.
	Safe \kl{lambda-terms} are semantically less expressive than their unrestricted counterpart.

\AP	A \intro{(higher-order, deterministic) recursion scheme} over the alphabet $\Alphabet$ is a tuple $\Gg=\tuple{\Alphabet,\allowbreak\Nn,\allowbreak X_0,\allowbreak\Rr}$, 
	where $\Nn$ is a finite set of \kl[type]{typed} \intro{nonterminals},
	 $X_0\in\Nn$ is the \intro{initial nonterminal}, 
	and $\Rr$ is a function assigning to every \kl{nonterminal}
	$X\in\Nn$ \kl{of type} $\alpha_1\arr\cdots\arr\alpha_k\arr\otyp$ a finite \kl{lambda-term} of the form $\lambda x_1\lamdots\lambda x_k.K$, 
	of the same \kl[of type]{type} $\alpha_1\arr\cdots\arr\alpha_k\arr\otyp$,
	in which $K$ is an \kl{applicative term} with \kl(lambda){free variables} in $\Nn\uplus\set{x_1,\dots,x_k}$.
	We refer to $\Rr(X)$ as the \intro{rule} for $X$.
\AP	The \intro(scheme){order} of a \kl{recursion scheme} $\ord(\Gg)$ is the maximum \kl(type){order} of its \kl{nonterminals}.

\AP	The \kl{lambda-term} \intro(lambda){represented}
	by a \kl{recursion scheme} $\Gg$ as above, denoted $\intro*\Lambda(\Gg)$,
	is the limit of applying recursively the following operation to $X_0$:
	take an occurrence of some \kl{nonterminal} $X$, and replace it with $\Rr(X)$
	(the \kl{nonterminals} should be chosen in a fair way, so that every \kl{nonterminal} is eventually replaced).
	Thus, $\Lambda(\Gg)$ is a (usually infinite) regular \kl{lambda-term}
	obtained by unfolding the \kl{nonterminals} of $\Gg$ according to their definition.
	We remark that when substituting $\Rr(X)$ for a \kl{nonterminal} $X$
	there is no need for any renaming of \kl(lambda){variables} (\kl{capture-avoiding substitution}),
	since $\Rr(X)$ does not contain \kl(lambda){free variables} other than \kl{nonterminals}.
	We only consider \kl{recursion schemes} for which $\Lambda(\Gg)$ is well-defined (e.g.~by requiring that $\Rr(X)$ is not a single \kl{nonterminal}).
\AP	A \kl{recursion scheme} $\Gg$ is \intro(scheme){safe} if $\Rr(X)$ is safe for all nonterminals $X$. 
	When this is the case, then also the lambda-term $\Lambda(\Gg)$ is \kl(lambda){safe}.

\AP	A \intro{tree} is a \kl{closed} \kl{applicative term} \kl{of type} $\otyp$.
	Note that such a \kl[lambda-term]{term} is coinductively of the form $a\applyto M_1\applyto \cdots\applyto M_r$,
	where $a\in\Alphabet$ is of \kl{rank} $r$, and where $M_1,\dots,M_r$ are again \kl{trees}.
	Thus, a \kl{tree} defined this way can be identified with a tree understood in the traditional sense:
	$a$ is the label of its root, and $M_1, \dots, M_r$ are subtrees rooted at the $r$ children of the root, from left to right.
	For \kl{trees} we employ the usual notions of \intro{node}, \intro{root}, \intro{leaf}, \intro{child}, \intro{parent}, \intro{branch}, and \intro{subtree}.
\AP	A \kl{tree} is \intro{regular} if it has finitely many distinct \kl{subtrees}.

\AP	The \intro{B\"ohm tree} of a \kl{closed} \kl{lambda-term} $M$ \kl{of type} $\otyp$, denoted $\intro*\BT M$,
	is the \kl{tree} defined coinductively as follows:
	if there is a sequence of \kl{beta-reductions} from $M$ to a \kl{lambda-term} of the form $a\applyto M_1\applyto\dots\applyto M_r$
	(where $a\in\Alphabet$ is a \kl{letter}),
	then $\BT M=a\applyto(\BT{M_1})\applyto\dots\applyto(\BT{M_r})$;
	otherwise $\BT M=\bot$, where $\bot\in\Alphabet$ is a distinguished \kl{letter} of \kl{rank} $0$.
	It is a classical result that $\BT M$ exists,
	and is uniquely defined \cite{Berarducci-Dezani:TCS:1999,KennawayKlopSleepDeVries:TCS:1997}.
\AP	The \intro{tree generated by} a \kl{recursion scheme} $\Gg$,
	denoted $\BT \Gg$, is $\BT {\Lambda(\Gg)}$.

\AP	We say that a \kl{closed} \kl{lambda-term} $N$ \kl{of type} $\otyp$ is \intro(lambda){normalizing} if $\BT{N}$ does not contain the special \kl{letter} $\bot$;
\AP	a \kl{recursion scheme} $\Gg$ is \intro(scheme){normalizing} if $\Lambda(\Gg)$ is
	\kl(lambda){normalizing}.
	This notion is analogous to productivity in grammars:
	in a \kl(scheme){normalizing} \kl{recursion scheme} / \kl{lambda-term} the reduction process always terminates producing a new \kl{node}.
	It is possible to transform every \kl{recursion scheme} $\Gg$ into a \kl(scheme){normalizing} \kl{recursion scheme} $\Gg'$
	\kl{generating} the same \kl{tree} as $\Gg$, up to renaming $\bot$ into some non-special \kl{letter} $\bot\!'$ (cf.~\cite[Section~5]{haddad-fics}).\label{page:fully-convergent}
	Moreover, the construction preserves \kl(scheme){safety} and the \kl(scheme){order}.

	\begin{exampleOK}\label{ex:scheme}
		Consider the \kl{ranked alphabet} $\Alphabet$ containing two \kl{letters} $\leta,\ndd$ of \kl{rank} 2,
		two \kl{letters} $\letb_1, \letb_2$ of \kl{rank} 1, and two \kl{letters} $\bot,\letc$ of \kl{rank} $0$.
		Let $\Gg$ be the \kl{recursion scheme} consisting of an \kl{initial nonterminal} $\nonS$ of \kl(type){order}-0 \kl[of type]{type} $\otyp$
		and an additional \kl{nonterminal} $\nonA$ of \kl(type){order}-2 \kl[of type]{type}
		$(\otyp \arr \otyp) \arr (\otyp \arr \otyp) \arr \otyp \arr \otyp \arr \otyp$,
		together with the following two \kl{rules}:
		\begin{align*}
			\Rr(\nonS) &= \nonA\applyto \letb_1\applyto \letb_2\applyto \letc\applyto \letc, \\
			\Rr(\nonA) &= \lambda \varf.\lambda \varg.\lambda \varx.\lambda \vary.
				\ndd\applyto (\leta\applyto \varx\applyto \vary)\applyto (\nonA\applyto \varf\applyto \varg\applyto (\varf\applyto \varx)\applyto (\varg\applyto \vary)).
		\end{align*}
		Then, $\BT\Gg$ is the infinite non-\kl{regular} \kl{tree}
		\begin{align*}
			\ndd\applyto(\leta\applyto\letc\applyto\letc)\applyto(\ndd\applyto(\leta\applyto(\letb_1\applyto\letc)\applyto(\letb_2\applyto\letc))\applyto
			(\ndd\applyto(\leta\applyto(\letb_1\applyto(\letb_1\applyto\letc))\applyto(\letb_2\applyto(\letb_2\applyto\letc)))\applyto\dots)),
		\end{align*}
		depicted in \cref{fig:ex:scheme}.
	\end{exampleOK}

	\begin{figure*}
		\centering
		\begin{minipage}[b]{0.48\textwidth}
			\centering
			\def\svgscale{0.5}\import{pics/}{tree_1.pdf_tex_ok}\vspace{-1ex}
			\caption{The \kl{tree} $\BT\Gg$ (\cref{ex:scheme})}
			\label{fig:ex:scheme}
		\end{minipage}
		\begin{minipage}[b]{0.48\textwidth}
			\centering
			\def\svgscale{0.5}\import{pics/}{tree_2.pdf_tex_ok}\vspace{-1ex}
			\caption{A \kl{tree} in $\lang\Gg$ (\cref{ex:lang})}
			\label{fig:ex:lang}
		\end{minipage}
	\end{figure*}

\paragraph{Recursion schemes as recognizers of languages of finite trees.}

	The standard semantics of a \kl{recursion scheme} $\Gg = \tuple{\Alphabet,\Nn,X_0,\Rr}$
	is the single infinite \kl{tree} $\BT\Gg$ \kl{generated by} the scheme.
	An alternative view is to consider a \kl{recursion scheme}
	as a \kl(scheme){recognizer} of a language of finite \kl{trees} $\lang \Gg$.
	This alternative view is relevant when discussing \kl{downward closures} of languages of finite \kl{trees}.
\AP	We employ a special \kl{letter} $\intro*\ndd\in\Alphabet$ of \kl{rank} $2$
	in order to represent $\lang \Gg$ by resolving the nondeterministic choice of $\ndd$ in the infinite \kl{tree} $\BT \Gg$ in all possible ways.
	Formally, for two \kl{trees} $T, U$, we write $T\rightarrow_{\ndd} U$
	if $U$ is obtained from $T$ by choosing an $\ndd$-labeled \kl{node} $u$ of $T$
	and a \kl{child} $v$ thereof, and replacing the \kl{subtree} rooted at $u$ with the \kl{subtree} rooted at $v$.
	The relation $\rightarrow^*_{\ndd}$ is the reflexive and transitive closure of $\rightarrow_{\ndd}$.
\AP	We define the language of finite \kl{trees} \intro(scheme){recognized} by $\Gg$ as $\intro*\lang\Gg = \NT{\BT\Gg}$, where\phantomintro{\NT}
	\begin{equation*}
		\reintro*\NT T = \setof U {T\rightarrow^*_{\ndd} U \text{, with $U$ finite and not containing ``$\ndd$'' or ``$\bot$''}}.
	\end{equation*}

	\begin{exampleOK}\label{ex:lang}
		For the \kl{recursion scheme} $\Gg$ from \cref{ex:scheme}, $\lang\Gg$ is the non-regular language of all finite \kl{trees} of the form
		\begin{align*}
			\leta\applyto(\underbrace{\letb_1\applyto(\letb_1\applyto(\dots\applyto(\letb_1}_n\applyto\letc)\dots)))
				\applyto(\underbrace{\letb_2\applyto(\letb_2\applyto(\dots(\letb_2}_n\applyto\letc)\dots)))&&\mbox{for }n \in \Nat,
		\end{align*}
		depicted in \cref{fig:ex:lang}.
	\end{exampleOK}

\paragraph{Alternating B-automata.}

	We introduce the model of automata used in this paper,
	namely \intro[cost automata]{alternating one-way/two-way B-automata} over \kl{trees} (over a \kl{ranked alphabet}).
\AP	We consider counters which can be \intro{incremented} $\intro*\inc$, \intro{reset} $\intro*\reset$, or left \intro{unchanged} $\intro*\e$.
\AP	Let $\Gamma$ be a finite set of \intro{counters}
	and let $\intro*\Counters = \set{\inc, \reset, \e}$ be the \intro[counter action]{alphabet of counter actions}.
\AP	Each \kl{counter} starts with value zero,
	and the \intro{value of a sequence} of \kl[counter action]{actions} 
	is the supremum of the values achieved during this sequence.
\AP	For instance $\inc \inc \reset\e \inc \e$ has \kl{value} $2$,
	$(\inc\reset)^\omega$ has \kl{value} $1$,
	and $\inc\reset\inc^2\reset\inc^3\reset\cdots$ has \kl{value} $\infty$.
\AP	For an infinite sequence of \kl{counter actions} $w \in \Counters^\omega$,
	let $\intro*\val(w) \in \Nat \cup \set{\infty}$ be its \kl{value}.
	In case of several \kl{counters}, $w = c_1 c_2 \cdots \in (\Counters^{\Gamma})^\omega$, we take the \kl{counter} with the maximal \kl{value}:
	$\val(w) = \max_{c \in \Gamma} \val(w(c))$,
	where $w(c) = c_1(c) c_2(c) \cdots$.
	
\AP	An \intro{(alternating, two-way) B-automaton}
	over a finite \kl{ranked alphabet} $\Alphabet$ is a tuple
	\begin{align*}
		\tuple{\Alphabet, \allowbreak Q, \allowbreak q_0, \allowbreak\pr, \allowbreak\Gamma, \allowbreak\delta}
	\end{align*}
	consisting of a finite set of \intro{states} $Q$, an \intro{initial state} $q_0 \in Q$, a function $\intro*\pr\colon Q\to\Nat$ assigning \intro{priorities} to \kl{states}, 
\AP	a finite set $\intro*\Gamma$ of \kl{counters}, and a \intro(aut){transition function}
	\phantomintro\delta\phantomintro\dup\phantomintro\de\phantomintro\ddown
	\begin{align*}
		\reintro*\delta : Q \times \Alphabet \to \Bb^{+}(\set{\dup, \de, \ddown_1, \ddown_2, \dots}\times \Counters^{\Gamma} \times Q)
	\end{align*}
	mapping a \kl{state} and a \kl{letter} $a$ to a (finite) positive Boolean combination of triples of the form $(d,\act,q)$;
	it is assumed that if $d=\ddown_i$ then $i\leq\rank(a)$. 
	Such a triple encodes the instruction to send the automaton in the direction $d$ while performing the action $\act$, and changing the \kl{state} to $q$.
	The direction $\ddown_i$ denotes moving to the $i$-th \kl{child}, $\dup$ moving to the \kl{parent}, and $\de$ staying in place.
	We assume that $\delta(q,a)$ is written in disjunctive normal form for all $q$ and $a$.
	
	The acceptance of an infinite input \kl{tree} $T$ by an \kl{alternating B-automaton} $\Aa$ is defined in terms of a game $\gameon{\Aa}{T}$
	between two players, called Eve and Adam.
	Eve is in charge of disjunctive choices
	and tries to minimize \kl{counter values} while satisfying the parity condition.
	Adam, on the other hand, is in charge of conjunctive choices and tries to either maximize \kl{counter values},
	or to sabotage the parity condition.
	Since the \kl(aut){transition function} is given in disjunctive normal form,
	each turn of the game consists of Eve choosing a disjunct and Adam selecting a single triple $(d,\act,q)$ thereof.
	In order to deal with the situation that the automaton wants to go up from the \kl{root} of the \kl{tree},
	we forbid Eve to choose a disjunct containing a triple with direction $\dup$ when the play is in the \kl{root}.
	Simultaneously, we assume that $\delta(q,a)$ for all $q$ and $a$ contains a disjunct in which no triple uses the direction $\dup$,
	so that from every position there is some move.
\AP	A \intro{play} of $\Aa$ on a \kl{tree} $T$ is a sequence
	$q_0,(d_1,\act_1,q_1),(d_2,\act_2,q_2),\dots$
	compatible with $T$ and $\delta$: $q_0$ is the \kl{initial state}, 
	and for all $i \in \Nat$, $(d_{i+1}, \act_{i+1}, q_{i+1})$ appears in $\delta(q_i, T(x_i))$, where $x_i$ is the \kl{node} of $T$ after following the directions $d_1, d_2, \dots, d_i$ starting from the \kl{root}.
	The \reintro{value} $\val(\pi)$ of such a \kl{play} $\pi$ is the \kl{value} $\val(\act_1 \act_2 \cdots)$ as defined above 
	if the largest number appearing infinitely often among the \kl{priorities} $\pr(q_0), \pr(q_1), \dots$ is even;
	otherwise, $\val(\pi) = \infty$.
\AP	We say that the \kl{play} $\pi$ is \intro[$n$-winning play]{$n$-winning} (for Eve) if $\val(\pi)\leq n$.
	
\AP	A \intro{strategy} for one of the players in the game $\gameon{\Aa}{T}$ is a function that returns the next choice given the history of the play.
	Note that choosing a \kl{strategy} for Eve and a \kl{strategy} for Adam fixes a \kl{play} in $\gameon{\Aa}{T}$. 
	We say that a \kl{play} $\pi$ is \intro{compatible} with a \kl{strategy} $\sigma$
	if there is some \kl{strategy} $\sigma'$ for the other player such that $\sigma$ and $\sigma'$ together yield the \kl{play} $\pi$. 
\AP	A \kl{strategy} for Eve is \intro[$n$-winning strategy]{$n$-winning} if every \kl{play} \kl{compatible} with it is \kl[$n$-winning play]{$n$-winning}.
\AP	We say that Eve \intro(game){$n$-wins} the game if there is some \kl{$n$-winning strategy} for Eve.
\AP	The \kl{B-automaton} \intro(aut){$n$-accepts} a \kl{tree} $T$ if Eve \kl(game){$n$-wins} the game $\gameon{\Aa}{T}$;
	it \intro(aut){accepts} $T$ if it \kl(aut){$n$-accepts} $T$ for some $n \in \Nat$.
\AP	The language \intro(aut){recognized} by $\Aa$ is the set of all \kl{trees} \kl(aut){accepted} by $\Aa$.

	\begin{exampleOK}\label{ex:b-autom}
		Let $\Alphabet$ be the \kl{ranked alphabet} containing a \kl{letter} $\leta$ of \kl{rank} 2 and a \kl{letter} $\letb$ of \kl{rank} $1$.
		Consider a \kl{B-automaton} over $\Alphabet$ with one \kl{counter} and three \kl{states} $\letq_0,\letq_1,\letq_2$, all of \kl{priority} $0$;
		the \kl{state} $\letq_0$ is initial, and the \kl(aut){transitions} are
		\begin{align*}
			\delta(\letq_0,\leta)&=(\ddown_1,\e,\letq_0)\land(\ddown_2,\e,\letq_0),\\
			\delta(\letq_0,\letb)&=((\ddown_1,\e,\letq_0)\land(\dup,\inc,\letq_1))\lor((\ddown_1,\e,\letq_0)\land(\down_1,\inc,\letq_2)),\hspace{-20em}\\
			\delta(\letq_1,\leta)&=(\dup,\inc,\letq_1)\lor(\de,\inc,\letq_1),&
			\delta(\letq_1,\letb)&=(\de,\e,\letq_1),\\
			\delta(\letq_2,\leta)&=(\ddown_1,\inc,\letq_2)\lor(\ddown_2,\inc,\letq_2),&
			\delta(\letq_2,\letb)&=(\de,\e,\letq_2).
		\end{align*}
		Here Adam chooses a $\letb$-labeled \kl{node} $u$ (using \kl{state} $\letq_0$),
		and then Eve selects a path to some $\letb$-labeled ancestor (using \kl{state} $\letq_1$) or descendant (using \kl{state} $\letq_2$) of $u$;
		the \kl{counter} computes the distance between these two \kl{nodes}.
		In consequence, a \kl{tree} is \kl(aut){accepted} if there is a bound $n\in\Nat$ such that every $\letb$-labeled \kl{node} has a $\letb$-labeled ancestor or descendant in distance at most $n$.
	\end{exampleOK}

\AP	If no $\delta(q, a)$ uses the direction $\dup$, then we call $\Aa$ \intro{one-way}.
	Blumensath, Colcombet, Kuperberg, Parys, and Vanden Boom \cite[Theorem 6]{quasi-weak-logic} show that every \kl{B-automaton} can be made one-way:
	
	\begin{theoremOK}\label{thm:2way-to-1way}
		Given an alternating \kl{two-way B-automaton}, one can compute an \kl{alternating one-way B-automaton} that \kl(aut){recognizes} the same language.
	\end{theoremOK}
	
	\begin{proof}
		This essentially follows from the result of Blumensath et al.\@ \cite[Theorem 6]{quasi-weak-logic} modulo some cosmetic changes.
		Namely, due to some differences in definitions, our \cref{thm:2way-to-1way} is weaker in two aspects and stronger in one aspect 
		than the result of Blumensath et al.~\cite[Theorem 6]{quasi-weak-logic}.
		We elaborate on these differences here.
	
		First, Blumensath et al.~\cite{quasi-weak-logic} do not say that a \kl{one-way B-automaton} $\Aa$ and a \kl{two-way B-automaton} $\Bb$ \kl(aut){recognize} the same languages,
		but rather that the cost functions defined by these \kl{B-automata} are equal (modulo domination equivalence).
		The latter means that there exists a non-decreasing function $\alpha\colon\Nat\to\Nat$ such that if one of the \kl{B-automata} ($\Aa$ or $\Bb$) \kl(aut){$n$-accepts} some \kl{tree} $T$,
		then the other \kl{B-automaton} \kl(aut){$\alpha(n)$-accepts} this \kl{tree} $T$.
		Clearly this is a stronger notion; it implies that the sets of \kl(aut){accepted} \kl{trees} are equal.
	
		Second, the definition of \kl{one-way B-automata} given by Blumensath et al.~\cite{quasi-weak-logic} forbids the usage of the direction $\de$ (along with $\dup$),
		while we allow to use $\de$ (only $\dup$ is forbidden).
		A translation to \kl{one-way B-automata} becomes only easier if their definition is less restrictive.
		We remark, however, that we actually need to allow the usage of $\de$ in order to correctly handle \kl{letters} of \kl{rank} $0$---we do not want a \kl{one-way B-automaton} to get stuck in a \kl{node} without \kl{children}.

		Third, the \kl{B-automata} of Blumensath et al.~\cite{quasi-weak-logic} work over binary \kl{trees}, that is, all \kl{letters} of the alphabet are assumed to be of \kl{rank} $2$, 
		while we allow \kl{letters} of arbitrary \kl{ranks}.
		It is not difficult to believe that the assumption about a binary alphabet is just a technical simplification, 
		and that all the proofs of Blumensath et al.~\cite{quasi-weak-logic} can be repeated for an arbitrary alphabet.
		Alternatively, it is possible to encode a \kl{tree} over an arbitrary \kl{ranked alphabet} into a binary \kl{tree}, 
		using the first-child next-sibling representation (with some dummy infinite binary \kl{tree} encoding ``no more children'').
		Such an encoding can easily be incorporated into a \kl{B-automaton}.
		Thus, in order to convert a \kl{two-way B-automaton} $\Aa$ into a \kl{one-way B-automaton} $\Bb$, 
		we can first convert it into a \kl{two-way B-automaton} $\Aa_2$ over a binary alphabet (reading the first-child next-sibling representation of a \kl{tree}),
		then convert $\Aa_2$ into a \kl{one-way B-automaton} $\Bb_2$ (using the results of Blumensath et al.~\cite[Theorem 6]{quasi-weak-logic}),
		and then convert $\Bb_2$ into $\Bb$ reading the actual \kl{tree} instead of its first-child next-sibling representation.	
	\end{proof}
	
	\begin{exampleOK}
		In general, the proofs of Blumensath et al.\@ \cite{quasi-weak-logic} underlying \cref{thm:2way-to-1way}
		are nontrivial.
		Nevertheless, in the concrete case of the \kl{B-automaton} $\Aa$ from \cref{ex:b-autom} it is not difficult to directly construct a \kl{one-way B-automaton} $\Bb$ recognising the same language.
		The trick is that, instead of going up to a close $\letb$-labeled ancestor,
		already in the ancestor we decide that it will serve as a close $\letb$-labeled ancestor for some \kl{node}.
		Thus the \kl(aut){transitions} become
		\begin{align*}
			\delta(\letq_0,\leta)&=(\ddown_1,\e,\letq_0)\land(\ddown_2,\e,\letq_0), \hspace{5em}
			\delta(\letq_0,\letb)=(\ddown_1,\inc,\letq_1),\\
			\delta(\letq_1,\leta)&=((\ddown_1,\inc,\letq_1)\land(\ddown_2,\reset,\letq_0))\lor((\ddown_1,\reset,\letq_0)\land(\ddown_2,\inc,\letq_1))\lor
				((\ddown_1,\inc,\letq_1)\land(\ddown_2,\inc,\letq_1)),\\
			\delta(\letq_1,\letb)&=(\ddown_1,\reset,\letq_0)\lor(\ddown_1,\inc,\letq_1).
		\end{align*}
	\end{exampleOK}

	As a special case of a result by Colcombet and G\"oller~\cite{games-with-bound-guess-actions} we obtain the following fact:
	
	\begin{factOK}\label{fact:qw-reg-decidable}
		One can decide whether a given $B$-automaton $\Aa$ \kl(aut){accepts} a given \kl{regular tree} $T$.
	\end{factOK}
	
	\begin{proof}
		First, thanks to \cref{thm:2way-to-1way}, we can assume that $\Aa$ is one-way.
		Next, recall that \kl(aut){acceptance} of $T$ is defined in terms of a game $(\Aa,T)$.
		When $\Aa$ is one-way and $T$ is \kl{regular}, this game has actually a finite arena.
		Indeed, for the future of a play, it does not matter what is the current \kl{node} of $T$,
		it only matters which \kl{subtree} starts in the current \kl{node}---and in $T$ we have finitely many different \kl{subtrees}.
		It is not difficult to decide whether such a finite-arena game is $n$-won by Eve for some $n\in\Nat$.
		Nevertheless, instead of showing this directly, we notice that
		games obtained this way are a special case of games considered by Colcombet and G\"oller~\cite{games-with-bound-guess-actions}, for which they prove decidability.
	\end{proof}


\section{Model-checking safe recursion schemes against alternating B-auto\-ma\-ta}
\label{sec:model-checking}

	In this section we prove the first main theorem of our paper, that is,
	decidability of the \intro{model-checking problem} of \kl{safe recursion schemes} against properties described by \kl{B-automata}:
	
	\begin{theoremOK}\label{thm:main-1}
		Given an \kl{alternating B-automaton} $\Aa$ and a \kl{safe recursion scheme} $\Gg$,
		one can decide whether $\Aa$ \kl(aut){accepts} the \kl{tree generated by} $\Gg$.
	\end{theoremOK}
	
	It is worth noticing that this theorem generalizes the result of Knapik et al.~\cite{easy-trees} on \kl{safe recursion schemes}
	from regular (\kl{MSO}) properties to the more general quantitative realm of properties described by \kl{B-automata}.
	On the other hand, our result is incomparable with the celebrated theorem of Ong~\cite{Ong:LICS:2006}
	showing decidability of model checking regular properties of possibly unsafe \kl{recursion schemes}.
	Whether model checking of possibly unsafe \kl{recursion schemes} against properties described by \kl{B-automata}
	is decidable remains an open problem.
	
	By \cref{thm:2way-to-1way}, every \kl{B-automaton}
	can effectively be transformed into an equivalent \kl{one-way B-automaton},
	so it is enough to prove \cref{thm:main-1} for a \kl{one-way B-automaton} $\Aa$.
	The proof of \cref{thm:main-1} is based on the following lemma,
	where we use in an essential way the assumption that the \kl{recursion scheme} is \kl(scheme){safe}:
	
	\begin{lemmaOK}\label{lem:reduce-order}
		For every \kl{safe recursion scheme} $\Gg$ of \kl(scheme){order} $m$
		and for every alternating \kl{one-way B-automaton} $\Aa$, 
		one can effectively construct a \kl{safe recursion scheme} $\reify\Gg$ of \kl(scheme){order} $m-1$
		and an alternating \kl{two-way B-automaton} $\Aa'$ such that
		\begin{align*}
			\Aa\text{ \kl(aut){accepts} }\BT \Gg
			\quad \text{ if and only if } \quad
			\Aa'\text{ \kl(aut){accepts} }\BT{\reify\Gg}.
		\end{align*}
	\end{lemmaOK}
	
	Notice that the above lemma allows us to decrease the \kl(scheme){order} of a \kl{recursion scheme},
	at the cost of transforming a \kl{one-way B-automaton} $\Aa$ into a \kl{two-way B-automaton} $\Aa'$.

	Before proving \cref{lem:reduce-order}, let us see how \cref{thm:main-1} follows from it:
	Using \cref{lem:reduce-order} we can reduce the \kl(scheme){order} of the considered \kl{safe recursion scheme} by one.
	We obtain a \kl{two-way B-automaton},
	which we convert back to a \kl{one-way B-automaton} using \cref{thm:2way-to-1way}.
	It is then sufficient to repeat this process,
	until we end up with a \kl{recursion scheme} of \kl(scheme){order} $0$.
	A \kl{recursion scheme} of \kl(scheme){order} $0$ \kl{generates} a \kl{regular tree}
	and, by \cref{fact:qw-reg-decidable},
	we can decide whether the resulting \kl{B-automaton} \kl(aut){accepts} this \kl{tree},
	answering our original question.

\paragraph{Lambda-trees.}

	We now come to the proof of \cref{lem:reduce-order}.
	The construction of $\reify\Gg$ from $\Gg$
	follows an analogous result for \kl{MSO} \cite{easy-trees,hyperalgebraic-trees},
	which we generalize to \kl{B-automata}.
	We call the construction of $\reify\Gg$ from $\Gg$ \reintro{reification}.
	This is the central idea in Knapik et al.~\cite{easy-trees,hyperalgebraic-trees},
	which first proved decidability of \kl{MSO} model checking of safe recursion schemes.
	We formally present it in \cref{sec:reification};
	here, we illustrate it with some examples.
	%
	
\AP	For a \emph{finite} set $\Xx$ of \kl{variables of type} $\otyp$,
	we define a new \kl{ranked alphabet} $\intro*\AlphabetXx$ that contains
	\begin{inparaenum}[(1)]
	\itemAP a \kl{letter} $\intro*\tconst{a}$ of \kl{rank} $0$ for every \kl{letter} $a\in\Alphabet$;
	\itemAP a \kl{letter} $\intro*\tvar{x}$ of \kl{rank} $0$ for every \kl(lambda){variable} $x\in\Xx$;
	\itemAP a \kl{letter} $\intro*\tlambda{x}$ of \kl{rank} $1$ for every \kl(lambda){variable} $x\in\Xx$;
	\itemAP a \kl{letter} $\intro*\tat$ of \kl{rank} $2$.
	\end{inparaenum}
	We remark that $\AlphabetXx$ is a usual finite \kl{ranked alphabet}.
\AP	A \intro{lambda-tree} is a \kl{tree} over the alphabet $\AlphabetXx$.
	%
	\kl{Reification} takes a \kl{lambda-term} $M$
	and produces a new \kl{lambda-term} $\reify M$, in which the maximal \kl(lambda){order} of \kl{subterms} is strictly smaller.
	Moreover, when $M$ is \kl{first-order},
	$\reify M$ is a \kl{lambda-tree}
	(i.e., a \kl{closed} \kl{lambda-term} of \kl(lambda){order} $0$ over the \kl{alphabet} specified above).
	Intuitively, order-zero \kl{lambda-binders} $\lambda x. K$ (with $x$ of type $\otyp$)
	and \kl{applications} $K\applyto L$ with an \kl(lambda){order}-zero argument $L$
	are \kl{reified} into the syntax:
	The \kl{lambda-binder} $\lambda x.K$ becomes a \emph{\kl{letter}} $\tlambda x$ applied to the recursively \kl{reified} $\reify K$
	and the \kl{application} $K\applyto L$ becomes also a \emph{\kl{letter}} $\tat$ applied to the recursively \kl{reified} $\reify K$ and $\reify L$.
	This is demonstrated in the next two examples.

	\begin{figure*}
		\centering
		\begin{minipage}[c]{0.48\textwidth}
			\centering
			\def\svgscale{0.5}\import{pics/}{lam_tree_1.pdf_tex_ok}\vspace{-1ex}
		\end{minipage}
		\begin{minipage}[c]{0.48\textwidth}
			\centering
			\def\svgscale{0.5}\import{pics/}{lam_tree_2.pdf_tex_ok}\vspace{-1ex}
		\end{minipage}
			\caption{The \kl{lambda-tree} $T = \reify M$ from \cref{ex:lambda-tree} (left)
			and its \kl{$(\Xx,2)$-derived tree} $\semdt{\Xx,2} T$ (right)}
			\label{fig:ex:lambda-tree}
	\end{figure*}

	\begin{exampleOK}\label{ex:lambda-tree}
		Consider the \kl{first-order lambda-term} (\kl{of type} $\otyp$)
		\begin{align*}
			M = (\lambda\varx.\lambda\vary.\leta\applyto\varx\applyto\vary)\applyto\letc_1\applyto\letc_2.
		\end{align*}
		In this case $\Alphabet$ contains a \kl{letter} $\leta$ of \kl{rank} 2 and two \kl{letters} $\letc_1, \letc_2$ of \kl{rank} 0,
		and $\Xx=\set{\varx,\vary}$.
		The new alphabet is thus $\AlphabetXx=\set{\tconst{\leta},\tconst{\letc_1},\tconst{\letc_2},\tvar{\varx},\tvar{\vary},\tlambda{\varx},\tlambda{\vary},\tat}$,
		where the \kl{letter} $\tat$ is of \kl{rank} $2$, the \kl{letters} $\tlambda{\varx},\tlambda{\vary}$ are of \kl{rank} $1$, and the other \kl{letters} are of \kl{rank} $0$.
		The \kl{reification} of the \kl{lambda-term} $M$
		is the \kl{lambda-tree}
		\begin{align*}
			\reify M=\tat\applyto(\tat\applyto(\tlambda{\varx}\applyto(\tlambda{\vary}\applyto(\tat\applyto(\tat\applyto\tconst{\leta}\applyto\tvar{\varx})\applyto\tvar{\vary})))\applyto\tconst{\letc_1})\applyto\tconst{\letc_2}
		\end{align*}
		depicted in \cref{fig:ex:lambda-tree} (left).
		Notice that while $M$ contains actual \kl{lambda-binders} ``$\lambda x$'' and ``$\lambda y$'',
		its \kl{reification} $\reify M$ contains only \kl{letters}
		(i.e., no \kl(lambda){variables} and no \kl{lambda-binders}).
	\end{exampleOK}

	The following example shows how \kl{reification} is applied to a \kl{lambda-term} which is not \kl{first-order}.
	
	\begin{exampleOK}
		Consider the \kl{lambda-term}
		\begin{align*}
			M = \lambda \varf.\lambda \varx.
				\leta \applyto \varx \applyto (\varf\applyto \varx).
		\end{align*}
		\kl{of type}
		\begin{align*}
			\alpha = (\otyp \arr \otyp) \arr \otyp \arr \otyp.
		\end{align*}
		Applying \kl{reification} to $M$
		(formally defined 
		in \cref{appendix})
		yields the \kl{lambda-term}
		\begin{align*}
			\reify M = \lambda{\reify\varf}.\tlambda\varx \applyto (\tat \applyto (\tat \applyto \tconst\leta \applyto \tvar\varx) \applyto (\tat \applyto \reify\varf \applyto \tvar\varx))
		\end{align*}
		of the \kl{reified} \kl[of type]{type}
		\begin{align*}
			\reify \alpha = \otyp\arr\otyp.
		\end{align*}
		Notice that in this case the \kl{reified} \kl{lambda-term} $\reify M$
		is not a \kl{lambda-tree}
		since the original \kl{lambda-term} $M$ is not \kl{first-order}:
		\kl{Reification} is performed only to \kl(lambda){order}-zero \kl{lambda-binders},
		and \kl{applications} with \kl(lambda){order}-zero arguments;
		higher-order \kl{lambda-binders} and \kl{applications} with higher-order arguments are not reified.
		In particular, in $\reify M$ we still have a \kl{lambda-binder} ``$\lambda \reify\varf$'',
		where the \kl(lambda){variable} $\varf$
		\kl{of type} $\beta = \otyp \arr \otyp$
		has become the \kl(lambda){variable} $\reify\varf$
		\kl{of type} $\reify\beta = \otyp$.
	\end{exampleOK}

\AP	We now define the \reintro{$(\Xx,s)$-derived tree} of a \kl{lambda-tree} $T = \reify M$, denoted $\semdt{\Xx,s} {T}$.
	The construction of the \kl{derived tree} can be seen as a counterpart of the \kl{Böhm tree} on the side of \kl{reified} \kl{lambda-terms}.
	Namely, the derived tree is defined in such a way that the \kl{derived tree} of the \kl{reification} of $M$ equals the \kl{Böhm tree} of $M$,
	that is, $\semdt{\Xx,s} {\reify M} = \BT M$
	(c.f.~\cref{lem:new-scheme,lem:reify:sucessor:BT}).
	Thus \kl{derived trees} formally explain how to recover the semantics of $M$ (its \kl{Böhm tree})
	by only looking at its \kl{reification} $\reify M$.
	This is used later in \cref{lem:1-way-to-2-way}
	to state how \kl{two-way automata} on $\reify M$ can simulate \kl{one-way automata} on $M$.
	
	The definition of the \kl{derived tree} exploits the fact that
	a \kl{first-order lambda-term} $M$ uses only \kl{variables of type} $\otyp$.
	We can thus read the \kl{B\"ohm tree} of $M$ directly,
	without performing any \kl{reduction},
	just by exploring its \kl{reification} $T = \reify M$.
	%
	Essentially, we walk down through $T$,
	skipping all \kl{reified} \kl{lambda-binders} $\tlambda x$ and choosing the left branch in all \kl{reified} \kl{applications} $\tat$.
	Whenever we reach some \kl{reified} \kl(lambda){variable} $\tvar x$, we go up to the corresponding \kl{reified} \kl{lambda-binder} $\tlambda x$,
	then up to the corresponding \kl{reified} \kl{application} $\tat$, 
	and then we again start going down in the argument of this \kl{application}.
	
	Formally, let $\Xx$ be a finite set of \kl{variables of type} $\otyp$, and let $s\in\Nat$.
	The intended meaning is that $\Xx$ contains \kl(lambda){variables} that may potentially appear in the considered \kl{lambda-tree} $T$,
	and that $s$ is a bound for the arity of \kl{types} in the \kl{lambda-term} represented by $T$ (the \kl{types} of all \kl{subterms} thereof should be of the form $\otyp^k\arr\otyp$ for $k\leq s$).
\AP	We take\footnote{These directions are unrelated with directions in tree automata from \cref{sec:preliminaries}.} $\intro*\Dirs_{\Xx,s}=\set{\intro*\Rdown}\cup\set{\intro*\Rup_x\mid x\in\Xx}\cup\set{\reintro*\Rup_i\mid 1\leq i\leq s}$.
	Intuitively, $\Rdown$ means that we go down to the left \kl{child} of a \kl{node} labelled by $\tat$
	or to the unique \kl{child} of a \kl{node} labelled by ``$\tlambda{x}$'',
	$\Rup_x$ means that we are going up while looking for the value of the \kl(lambda){variable} $x$,
	and $\Rup_i$ means that we are going up while looking for the $i$-th argument of an \kl{application}.
\AP	For a \kl{node} $v$ of $T$ denote its \kl{parent} by $\intro*\parent(v)$,
	and its $i$-th \kl{child} by $\intro*\child_i(v)$,
	where $1 \leq i \leq k$ and $k$ is the arity of $v$.
\AP	For $d\in\Dirs_{\Xx,s}$, and for a \kl{node} $v$ of $T$ labeled by $\zeta\in\AlphabetXx$, 
	we define the \intro{$(\Xx,s)$-successor} of $(d,v)$,
	when it exists, as
	\begin{compactenum}
	\itemAP	$(\Rdown,\child_1(v))$ if $d=\Rdown$ and $\zeta=\tlambda{x}$ (for some $x$) or $\zeta=\tat$,
	\itemAP	$(\Rup_x,\parent(v))$ if $d=\Rdown$ and $\zeta=\tvar{x}$ (for some $x$),
	\itemAP	$(\Rup_x,\parent(v))$ if $d=\Rup_x$ and $\zeta\neq\tlambda{x}$ (including the case when $\zeta=\tlambda{y}$ for $y\neq x$),
	\itemAP	$(\Rup_1,\parent(v))$ if $d=\Rup_x$ and $\zeta=\tlambda{x}$,
	\itemAP	$(\Rup_{i+1},\parent(v))$ if $d=\Rup_i$ for $i<s$ and $\zeta=\tlambda{y}$ (for some $y$),
	\itemAP	$(\Rup_{i-1},\parent(v))$ if $d=\Rup_i$ for $i>1$ and $\zeta=\tat$,
	\itemAP	$(\Rdown,\child_2(v))$ if $d=\Rup_1$ and $\zeta=\tat$.
	\end{compactenum}
	In particular, the \kl{$(\Xx,s)$-successor} of $(d,v)$ is again a pair $(d',v')$, where $d'\in\Dirs_{\Xx,s}$ and $v'$ is a \kl{node} of $T$.
	Note that every pair $(d,v)$ has at most one \kl{$(\Xx,s)$-successor}, but there may be pairs without any \kl{$(\Xx,s)$-successors}.
	In particular, pairs with $v$ labelled by $\tconst a$ do not have a \kl{successor}.
	
	Rule~1 allows us go to down to the first \kl{child} in the case of \kl{reified} \kl{lambda-binders} $\tlambda x$
	and \kl{reified} \kl{applications} $\tat$.
	Rule~2 records that we have seen a \kl{reified} \kl(lambda){variable} $\tvar x$ (which is a \kl{letter}),
	and thus we need to find its value by going up.
	Rule~3 climbs the \kl{tree} upwards as long as we do not see the corresponding \kl{reified} \kl{lambda-binder} $\tlambda x$.
	Rule~4 records that we have seen $\tlambda x$ and initializes its level to $1$.
	We now need to find the corresponding \kl{application}.
	Rule~5 increments the level and goes up when we encounter a \kl{reified} \kl{lambda-binder} $\tlambda y$ (which is just a \kl{letter}),
	and Rule~6 decrements it for \kl{reified} \kl{applications}~$\tat$.
	Finally, when we see a \kl{reified} \kl{application} at level~$1$,
	we apply Rule~7 which searches for the value of~$\tvar x$ in the right \kl{child}.

\AP	An \intro{$(\Xx,s)$-maximal path} from $(d_1,v_1)$ is a sequence of pairs $(d_1,v_1),(d_2,v_2),\dots$
	in which every $(d_{i+1},v_{i+1})$ is the \kl{$(\Xx,s)$-successor} of $(d_i,v_i)$,
	and which is either infinite or ends in a pair that has no \kl{$(\Xx,s)$-successor}.
\AP	For $d\in\Dirs_{\Xx,s}$, and for a \kl{node} $v$ of $T$,
	we define the \intro{$(\Xx,s)$-derived tree} from $(T,d,v)$,
	denoted by $\semdt{\Xx,s}{T, d, v}$,
	by coinduction:
	\begin{compactitem}
	\item	if the \kl{$(\Xx,s)$-maximal path} from $(d,v)$ is finite and ends in $(\Rdown,w)$ for a \kl{node} $w$ labeled by $\tconst{a}$,
		then 
		\begin{align*}
			\semdt{\Xx,s}{T, d, v} = a\applyto\semdt{\Xx,s}{T,\Rup_1,\parent(w)}\applyto\dots\applyto\semdt{\Xx,s}{T,\Rup_{\rank(a)},\parent(w)};
		\end{align*}
	\item	otherwise, $\semdt{\Xx,s}{T, d, v} = \bot$.
	\end{compactitem}
\AP	The \reintro{$(\Xx,s)$-derived tree} from $T$ is 
	$\intro[\semdt]{\semdt {\Xx,s} T} = \semdt{\Xx,s}{T,\Rdown,v_0}$, where $v_0$ is the \kl{root} of $T$.
	We say that $T$ is \intro(lambda-tree){normalizing} if $\semdt {\Xx,s} T$ does not contain $\bot$.
	
	\begin{figure*}
		\centering
		\begin{minipage}[c]{0.39\textwidth}
			\centering
			\def\svgscale{0.5}\import{pics/}{successors_1.pdf_tex_ok}\vspace{-1ex}
		\end{minipage}
		\begin{minipage}[c]{0.59\textwidth}
			\centering
			\def\svgscale{0.5}\import{pics/}{successors_2.pdf_tex_ok}\vspace{-1ex}
		\end{minipage}
		\caption{Illustration to \cref{ex:successors}. Arrows denote \kl{successors} in the construction of the \kl{derived tree}.}
		\label{fig:ex:successors}
	\end{figure*}

	\begin{exampleOK}\label{ex:successors}
		Let us come back to the \kl{lambda-tree} $T$ from \cref{ex:lambda-tree} (depicted in \cref{fig:ex:lambda-tree}).
		Denote the \kl{nodes} on the leftmost \kl{branch} of $T$ by $v_1,\dots,v_7$ ($v_1$ is the \kl{root}, and $v_7$ is the $\tconst{\leta}$-labeled \kl{leaf}),
		and the other four \kl{nodes} with labels $\tconst{\letc_2},\tconst{\letc_1},\tvar{\vary},\tvar{\varx}$ by $w_1,w_2,w_3,w_4$, respectively.

		To find the \kl{root} of the \kl{$(\Xx,2)$-derived tree} of $T$, we need to follow the \kl{$(\Xx,2)$-maximal path} from $(\Rdown,v_1)$.
		The \kl{$(\Xx,2)$-successor} of $(\Rdown,v_1)$ is $(\Rdown,v_2)$;
		its \kl{$(\Xx,2)$-successor} is $(\Rdown,v_3)$, and so on;
		the path ends in $(\Rdown,v_7)$, which has no \kl{$(\Xx,2)$-successor}.
		Thus the \kl{root} of the \kl{derived tree} is labelled with $\leta$.
		This is shown in the left part of \cref{fig:ex:successors}.

		To find the right \kl{child} of this \kl{root}, we need to follow the \kl{$(\Xx,2)$-maximal path} from $(\Rup_2,v_6)$.
		This path goes through $(\Rup_1,v_5),\allowbreak(\Rdown,w_3),(\Rup_\vary,v_5),(\Rup_\vary,v_4),(\Rup_1,v_3),(\Rup_2,v_2),(\Rup_1,v_1),(\Rdown,w_1)$;
		the last pair has no \kl{$(\Xx,2)$-successor}, so the right \kl{child} of the \kl{root} in the \kl{$(\Xx,2)$-derived tree} of $T$ is labelled by $\letc_2$.
		Note that the \kl{node} $v_5$ is visited twice,
		with two different directions.
		This is shown in the right part of \cref{fig:ex:successors}.

		The left \kl{child} can be found in an analogous way,
		starting from $(\Rup_1,v_6)$.
		The resulting \kl{$(\Xx,2)$-derived tree} $\semdt{\Xx,2} T$ is thus $\leta\applyto\letc_1\applyto\letc_2$, depicted in \cref{fig:ex:lambda-tree} (right).
	\end{exampleOK}

	The next example shows how \kl{reification} is applied to a whole recursive scheme.
	\begin{exampleOK}
		Recall the \kl{recursion scheme} $\Gg$ from \cref{ex:scheme}, having the following two \kl{rules}:
		\begin{align*}
			\Rr(\nonS) &= \nonA\applyto \letb_1\applyto \letb_2\applyto \letc\applyto \letc, \\
			\Rr(\nonA) &= \lambda \varf.\lambda \varg.\lambda \varx.\lambda \vary.
				\ndd\applyto (\leta\applyto \varx\applyto \vary)\applyto (\nonA\applyto \varf\applyto \varg\applyto (\varf\applyto \varx)\applyto (\varg\applyto \vary)).
		\end{align*}
		Applying \kl{reification} to the \kl{recursion scheme} $\Gg$
		one obtains the \kl{recursion scheme} $\reify\Gg$ (formally defined in \cref{eq:reified scheme} in the appendix)
		with the following two \kl{rules}
		\begin{align*}
			\reify\Rr(\reify\nonS) &= \tat\applyto(\tat\applyto(\nonA\applyto\tconst{\letb_1}\applyto\tconst{\letb_2})\applyto\tconst{\letc})\applyto\tconst{\letc}, \\
			\reify\Rr(\reify\nonA) &= \lambda\varf.\lambda\varg.\tlambda{\varx}\applyto\big(\tlambda{\vary}\applyto\big(\tat\applyto\big(\tat\applyto\tconst{\ndd}\applyto
				(\tat\applyto(\tat\applyto\tconst{\leta}\applyto\tvar{\varx})\applyto\tvar{\vary})\big)\applyto
				\big(\tat\applyto(\tat\applyto(\nonA\applyto\varf\applyto\varg)\applyto(\tat\applyto\varf\applyto\tvar{\varx}))\applyto(\tat\applyto\varg\applyto\tvar{\vary})\big)\big)\big).
		\end{align*}
	\end{exampleOK}

\AP	The following lemma describes existence of the \kl{reified} \kl{recursion scheme} $\reify\Gg$, satisfying necessary properties.
	It crucially relies on the \kl(scheme){safety} assumption.
	%
	
	\begin{lemmaOK}\label{lem:new-scheme}
		For every \kl(scheme){normalizing} \kl{safe recursion scheme} $\Gg$ of \kl(scheme){order} $m\geq 1$
		one can construct a \kl{safe recursion scheme} $\reify\Gg$ of \kl(scheme){order} $m-1$,
		a finite set of \kl(lambda){variables} $\Xx$,
		and a number $s\in\Nat$
		such that
		\begin{align*}
			\semdt{\Xx,s}{\BT {\reify\Gg}} = \BT \Gg.
		\end{align*}
	\end{lemmaOK}

	All the crucial ingredients of the proof of \cref{lem:new-scheme} (with some differences in definitions, and with some omitted details)
	are already contained in papers of Knapik et al.~\cite{easy-trees,hyperalgebraic-trees}.
	In the interest of being self-contained,
	we provide a full proof of \cref{lem:new-scheme} in \cref{appendix}.
	Here, we content ourselves with providing a high-level description of the proof.
	To construct $\reify\Gg$,
	one needs to replace in $\Gg$ every \kl(lambda){variable} $x$ \kl{of type} $\otyp$ by $\tvar{x}$, every \kl{lambda-binder} concerning such a \kl(lambda){variable} by $\tlambda{x}$, 
	and every \kl{application} with an argument \kl{of type} $\otyp$ by a construct creating a $\tat$-labeled \kl{node},
	as demonstrated in the examples above.
	Types of \kl{subterms} change and the \kl(scheme){order} of the \kl{recursion scheme} decreases by one.
	While in general computing $\BT{\Lambda(\Gg)}$ requires one to rename \kl(lambda){variables}
	in order to perform \kl{capture-avoiding substitutions},
	in the \kl{tree generated by} the modified \kl{recursion scheme} we leave the original variable names unchanged.
	In general (i.e., when the transformation is applied to an arbitrary, possibly unsafe, \kl{recursion scheme})
	this is incorrect due to overlapping variable names and thus possibly unsound \kl{substitutions}.
	The assumption that $\Gg$ is \kl(scheme){safe} is crucial here:
	there is no need to rename \kl(lambda){variables}
	when applying the transformation to a \kl{safe recursion scheme}.
	We refer to \cref{appendix} for a detailed proof of \cref{lem:new-scheme}.
		
	Recall that we are heading towards proving \cref{lem:reduce-order}.
	Having \cref{lem:new-scheme}, it remains to transform a \kl{one-way B-automaton} $\Aa$ operating on the \kl{tree generated by} $\Gg$
	into a \kl{two-way B-automaton} $\Aa'$ operating on the \kl{lambda-tree} \kl{generated by} $\reify\Gg$, as described by the following lemma
	(as mentioned on page~\pageref{page:fully-convergent}, we can assume that $\Gg$ is \kl(scheme){normalizing},
	which implies that $\BT{\reify\Gg}$ is \kl(lambda-tree){normalizing}: the \kl{tree} $\semdt{\Xx,s}{\BT{\reify\Gg}}=\BT{\Gg}$ does not contain $\bot$):

	\begin{lemmaOK}\label{lem:1-way-to-2-way}
		Let $\Aa$ be an \kl{alternating one-way B-automaton} over a finite alphabet $\Alphabet$,
		let $\Xx$ be a finite set of \kl(lambda){variables}, and let $s\in\Nat$.
		One can construct an \kl{alternating two-way B-automaton} $\Aa'$ such that
		for every \kl(lambda){normalizing} \kl{lambda-tree} $T$ over $\AlphabetXx$,
		\begin{align*}
			\Aa\text{ \kl(aut){accepts} }\semdt {\Xx, s} T
			\quad \text{ if and only if } \quad
			\Aa'\text{ \kl(aut){accepts} } T.
		\end{align*}
	\end{lemmaOK}
	
	\begin{proof}
		The \kl{B-automaton} $\Aa'$ simulates $\Aa$ on the \kl{lambda-tree}.
		Whenever $\Aa$ wants to go down to the $i$-th \kl{child}, $\Aa'$ has to follow the \kl{$(\Xx,s)$-maximal path} from $(\Rup_i,v)$ (where $v$ is the current \kl{node}).
		To this end, it has to remember the current pair $(d,v)$, and repeatedly find its \kl{$(\Xx,s)$-successor}.
		Here $v$ is always just the current \kl{node} visited by the \kl{B-automaton};
		the $d$ component comes from the (finite) set $\Dirs_{\Xx,s}$,
		and thus it can be remembered in the \kl{state}.
		It is straightforward to encode the definition of an \kl{$(\Xx,s)$-successor} in \kl(aut){transitions} of an automaton; we thus omit these tedious details.
		We do not have to worry about infinite \kl{$(\Xx,s)$-maximal paths}, because by assumption the \kl{$(\Xx,s)$-derived tree} does not contain $\bot$-labeled \kl{nodes}.
	\end{proof}

	\Cref{lem:reduce-order} is thus proved by applying \cref{lem:new-scheme} and \cref{lem:1-way-to-2-way}.
	In turn, this proves \cref{thm:main-1}, the main result of this section.

\section{Downward closures of tree languages}
\label{sec:ideals}

	In this section we lay down a method for computation of the \kl{downward closure}
	for classes of languages of finite \kl{trees} closed under \kl{linear FTT transductions}, which we define in \cref{dec:dc:transductions}.
	This method is analogous to the one of Zetzsche~\cite{Zetzsche:ICALP:2015} for the case of finite words.
	In \cref{sec:dc:prelim} we define the \kl{downward closure} of languages of finite ranked \kl{trees} with respect to the \kl{embedding well-quasi order}
	and in \cref{sec:dc:sup} we define the \kl{simultaneous unboundedness problem} for \kl{trees}
	and show how computing the \kl{downward closure} reduces to it.
	In \cref{sec:dc:diagonal} we define the \kl{diagonal problem} for finite \kl{trees}
	and show how the previous problem reduces to it.
	The development of this section is summarised by the following theorem (notions used in its statement are defined in the sequel):
	\begin{theoremOK}
		\label{thm:reduction}
	 	Let $\Class$ be a class of languages of finite \kl{trees} effectively closed under \kl{linear FTT transductions}.
		If the \kl{diagonal problem} for $\Class$ is decidable,
		then \kl{downward closures} are computable for~$\Class$.
	\end{theoremOK}

	Let us emphasize that results of this section 
	can be applied to any class of languages of finite \kl{trees} closed under \kl{linear FTT transductions},
	not just those \kl(scheme){recognized} by \kl{safe recursion schemes}.
	In \cref{sec:downward-closure} we will solve the \kl{diagonal problem} in the particular case of languages of finite \kl{trees} \kl(scheme){recognized} by \kl{safe recursion schemes},
	and then we will exploit \cref{thm:reduction} to show that \kl{downward closures} are effectively computable for these languages.

\subsection{Representations of downward-closed languages}
\label{sec:dc:prelim}

\AP	Let $\intro*\embedsinto$ be the least relation on finite \kl{trees} such that
	\begin{inparaenum}[(1)]
		\item $S\embedsinto b\applyto T_1\applyto\dots\applyto T_r$ if $S \embedsinto T_i$ for some $i\in\set{1,\dots,r}$, and
		\item $a\applyto S_1\applyto\dots\applyto S_r\embedsinto b\applyto T_1\applyto\dots\applyto T_r$ if $S_i \embedsinto T_i$ for all $i\in\set{1,\dots,r}$.
	\end{inparaenum}
	When $S\embedsinto T$, we say that \intro{$S$ homeomorphically embeds into $T$}.
\AP	For a language of finite \kl{trees} $\Ll$, its \intro{downward closure}, denoted by $\intro*\downwardclosure \Ll$,
	is the set of \kl{trees} $S$ such that $S \embedsinto T$ for some \kl{tree} $T \in \Ll$.

	\begin{exampleOK}
		The \kl{tree} $\leta\applyto\letc_1\applyto\letc_2$ \kl{embeds} into the \kl{tree} $\letb\applyto(\leta\applyto(\leta\applyto\letc_1\applyto\letc_1)\applyto\letc_2)$,
		but it does not \kl{embed} into the \kl{tree} $\leta\applyto(\leta'\applyto\letc_1\applyto\letc_2)\applyto\letc_1$.
	\end{exampleOK}

	\begin{exampleOK}
		The \kl{downward closure} of the language $\lang\Gg$ from \cref{ex:lang} consists of all finite \kl{trees} of the form
		\begin{align*}
			&\underbrace{\letb_1\applyto(\letb_1\applyto(\dots\applyto(\letb_1}_n\applyto\letc)\dots)) &&\mbox{for $n\in \Nat$,}\\
			&\underbrace{\letb_2\applyto(\letb_2\applyto(\dots\applyto(\letb_2}_m\applyto\letc)\dots)) &&\mbox{for $m\in \Nat$, or}\\
			&\leta\applyto(\underbrace{\letb_1\applyto(\letb_1\applyto(\dots\applyto(\letb_1}_n\applyto\letc)\dots)))
				\applyto(\underbrace{\letb_2\applyto(\letb_2\applyto(\dots\applyto(\letb_2}_m\applyto\letc)\dots)))&&\mbox{for $n,m \in \Nat$.}
		\end{align*}
		Notice that $\lang\Gg$ is a non-regular language of finite \kl{trees},
		while its \kl{downward closure} above is in fact regular.
	\end{exampleOK}

\paragraph{Simple tree regular expressions.}

\AP	Goubault-Larrecq and Schmitz~\cite{Goubault-Larrecq:Schmitz:ICALP:2016} describe \kl{downward-closed} sets of \kl{trees} using \intro{simple tree regular expressions} (\reintro*{{\STRE}s}),
	which we now introduce.

\AP	A \intro{context} is a \kl{tree} possibly containing one or more occurrences of a special \kl{leaf} $\intro*\hole$,
	called a \intro{hole}.
\AP	Given a \kl{context} $C$ and a set of \kl{trees} $\Ll$,
	we write $\replace C \Ll$ for the set of \kl{trees} obtained from $C$ by replacing every occurrence of the \kl{hole} $\hole$ by some \kl{tree} from $\Ll$.
	Different occurrences of $\hole$ are replaced by possibly different \kl{trees} from $\Ll$.
	The definition readily extends to a set of \kl{contexts} $\Cc$,
	by writing $\replace \Cc \Ll$ for $\bigcup_{C \in \Cc} \replace C \Ll$.
	If $C$ does not have any $\hole$, then $\replace C \Ll$ is just $\set C$.

\AP	An \intro*\STRE is defined according to the following abstract syntax:
	\begin{align*}
		S       &::= P + \dots + P,&
		I	&::= C + \dots + C,&
		S_{\hole}	&::= \hole \alt S.\\
		P	&::= \opt a\applyto S\applyto\dots\applyto S \alt I^*.S,&
		C	&::= a\applyto S_{\hole}\applyto\dots\applyto S_{\hole},
	\end{align*}
	These expressions allow empty sums, which are denoted by $0$.
\AP 
	Subexpressions of the form $P$, $I$, and $C$ are called \intro{pre-products}, \intro{iterators}, and \reintro{contexts}, respectively.
	The word ``context'' is thus used to describe two different kinds of objects: \kl{trees} with \kl{holes}, and expressions of the form $C$ (denoting sets of \kl{trees} with \kl{holes}).
	
\AP	An \STRE $S$ denotes a set of \kl{trees} $\intro[\sem]{\sem S}$ \kl{downward-closed} for $\embedsinto$,
	which is defined recursively as follows:\phantomintro\opt\phantomintro\iter
	\begin{align*}
		\sem {P_1 + \dots + P_k} &= \sem {P_1} \cup\dots\cup \sem {P_k},\\
		\sem {\reintro*\opt a\applyto S_1\applyto\dots\applyto S_r} &= \downwardclosure{\setof{a\applyto T_1\applyto\dots\applyto T_r}{\forall i\,.\,T_i \in \sem {S_i}}},\\
		\sem {\reintro*\iter I.S} &= \bigcup_{n\in\Nat} \underbrace{\sem{I}[\dots[\sem{I}[}_n\!\sem{S}]]\dots],\\
		\sem {C_1 + \dots + C_k} &= \sem {C_1} \cup \dots \cup \sem {C_k},\\
		\sem {a\applyto S_{\hole, 1}\applyto\dots\applyto S_{\hole, r}} &= \downwardclosure{\setof{a\applyto T_1\applyto\dots\applyto T_r}{\forall i\,.\,T_i \in \sem {S_{\hole, i}}}},\\
		\sem {\hole} &= \set {\hole}.
	\end{align*}
\AP	Two \STRE{}s $S, T$ are \intro[equivalent STREs]{equivalent} if $\sem S = \sem T$.
	Since the sets $\sem {S_i}$ are \kl{downward closed}, we can see that if all $\sem {S_i}$ are nonempty, then 
	\begin{align}\label[formula]{star}\tag{$\star$}
		\sem{\opt a\applyto S_1\applyto\dots\applyto S_r} &= \setof{a\applyto T_1\applyto\dots\applyto T_r}{\forall i\,.\,T_i \in \sem {S_i}} \cup \sem {S_1} \cup \cdots \cup \sem {S_r}.
	\end{align}
	If, however, $\sem{S_i}=\emptyset$ for some $i\in\set{1,\dots,r}$, then $\sem{\opt a\applyto S_1\applyto\dots\applyto S_r}=\emptyset$.
	We have the same property also for $\sem {a\applyto S_{\hole, 1}\applyto\dots\applyto S_{\hole, r}}$.

	\begin{exampleOK}
		The set $\sem {\iter{(\leta\applyto\letb\applyto\hole)}.\opt \letc}$ (where $\leta$ is of \kl{rank} $2$, and $\letb,\letc$ are of \kl{rank} $0$)
		consists of \kl{trees} of the form either $\letb$, or $\letc$,
		or $\leta\applyto\letb\applyto\allowbreak(\leta\applyto\letb\applyto(\dots\allowbreak(\leta\applyto\letb\applyto x)\dots))$, where $x$ is either $\letb$ or $\letc$.
	\end{exampleOK}
	
	The following lemma is shown by Goubault-Larrecq and Schmitz~\cite[Proposition 18]{Goubault-Larrecq:Schmitz:ICALP:2016}:
	
	\begin{lemmaOK}
		\label{lem:STRE}
		For every \kl{downward-closed} set of \kl{trees} $\Ll$
		there exists an \STRE $S$ such that $\Ll = \sem S$
		(and, vice versa, every \STRE $S$ denotes a \kl{downward-closed} set of \kl{trees} $\sem S$).
	\QED\end{lemmaOK}

\paragraph{Products.}

\AP	Among all \STRE{}s, Goubault-Larrecq and Schmitz~\cite{Goubault-Larrecq:Schmitz:ICALP:2016} distinguish \reintro{products}, which describe \emph{ideals} of \kl{trees}.
	Because every \kl{downward-closed} set of \kl{trees} is a finite union of ideals, such a set can be described by a finite list of products;
	this is the idea staying behind \cref{cor:STRE:product} below.
	
	In order to define products, Goubault-Larrecq and Schmitz~\cite{Goubault-Larrecq:Schmitz:ICALP:2016} give a way of simplifying {\STRE}s by means of a rewrite relation $\to_1$.
	Intuitively, the idea is to move the operator ``$+$'' inside-out as much as possible,
	and a product is a \STRE where no more rewriting can be done.
\AP	A \kl{context} $a\applyto S_{\hole, 1}\applyto\dots\applyto S_{\hole, r}$ is \intro[linear context]{linear} if at most one $S_{\hole, i}$ is a \kl{hole} $\hole$,
	and it is \intro[full context]{full} if $r\geq 1$ and all the $S_{\hole, j}$'s are
	\kl{holes} $\hole$.
\AP	An \kl{iterator} $C_1 + \dots + C_k$ is \intro[linear iterator]{linear} (\intro[full iterator]{full}) if all the $C_i$'s are \kl{linear contexts} (\kl{full contexts}, respectively).
	Assuming that ``$+$'' is commutative and associative, we define the rewrite relation $\to_1$ as follows:
	\begin{align}
		P + P' &\to_1 P'
			&&\textrm { if } \sem{P} \subseteq \sem{P'}, \\
		C + C' &\to_1 C'
			&&\textrm { if } \sem{C} \subseteq \sem{C'}, \label{rule:elim-smaller-iter}\\
		0^*.S&\to_1 S, \label{rule:elim-empty-iter}\\
		\opt a\applyto S_1\applyto\dots\applyto S_{i-1}\applyto 0\applyto S_{i+1}\applyto\dots\applyto S_r&\to_1 0, \label{rule:zero-in-opt}\\
		a\applyto S_{\hole, 1}\applyto\dots\applyto S_{\hole, i-1}\applyto 0\applyto S_{\hole, i+1}\applyto\dots\applyto S_{\hole, r} &\to_1 0, \label{rule:zero-in-context}\\ 
		I^*.0 &\to_1 0
			&& \textrm{ if $I$ is \kl[full iterator]{full},} \label{rule:remove-full-zero}\\
		(I + (a\applyto S_1\applyto\dots\applyto S_r))^*.S &\to_1 I^*.(S + (\opt a\applyto S_1\applyto\dots\applyto S_r)), \hspace{-2em}\label{rule:extract-ho-hole}\\
		\opt a\applyto S_1\applyto\dots\applyto S_{i-1}\applyto(S_i+S_i')\applyto S_{i+1}\applyto\dots\applyto S_r &\to_1\nonumber\\*
			&\hspace{-9.7em}\opt a\applyto S_1\applyto\dots\applyto S_{i-1}\applyto S_i\applyto S_{i+1}\applyto\dots\applyto S_r
			+\opt a\applyto S_1\applyto\dots\applyto S_{i-1}\applyto S_i'\applyto S_{i+1}\applyto\dots\applyto S_r,
			\hspace{-8em}\label{rule:distr-in-opt}\\
		a\applyto S_{\hole,1}\applyto\dots\applyto S_{\hole,i-1}\applyto(S_{\hole,i}+S_{\hole,i}')\applyto S_{\hole,i+1}\applyto\dots\applyto S_{\hole,r} &\to_1\nonumber\\*
			&\hspace{-16em}a\applyto S_{\hole,1}\applyto\dots\applyto S_{\hole,i-1}\applyto S_{\hole,i}\applyto S_{\hole,i+1}\applyto\dots\applyto S_{\hole,r}
			+a\applyto S_{\hole,1}\applyto\dots\applyto S_{\hole,i-1}\applyto S_{\hole,i}'\applyto S_{\hole,i+1}\applyto\dots\applyto S_{\hole,r},
			\hspace{-8em}\label{rule:distr-in-context}\\
		I^*.(S + S') &\to_1
			I^*.S + I^*.S'
			&&\textrm{ if $I$ is \kl[linear iterator]{linear}.}\label{rule:distr-in-linear}
	\end{align}
	We allow to apply $\to_1$ for subexpressions of an \STRE,
	that is, we write $S\to_1 S'$ also when $S'$ is obtained from $S$ by replacing some its subexpression $R$ with $R'$ such that $R\to_1 R'$.
	
\AP	A \intro{product} is a \kl{pre-product} $P$ that is a normal form with respect to $\to_1$,
	that is, there is no $P'$ such that ${P \to_1 P'}$.
	We know that the rewrite relation $\to_1$ preserves the denotation of \STRE~\cite[Fact~19]{Goubault-Larrecq:Schmitz:ICALP:2016},
	and that every \STRE has a normal form with respect to $\to_1$~\cite[Lemma~20]{Goubault-Larrecq:Schmitz:ICALP:2016}.
	The following corollary is immediate:

	\begin{corollaryOK}
		\label{cor:STRE:product}
		Every \STRE $S$ is \kl[equivalent STREs]{equivalent} to a sum of \kl{products} $P_1 + \dots + P_k$.
	\QED\end{corollaryOK}

\paragraph{\kl{Pure products}.}

	Since the definition of a \kl{product} is rather indirect,
	we introduce a stronger notion of \reintro{pure products},
	which is defined as a syntactic restriction of \STRE{}s.
	Such a definition is more convenient for our purposes.
	Simultaneously, it still allows us to obtain a decomposition result stated in \cref{lem:pure-products},
	which is an analogue of \cref{cor:STRE:product} for \kl{pure products} instead of \kl{products}.
	
\AP	A \intro{pure product} is defined according to the following abstract syntax:
	\begin{align*}
		P	&::= \opt a\applyto P\applyto\dots\applyto P \alt I^*.P,&
		C	&::= a\applyto P_{\hole}\applyto\dots\applyto P_{\hole},\\
		I	&::= C + \dots + C,&
		P_{\hole}	&::= {\hole}\applyto \alt P,
	\end{align*}
	where the sum of \kl{contexts} is nonempty, and where in a \kl{context} $C=a\applyto P_{\hole,1}\applyto\dots\applyto P_{\hole,r}$ it is required that at least one $P_{\hole,i}$ is a \kl{hole} $\hole$.
	The semantics $\sem P$ of \kl{pure products} is inherited from \STRE.
	Notice, however, that $\sem P$ is always nonempty, so we can use \cref{star} to define $\sem{\opt a\applyto P_1\applyto\dots\applyto P_r}$ and $\sem {a\applyto P_{\hole, 1}\applyto\dots\applyto P_{\hole, r}}$.

	Formally, a \kl{pure product} needs not be a \kl{product}: a \kl{pure product} is allowed to contain an \kl{iterator} $C+C'$ with $\sem{C}\subseteq\sem{C'}$,
	to which \cref{rule:elim-smaller-iter} can be applied.
	Nevertheless, by replacing every such sum $C+C'$ with $C'$ we can obtain an \kl[equivalent STREs]{equivalent} \kl{pure product} that is a \kl{product}
	(it is easy to see that no rule other than \cref{rule:elim-smaller-iter} can be applied to a \kl{pure product}).
	Thus, it is justified to say that, morally, the notion of a \kl{pure product} strengthens the notion of a \kl{product}.
	
	Based on the results of Goubault-Larrecq and Schmitz~\cite{Goubault-Larrecq:Schmitz:ICALP:2016},
	in the remaining part of this subsection we deduce the following lemma:

	\begin{lemmaOK}\label{lem:pure-products}
		Every set of \kl{trees} $\Ll$ \kl{downward-closed} for $\embedsinto$ can be represented as $\Ll=\sem{P_1}\cup\dots\cup\sem{P_k}$, in which $P_1,\dots,P_k$ are \kl{pure products}.
	\end{lemmaOK}
	
	This decomposition result strengthens the results of Goubault-Larrecq and Schmitz~\cite{Goubault-Larrecq:Schmitz:ICALP:2016}
	by showing that \kl{pure products} (instead of just \kl{products}) suffice in order to represent \kl{downward-closed} sets of \kl{trees}.

\paragraph{From products to pure products.}

	In \cref{lemma:make-pure} we show how to convert an arbitrary \kl{product} into a \kl{pure product}.
	\cref{lem:pure-products} is then a direct consequence of \cref{lem:STRE}, \cref{cor:STRE:product}, and \cref{lemma:make-pure}.
		
	\begin{lemmaOK}\label{lemma:make-pure}
		For every \kl{product} $P$ one can create an \kl[equivalent STREs]{equivalent} \kl{pure product} $P'$.
	\end{lemmaOK}
	
	\begin{proof}
		The proof is by induction on the size of $P$.
		Before starting the actual proof, let us observe two facts, which we use implicitly below.
		First, every subterm of $P$ that is a \kl{pre-product} is actually a \kl{product} (i.e., it cannot be rewritten by $\to_1$).
		Second, if we replace a \kl{product} subexpression of $P$ by some \kl[equivalent STREs]{equivalent} \kl{product}, then the resulting \STRE is still a \kl{product}
		(i.e., it cannot be rewritten by $\to_1$).
		
		Coming now to the proof, suppose that $P$ is of the form $\opt a\applyto S_1\applyto\dots\applyto S_r$.
		Then, because $P$ is a \kl{product}, that is, because it cannot be rewritten by $\to_1$, we can observe that all the $S_i$'s are \kl{products}.
		Indeed, if some $S_i$ was a sum of two or more \kl{products} (or $0$),
		then $P$ could be rewritten using \cref{rule:distr-in-opt} (or \cref{rule:zero-in-opt}, respectively).
		By the induction assumption for every \kl{product} $S_i$ we can create an \kl[equivalent STREs]{equivalent} \kl{pure product} $S_i'$;
		as $P'$ we take $\opt a\applyto S_1'\applyto\dots\applyto S_r'$.
		
		Next, suppose that $P$ is of the form $I^*.S$.
		Consider a \kl{context} $C=a\applyto S_{\hole,1}\applyto\dots\applyto S_{\hole,r}$, being a component of $I$.
		We can observe that all $S_{\hole,i}$ are either $\hole$ or \kl{products}.
		Indeed, if some $S_{\hole,i}$ was a sum of two or more \kl{products} (or $0$), 
		then $C$ could be rewritten using \cref{rule:distr-in-context} (or \cref{rule:zero-in-context}, respectively).
		As previously, using the induction assumption we can replace every \kl{product} $S_{\hole,i}$
		that is not a \kl{hole} $\hole$ by an \kl[equivalent STREs]{equivalent} \kl{pure product} $S_{\hole,i}'$.
		Applying this to every \kl{context} $C$ in $I$, we obtain a new \kl{iterator} $I_\circ$ in which every \STRE subterm is a single \kl{product}.
		Likewise, writing $S=P_1+\dots+P_k$, we can replace every \kl{product} $P_i$ by an \kl[equivalent STREs]{equivalent} \kl{pure product} $P_i'$.
		This way, we obtain a \kl{product} $P^\circ=I_\circ^*.(P_1'+\dots+P_k')$ \kl[equivalent STREs]{equivalent} to $P$.
		Observe also that there is at least one \kl{context} in $I_\circ$, and that every \kl{context} in $I_\circ$ contains a \kl{hole} 
		(because \cref{rule:elim-empty-iter,rule:extract-ho-hole} cannot be applied to $P^\circ$), as required in our definition of a \kl{pure product}.
		Thus, when $k=1$, $P^\circ$ is a \kl{pure product}, hence it can be taken as $P'$.
		It remains to deal with the situation when $k\neq 1$.
		
		One possibility is that $k=0$.
		Then $I_\circ$ is not \kl[full iterator]{full} (otherwise \cref{rule:remove-full-zero} could be applied to $P^\circ$), 
		which means that in $I_\circ$ there is a \kl{context} $C'=a\applyto S_{\hole,1}'\applyto\dots\applyto S_{\hole,r}'$ 
		such that $S_{\hole,j}'\neq\hole$ for some $j\in\set{1,\dots,r}$.
		Fix one such $C'$ and $j$, and define $P':=I_\circ^*.S_{\hole,j}'$.
		Then $P'$ is a \kl{pure product}.
		Clearly $\sem{P^\circ}\subseteq \sem{P'}$, because $\sem{0}\subseteq\sem{S_{\hole,j}'}$.
		On the other hand, $\sem{S_{\hole,i}'}\neq\emptyset$ for all $i\in\set{1,\dots,r}$,
		because $S_{\hole,i}'$ is either a \kl{hole} or a \kl{pure product},
		and it can be easily seen (by induction on its structure) that a \kl{pure product} always denotes a nonempty set;
		in consequence $\sem{S_{\hole,j}'}\subseteq\replace{\sem{C'}}\emptyset\subseteq\replace{\sem{I_\circ}}\emptyset$, so also $\sem{P'}\subseteq\sem{P^\circ}$.
		
		Another possibility is that $k\geq 2$.
		Then $I_\circ$ is not \kl[linear iterator]{linear} (if $I_\circ$ were \kl[linear iterator]{linear}, then \cref{rule:distr-in-linear} could be applied to $P^\circ$),
		which means that in $I_\circ$ there is a \kl{context} $C'=a\applyto S_{\hole,1}'\applyto\dots\applyto S_{\hole,r}'$ with two or more \kl{holes}.
		Fix one such $C'$.
		For simplicity, we show the proof assuming that the first $\ell$ among $S_{\hole,i}'$ are \kl{holes}, 
		and the remaining $r-\ell$ among $S_{\hole,i}'$ are \kl{products} (i.e., are not \kl{holes});
		the general situation can be handled in the same way, but writing it down would require us to use intricate indices.
		We define 
		\begin{align*}
			R_1&=\opt a\applyto P_1'\applyto P_1'\applyto\dots\applyto P_1'\applyto S_{\hole,\ell+1}'\applyto\dots\applyto S_{\hole,r}'&&\text{and}\\
			R_i&=\opt a\applyto R_{i-1}\applyto P_i'\applyto\dots\applyto P_i'\applyto S_{\hole,\ell+1}'\applyto\dots\applyto S_{\hole,r}'&&\text{for }i\in\set{2,\dots,k},
		\end{align*}
		and we define $P':=I_\circ^*.R_k$.
		Notice that $P'$ is a \kl{pure product}.
		On the one hand, $\sem{P_i'}\subseteq\sem{R_k}$ for every $i\in\set{1,\dots,k}$ 
		(it is important here that there are at least two \kl{holes}, so $P_i'$ actually appears in $R_i$),
		so $\sem{P^\circ}\subseteq \sem{P'}$.
		Let us see the opposite inclusion.
		First, by definition, $\sem{P_i'}\subseteq\sem{I_\circ^*.(P_1'+\dots+P_k')}=\llbracket{P^\circ}\rrbracket$ for all $i\in\set{1,\dots,k}$. 
		Second, because $R_i$ for $i\in\set{2,\dots,k}$ is obtained by substituting $P_i'$ and $R_{i-1}$ for all \kl{holes} in $C'$,
		we have $\sem{R_i}\subseteq\replace{\sem{C'}}{\sem{P_i'}\cup\sem{R_{i-1}}}$; likewise $\sem{R_1}\subseteq\replace{\sem{C'}}{\sem{P_1'}}$.
		Then, by induction on $i\in\set{1,\dots,k}$ we have that $\sem{R_i}\subseteq\sem{P^\circ}$:
		indeed, due to the above observation and the induction hypothesis (if $i\geq 2$) we have that
		$\sem{R_i}\subseteq\replace{\sem{C'}}{\sem{P^\circ}}\subseteq\replace{\sem{I_\circ}}{\sem{P^\circ}}\subseteq\sem{P^\circ}$.
		In particular, $\sem{R_k}\subseteq\sem{P^\circ}$, so also $\sem{P'}\subseteq\sem{P^\circ}$.
	\end{proof}

\subsection{Transductions}\label{dec:dc:transductions}

	A (nondeterministic) \intro{finite tree transducer} (\reintro{FTT}) is a tuple $\Aa = (\Alphabet_\mathit{in},\allowbreak \Alphabet_\mathit{out},\allowbreak S,\allowbreak p^I,\allowbreak\Delta)$,
	where $\Alphabet_\mathit{in}$, $\Alphabet_\mathit{out}$ are the input and output \kl[ranked alphabet]{alphabets} (finite, ranked), 
	$S$ is a finite set of \intro(ftt){control states}, $p^I\in S$ is an \intro(ftt){initial state}, 
	and $\Delta$ is a finite set of \intro(ftt){transition rules} of the form either $(p, a\applyto\varx_1\applyto\dots\applyto\varx_r) \rightarrow V$
	or $(p, \varx) \rightarrow V$,
	where $p\in S$ is a \kl(ftt){control state}, $a\in\Alphabet_\mathit{in}$ is a \kl{letter} of \kl{rank} $r$, 
	and $V$ is a finite \kl{tree} over the alphabet $\Alphabet_\mathit{out}\cup(S\times\set{\varx_1,\dots,\varx_r})$
	or $\Alphabet_\mathit{out}\cup(S\times\set{\varx})$, respectively.
	Here $\varx,\varx_1,\dots,\varx_r$ are just some special symbols, and the \kl{rank} of all the pairs from $S\times\set{\varx_1,\dots,\varx_r}$ or $S\times\set{\varx}$ is $0$.
\AP	An \kl{FTT} is \intro(ftt){linear} if for each \kl(ftt){rule} of the form $(p, a\applyto\varx_1\applyto\dots\applyto\varx_r) \rightarrow V$
	and for each $i\in\set{1,\dots,r}$, in $V$ there is at most one \kl{letter} from $S\times\set{\varx_i}$,
	and moreover for each \kl(ftt){rule} of the form $(p, \varx) \rightarrow V$, in $V$ there is at most one \kl{letter} from $S\times\set{\varx}$.

\AP	An \kl{FTT} $\Aa = (\Alphabet_\mathit{in},\allowbreak \Alphabet_\mathit{out},\allowbreak S,\allowbreak p^I,\allowbreak\Delta)$
	reading a \kl{tree} $T$ over the alphabet $\Alphabet_\mathit{in}$ starts in the \kl(ftt){state} $p^I$ at the \kl{root} of $T$.
	Then, when $\Aa$ is in a \kl(ftt){state} $p$ at the \kl{root} of a \kl{subtree} $T'=a\applyto T_1\applyto\dots\applyto T_r$ of $T$,
	it can use a \kl(ftt){rule} of the form $(p,a\applyto\varx_1,\dots\applyto\varx_r)\rightarrow V$ from $\Delta$;
	it produces a \kl{tree} starting like $V$, but \kl{leaves} of the form $(q,\varx_i)$ are replaced by the output of running $\Aa$ in the \kl(ftt){state} $q$ at the \kl{root} of $T_i$.
	Alternatively, $\Aa$ can use a \kl(ftt){rule} of the form $(p,\varx)\rightarrow V$;
	then it produces a \kl{tree} starting like $V$, but \kl{leaves} of the form $(q,\varx)$ are replaced by the output of running $\Aa$ in the \kl(ftt){state} $q$ at the same \kl{node}
	(i.e., this is an $\varepsilon$-transition producing some output).
	In this way, an \kl{FTT} $\Aa$ defines a relation between finite \kl{trees}, also denoted $\Aa$;
	for a fully formal definition see Comon et al.~\cite[Section 6.4.2]{tata2007}.
	For a language $\Ll$ we write $\Aa(\Ll)$ for the set of \kl{trees} $U$ such that $(T,U)\in\Aa$ for some $T\in \Ll$.
\AP	A function that maps $\Ll$ to $\Aa(\Ll)$ for some \kl(ftt)[linear]{linear FTT} $\Aa$ is called a \intro{linear FTT transduction}.
	
	We now recall two easy facts about \kl{linear FTT transductions}.
	The first fact says that taking \kl{downward closures} is an \kl{FTT transduction}:
	
	\begin{factOK}\label{fact:dwclosure-is-FTT}
		Given a finite \kl{ranked alphabet} $\Alphabet$ one create a \kl{linear FTT} $\Aa$ such that for every language $\Ll$ of finite \kl{trees} over $\Alphabet$,
		the language $\Aa(\Ll)$ equals $\downwardclosure \Ll$.
	\end{factOK}
	
	\begin{proof}
		It is enough to take $\Aa=(\Alphabet,\Alphabet,\set{\stap},\stap,\Delta)$ with a single \kl(ftt){state} $\stap$,
		where $\Delta$ consists of the following \kl(ftt){rules} for every \kl{letter} $a\in\Alphabet$ of \kl{rank} $r$:
		\begin{align*}
			(\stap, a\applyto\varx_1\applyto\dots\applyto\varx_r) &\rightarrow a\applyto(\stap,\varx_1)\applyto\dots\applyto(\stap,\varx_r),&&\mbox{and}\\
			(\stap, a\applyto\varx_1\applyto\dots\applyto\varx_r) &\rightarrow (\stap,\varx_i),&&\mbox{for all }i\in\set{1,\dots,r}.
		\end{align*}
		Such a \kl{transducer} can convert every \kl{tree} $T\in\Ll$ into every \kl{tree} $S$ that \kl{homeomorphically embeds} into $T$.
	\end{proof}

	The second fact says that \kl{linear FTT transductions} can implement intersections with regular languages:
	
	\begin{factOK}\label{fact:reg-intersection-is-FTT}
		Given (a finite tree automaton recognizing) a regular language $\Rr$ of finite \kl{trees} over a finite \kl{ranked alphabet} $\Alphabet$,
		one create a \kl{linear FTT} $\Aa$ such that for every language $\Ll$ of finite \kl{trees} over $\Alphabet$, the language $\Aa(\Ll)$ equals $\Ll\cap \Rr$.
	\end{factOK}

	\begin{proof}
		We are given an automaton that accepts a \kl{tree} $T$ if $T\in\Rr$ (and rejects it otherwise),
		and we want to construct a \kl{linear FTT} $\Aa$ that converts a \kl{tree} $T$ into itself if $T\in\Rr$ (and does not allow to produce any output \kl{tree} otherwise).
		Creating $\Aa$ out of the automaton is just a matter of changing the syntax:
		we take to $\Aa$ all transition rules of the automaton, enhancing them so that the input \kl{tree} is produced again on the output.
	\end{proof}

\subsection{The \kl{simultaneous unboundedness problem} for \kl{trees}}
\label{sec:dc:sup}

\AP	We say that a \kl{pure product} $P$ is \intro{diversified}, if no \kl{letter} appears in $P$ more than once.
	The \intro{simultaneous unboundedness problem} (\intro{\SUP}) for a class $\Class$ of finite \kl{trees} asks,
	given a \kl{diversified pure product} $P$ and a language $\Ll\in\Class$ such that $\Ll \subseteq \sem P$,
	whether $\sem P \subseteq \downwardclosure \Ll$.

	\begin{remarkOK}
		This is a generalization of \SUP over finite words.
		In the latter problem, one is given a language of finite words $\Ll$ such that $\Ll \subseteq a_1^* \dots a_k^*$,
		and must check whether $a_1^* \dots a_k^* \subseteq \downwardclosure \Ll$.
		A word in $a_1^* \dots a_k^*$
		can be represented as a linear \kl{tree} by interpreting $a_1, \dots, a_k$ as unary \kl{letters}
		and by appending a new \kl{leaf} $e$ at the end.
		Thus $a_1^* \dots a_k^*$ can be represented as the language of the \kl{diversified pure product}
		$\iter{(a_1\applyto\hole)}.\iter{(a_2\applyto\hole)} . \cdots . \iter{(a_k\applyto\hole)} .\opt e$.
	\end{remarkOK}
	
	Every \kl{pure product} $P$ can be made \kl{diversified} by adding additional marks to \kl{letters} appearing in $P$.
	Namely, for each \kl{letter} $a$ appearing $k$ times in $P$, we consider ``marked'' \kl{letters} $a_1,\dots,a_k$, and for each occurrence of $a$ in $P$ we substitute a different \kl{letter} $a_i$.
	To specify a correspondence between the original \kl{pure product} $P$ and the resulting \kl{diversified pure product} $P'$ we define a $\cl(\cdot)$ operation:
	when $X'$ is an object (e.g., a \kl{pure product}, a \kl{context}, a \kl{tree}, etc.) over such an extended alphabet,
	we write $\cl(X')$ for the object obtained from $X'$ by removing marks from its labels (i.e., replacing back all $a_1,\dots,a_k$ by $a$).
	We also define $\cl(\Ll')=\setof{\cl(T')}{T'\in\Ll'}$ for a set of \kl{trees} $\Ll'$.
	In particular, when $P'$ is obtained by adding marks to all \kl{letters} in $P$, we have $\cl(P')=P$.
	We have the following claim:

	\begin{claimOK}\label{claim:diversified}
		$\sem{\cl(X')}=\cl(\sem{X'})$ whenever $X$ is an \STRE, a \kl{pure product}, a \kl{context}, or an \kl{iterator} over the extended alphabet.
	\end{claimOK}

	\begin{proof}
		The claim follows by a straightforward induction on the size of $X'$, because the $\cl(\cdot)$ operation commutes with all constructs appearing in the definition of $\sem{\cdot}$,
		namely $\cl(\Ll'_1\cup\Ll'_2)=\cl(\Ll'_1)\cup\cl(\Ll'_2)$,
		$\cl(\downwardclosure{\Ll'})=\downwardclosure{(\cl(\Ll'))}$,
		and $\cl(\replace{\Cc'}{\Ll'})=\replace{\cl(\Cc')}{\cl(\Ll')}$. 
	\end{proof}

	Following Zetzsche \cite{Zetzsche:ICALP:2015},
	we can reduce computation of the \kl{downward closure} to \SUP:

	\begin{lemmaOK}
		\label{thm:SUP}
	 	Let $\Class$ be a class of languages of finite \kl{trees} closed under \kl{linear FTT transductions}.
	 	One can compute a finite tree automaton recognizing the \kl{downward closure} of a given language from $\Class$ if and only if
		\SUP is decidable for $\Class$.
	\end{lemmaOK}

	\begin{proof}
		If \kl{downward closures} are computable,
		then one can compute a finite tree automaton recognizing $\downwardclosure \Ll$.
		Moreover, given a (\kl{diversified}) \kl{pure product} $P$, one can easily construct a finite tree automation recognizing $\sem P$,
		following the inductive definition of $\sem{\cdot}$.
		Having these two automata, one can check whether $\sem P\subseteq \downwardclosure \Ll$:
		language inclusion for finite tree automata is decidable~\cite[Section 1.7]{tata2007}.

		For the other direction, assume that \SUP is decidable for $\Class$ and let $\Ll \in \Class$.
		The \kl{downward closure} $\Ll^d:=\downwardclosure \Ll$ is effectively in $\Class$ since it can be obtained as a \kl{linear FTT transduction} of $\Ll$ by \cref{fact:dwclosure-is-FTT}.
		Thus, it is enough to compute a finite tree automaton recognizing the \kl{downward-closed} language $\Ll^d$.
		Furthermore, by \cref{lem:pure-products}
		$\Ll^d$ equals $\sem{P_1}\cup\dots\cup{}\sem{P_k}$ for some (unknown) \kl{pure products} $P_1,\dots,P_k$,
		and a finite tree automaton recognizing $\sem{P_1}\cup\dots\cup\sem{P_k}$ can be easily computed out of $P_1,\dots,P_k$.
		In consequence, it suffices to guess these \kl{pure products} and check whether the equality $\Ll^d = \sem{P_1}\cup\dots\cup\sem{P_k}$ indeed holds.

		We start by showing how to decide whether $\Ll^d \subseteq \sem{P_1}\cup\dots\cup\sem{P_k}$.
		Firstly, $\Rr := \sem{P_1}\cup\dots\cup\sem{P_k}$ is (effectively) a regular language, and consequently its complement $\Rr^c$ is also regular.
		In consequence, $\Mm := \Ll^d \cap \Rr^c$ is effectively in $\Class$, because it can be obtained from $\Ll^d$
		by intersecting it with $\Rr^c$, which is a \kl{linear FTT transduction} by \cref{fact:reg-intersection-is-FTT}.
		Secondly, emptiness of any language $\Mm \in \Class$ is decidable by reducing to \SUP,
		since it suffices to apply to it the \kl{linear FTT} $\Aa$ that ignores the input and outputs all \kl{trees} of the form $\leta\applyto(\leta\applyto(\dots\applyto(\leta\applyto\lete)\dots))$
		(for some fixed \kl{letters} $\leta$ of \kl{rank} $1$ and $\lete$ of \kl{rank} $0$), and to compare the result with the \kl{diversified pure product} $P := (\leta\applyto\hole)^*.\opt \lete$.
		Indeed, $\Aa(\Mm)=\downwardclosure{\Aa(\Mm)}=\sem P$ if $\Mm$ is nonempty, and $\Aa(\Mm)=\downwardclosure{\Aa(\Mm)}=\emptyset$ if $\Mm$ is empty;
		thus, on the one hand, $\Aa(\Mm) \subseteq \sem P$ and, on the other hand,
		$\Mm$ is nonempty if and only if $\sem P \subseteq \downwardclosure{\Aa(\Mm)}$.

		For the other inclusion $\sem{P_1}\cup\dots\cup\sem{P_k}\subseteq \Ll^d$ we can equivalently check whether $\sem{P_i}\subseteq \Ll^d$ for all $i\in\set{1,\dots,k}$,
		which implies that it suffices to show decidability of checking the containment $\sem P \subseteq \Ll^d$ for a single \kl{pure product} $P$.
		We make $P$ \kl{diversified} by adding additional marks to \kl{letters} appearing in $P$.
		As described before \cref{claim:diversified},
		we achieve this by unambiguously replacing the $i$-th occurrence of \kl{letter} $a$ with the new \kl{letter} $a_i$.
		Let $P'$ be the resulting \kl{diversified pure product}.
		We also create a corresponding \kl{linear FTT} $\Aa$; 
		it replaces every label $a$ in the input \kl{tree} by an arbitrary \kl{letter} among the corresponding \kl{letters} $a_i$ (for every occurrence of $a$ we choose a mark~$i$ independently).
		We obtain $\sem{P}=\cl(\sem{P'})=\setof{\cl(T')}{T'\in\sem{P'}}$ by \cref{claim:diversified}, and $\Aa(\Ll^d)=\setof{T'}{\cl(T')\in\Ll^d}$ by definition,
		which gives us the following equivalence:
		\begin{align*}
			\sem{P} \subseteq \Ll^d \iff \sem {P'} \subseteq \Aa(\Ll^d).
		\end{align*}

		Thus, instead of checking whether $\sem P \subseteq \Ll^d$, we can check whether $\sem {P'} \subseteq \Aa(\Ll^d)$.
		Finally, we consider a language $\Ll':=\Aa(\Ll^d)\cap\sem {P'}$, which can be obtained from $\Aa(\Ll^d)$ by a \kl{linear FTT transduction} (cf.~\cref{fact:reg-intersection-is-FTT}),
		and thus which is effectively in $\Class$.
		Then, on the one hand, $\Ll'\subseteq \sem {P'}$ and, on the other hand, $\sem {P'}\subseteq\Aa(\Ll^d)$ if and only if $\sem {P'} \subseteq \Ll'$.
		Recall that $\Ll^d$ and $\sem {P'}$ are \kl{downward closed}.
		It does not matter whether we first remove some parts of a \kl{tree} and then we add marks to labels, or we first add marks to labels and then we remove some part of a \kl{tree},
		so $\Aa(\Ll^d)$ and $\Ll'$ are \kl{downward closed} as well (and hence $\downwardclosure{\Ll'}=\Ll'$).
		It follows that checking whether $\sem {P'} \subseteq \Ll'$ is an instance of SUP.
	\end{proof}
	
	\begin{remarkOK}\label{remark:why-sums}
		\kl{Pure products} for \kl{trees} correspond to expressions of the form $a_0^?A_1^*a_1^?\dots A_k^*a_k^?$ for words
		(where $A_i$ are sets of \kl{letters}).
		In \SUP for words simpler expressions of the form $b_1^* \dots b_k^*$ suffice.
		This is not possible for \kl{trees}:
		\begin{inparaenum}[(1)]
		\item expressions of the form $a^?\applyto P_1\applyto P_2$ cannot be removed since they are responsible for branching, and
		\item reducing the two \kl{contexts} in $((a\applyto P_1\applyto\hole)+(b\applyto P_2\applyto\hole))^*.P_3$ to a single one
		would require changing \kl{trees} of the form $a\applyto T_1\applyto(b\applyto T_2\applyto T_3)$ into \kl{trees} of the form $c\applyto T_1\applyto T_2\applyto T_3$,
		which is not a \kl{linear FTT transduction}.
		\end{inparaenum}
	\end{remarkOK}

\subsection{The \kl{diagonal problem for trees}}
\label{sec:dc:diagonal}

	In \SUP for words, instead of checking whether $a_1^* \dots a_k^* \subseteq \downwardclosure \Ll$,
	one can equivalently check whether, for each $n\in\Nat$,
	there is a word in $\downwardclosure\Ll\cap a_1^*\dots a_k^*$ containing at least $n$ occurrences of every \kl{letter} $a_i$, where $i\in\set{1,\dots,k}$.
\AP	The latter problem (for an arbitrary language $\Ll'$ in place of $\downwardclosure\Ll\cap a_1^*\dots a_k^*$) is known as the \emph{diagonal problem} for words.
\AP	In this section, we define an analogous \kl{diagonal problem for trees},
	and we show how to reduce \SUP to it.

\AP	Given a set of \kl{letters} $\Sigma$,
	we say that a language of finite \kl{trees} $\Ll$ is \intro{$\Sigma$-diagonal}
	if, for every $n \in \Nat$, there is a \kl{tree} $T \in \Ll$
	such that for every \kl{letter} $a\in\Sigma$ and every \kl{branch}~$B$ of~$T$ there are at least $n$ occurrences~$a$ in~$B$.
\AP	The \intro{diagonal problem} for a class $\Class$ of finite \kl{trees} asks, given a language $\Ll\in\Class$ and a set of \kl{letters} $\Sigma$, whether $\Ll$ is \kl{$\Sigma$-diagonal}.

\paragraph{Versatile trees.}

	Contrary to the case of words,
	the presence of sums in our expressions creates some complications in reducing from \SUP to the \kl{diagonal problem}.
	Namely, suppose that we want to check whether $\sem{((\leta\applyto\hole)+(\letb\applyto\hole))^*.\letc}\subseteq\downwardclosure\Ll$.
	This question is not equivalent to checking whether $\downwardclosure\Ll\cap\sem{((\leta\applyto\hole)+(\letb\applyto\hole))^*.\letc}$ contains \kl{trees} with arbitrarily many $\leta$ and $\letb$.
	Indeed, it is possible that $\downwardclosure\Ll$ contains \kl{trees} of the form $\leta\applyto(\leta\applyto(\dots\applyto(\leta\applyto(\letb\applyto(\letb\applyto(\dots\applyto(\letb\applyto\letc)\dots))))\dots))$
	with arbitrarily many $\leta$ and $\letb$,
	but this does not yet mean that it contains arbitrarily large \kl{trees} of the form $\leta\applyto(\letb\applyto(\leta\applyto(\letb\applyto(\dots\applyto(\leta\applyto(\letb\applyto\letc))\dots))))$.
	Denote the latter \kl{tree} with $n$ occurrences of $\leta$ by $T_n$;
	the original question is rather equivalent to checking whether $\downwardclosure\Ll\cap\setof{T_n}{n\in\Nat}$ contains \kl{trees} with arbitrarily many $\leta$ and $\letb$.
	This is the case, because every \kl{tree} in $\sem{((\leta\applyto\hole)+(\letb\applyto\hole))^*.\letc}$ can be \kl{embedded} in a large enough \kl{tree} $T_n$
	(e.g., $\letb\applyto(\letb\applyto(\leta\applyto\letc))$ \kl{embeds} in $T_3=\leta\applyto(\letb\applyto(\leta\applyto(\letb\applyto(\leta\applyto(\letb\applyto\letc)))))$).
	
\AP	We thus deal with sums by considering \kl{trees} like $T_n$, which we call \intro{versatile trees}.
	Intuitively, in order to obtain a \kl{versatile tree} of a \kl{pure product} $P$, for every sum $I=C_1+\dots+C_k$ in $P$
	we fix some order of the \kl{contexts} $C_1,\dots,C_k$, and we allow the \kl{contexts} to be appended only in this order.
\AP	Formally, the set $\intro[\ctree]{\ctree{P}}$ of \kl{versatile trees} of a \kl{pure product} $P$
	is defined by induction on the structure of $P$:
	\begin{align*}
		\ctree{I^*.P} &= \mathrlap{\bigcup_{n\in\Nat}\ctree{I}[\underbrace{(\ctree{I}\cup\set{\hole})[\dots[(\ctree{I}\cup\set{\hole})[}_n \ctree{P}]]\dots]],}\\
		\ctree{a^?\applyto P_1\applyto\dots\applyto P_r} &= \ctree{a\applyto P_1\applyto\dots\applyto P_r},\\
		\ctree{C_1 + \dots + C_k} &= \ctree{C_1}[\dots[\ctree{C_k}]\dots],\\
		\ctree{a\applyto P_{\hole,1}\applyto\dots\applyto P_{\hole,r}} &= \setof{a\applyto T_1\applyto\dots\applyto T_r}{\forall i\,.\,T_i\in\ctree{P_{\hole,i}}},\\
		\ctree{\hole} &= \set\hole.
	\end{align*}
	For example, if $I = (\leta\applyto S_1\applyto\hole\applyto\hole) + (\letb\applyto\hole\applyto S_2)$,
	then $\ctree I = \set{\leta\applyto S_1\applyto(\letb\applyto\hole\applyto S_2)\applyto(\letb\applyto\hole\applyto S_2)}$;
	in particular, we have $\letb\applyto(\leta\applyto S_1\applyto\hole\applyto\hole)\applyto S_2\not\in\ctree I$.
	Notice that the \kl{roots} of all \kl{trees} in $\ctree{P}$ have the same label; denote this label by $\rootLab(P)$.

\paragraph{From \SUP to the diagonal problem.}
	
\AP	Assuming that $P$ is \kl{diversified}, for a number $n\in\Nat$ we say that a \kl{tree} $T$ is \intro{$n$-large with respect to $P$}
	if, for every subexpression of $P$ of the form $I^*.P'$,
	above every occurrence of $\rootLab(P')$ in the \kl{tree} $T$ there are at least $n$ ancestors labeled by $\rootLab(I^*.P')$.
	In other words, for $T\in\ctree{P}$ this means that in $T$ every \kl{context} appearing in $P$ was appended at least $n$ times, on all branches where it was possible to append it.
	Clearly $\ctree P \subseteq \sem P$.
	On the other hand, every \kl{tree} from $\sem{P}$ can be \kl{embedded} into every \kl{versatile tree} which is large enough.
	We thus obtain the following lemma:

	\begin{lemmaOK}\label{lem:large-gives-all}
		For every \kl{diversified pure product} $P$, and for every sequence of \kl{trees} $T_1,T_2,\dots\in\ctree{P}$ such that every $T_n$ is \kl{$n$-large},
		\begin{align*}\tag*{\QED}
			\downwardclosure{\setof{T_n}{n\in\Nat}} = \sem{P}.
		\end{align*}
	\end{lemmaOK}
	
	Using \kl{versatile trees} we can reduce \SUP to the \kl{diagonal problem}:

	\begin{lemmaOK}\label{lem:test-for-large}
		Let $\Class$ be a class of languages of finite \kl{trees} closed under \kl{linear FTT transductions}.
		\SUP for $\Class$ reduces to the \kl{diagonal problem} for $\Class$.
	\end{lemmaOK}
	
	\begin{proof}
		In an instance of \SUP we are given a \kl{diversified pure product} $P$ and a language $\Ll\in\Class$.
		Consider the language of \kl{trees} $\Ll' = \downwardclosure{\Ll}\cap\ctree{P}$.
		Clearly $\ctree{P}$ is regular, so $\Ll'\in\Class$
		by \cref{fact:dwclosure-is-FTT,fact:reg-intersection-is-FTT}.
		The following claim is a direct consequence of \cref{lem:large-gives-all}:
		
		\begin{claimOK}
			$\sem P \subseteq \downwardclosure \Ll$ if and only if
			for every $n\in\Nat$ there is a \kl{tree} in $\Ll'$ that is \kl{$n$-large with respect to $P$}.
		\QED\end{claimOK}
		
		We have reduced to a problem which is very similar to the \kl{diagonal problem},
		except that we should put no requirement on the number of occurrences of $\rootLab(I^*.P')$
		for \kl{branches} not containing an occurrence of $\rootLab(P')$.
		In order to fix this, let $\Ll''$ be the set of \kl{trees} $T''$
		obtained from some \kl{tree} $T'$ of $\Ll'$ by the following procedure:
		whenever a \kl{branch} of $T'$ does not contain an occurrence of $\rootLab(P')$,
		then the \kl{leaf} finishing this \kl{branch} can be replaced by an arbitrarily large \kl{tree} with internal \kl{nodes} labeled by $\rootLab(I^*.P')$.
		Let $\Sigma$ be the set of \kl{root} labels of the form $\rootLab(I^*.P')$
		for every subexpression $I^*.P'$ of $P$.
		The following claim is a direct consequence of the definition:

		\begin{claimOK}
			$\Ll''$ is \kl{$\Sigma$-diagonal} if and only if
			for every $n\in\Nat$ there is a \kl{tree} in $\Ll'$ which is \kl{$n$-large with respect to $P$}.
		\QED\end{claimOK}
		
		The operation mapping $\Ll'$ to $\Ll''$ can be realized as a \kl{linear FTT transduction},
		and thus $\Ll''\in\Class$.
		This completes the reduction from \SUP to the \kl{diagonal problem}.
	\end{proof}

	\begin{remarkOK}\label{rem:diagonal}
		Another formulation of the \kl{diagonal problem} for languages of finite \kl{trees}
		\cite{diagonal-safe,ClementeParysSalvatiWalukiewicz:LICS:2016,Parys:FSTTCS:2017}
		requires that, for every $n\in\Nat$, there is a \kl{tree} $T\in\Ll$
		containing at least $n$ occurrences of every \kl{letter} $a\in\Sigma$
		(not necessarily on the same \kl{branch}, unlike in our case).
		Such a formulation of the \kl{diagonal problem} seems too weak to compute \kl{downward closures} for languages of finite \kl{trees}.
	\end{remarkOK}

	The main result of this section, \cref{thm:reduction},
	stating that the \kl{downward closure} computation reduces to the \kl{diagonal problem},
	follows at once from \cref{thm:SUP} and \cref{lem:test-for-large} above.


\section{Languages of safe recursion schemes}
\label{sec:downward-closure}

	In the previous section, we have developed a general machinery allowing one to compute \kl{downward closures}
	for classes of languages of finite \kl{trees} closed under \kl{linear FTT transductions}.
	In this section, we apply this machinery to the particular case of \kl(scheme){languages recognized} by \kl{safe recursion schemes}.
	The following is the main theorem of this section:
	
	\begin{theoremOK}\label{thm:main-2}
		Finite tree automata recognizing \kl{downward closures} of languages of finite \kl{trees} \kl(scheme){recognized} by \kl{safe recursion schemes} are computable.
	\end{theoremOK}
	
	In order to prove the theorem we need to recall a formalism
	necessary to express the \kl{diagonal problem} in logic.

\paragraph{Cost logics.}

\AP
	\intro{Cost monadic logic} (\intro{CMSO}) was introduced by Colcombet~\cite{cost-functions} as a quantitative extension of monadic second-order logic (\intro{MSO}).
	As usual, the logic can be defined over any relational structure, but we restrict our attention to \kl{CMSO} over \kl{trees}. 
\AP 
	In addition to \intro{first-order variables} ranging over \kl{nodes} of a \kl{tree} and
	\reintro{monadic second-order variables} (also called \intro{set variables}) ranging over sets of \kl{nodes}, 
	\kl{CMSO} uses a single additional variable $\intro{\varN}$, called the \intro{numeric variable}, which ranges over $\Nat$.
\AP	The atomic formulas in \kl{CMSO} are those from \kl{MSO} (the membership relation $x \in X$ and relations $a(x, x_1, \dots, x_r)$ 
	asserting that $a \in\Alphabet$ of \kl{rank} $r$ is the label at \kl{node} $x$ with \kl{children} $x_1, \dots, x_r$ from left to right), 
\AP	as well as a new predicate $\intro{\size X <\varN}$, where $X$ is any \kl{set variable} and $\kl{\varN}$ is the \kl{numeric variable}. 
	Arbitrary \kl{CMSO} formulas are built inductively by applying Boolean connectives and by quantifying (existentially or universally) over \kl[first-order variable]{first-order} or \kl{set variables}. 
	We require that predicates of the form $\kl{\size X <\varN}$ appear positively in the formula (i.e., within the scope of an even number of negations).
	We regard $\kl{\varN}$ as a parameter.
\AP	As usual, a \intro{sentence} is a formula without \kl[first-order variable]{first-order} or \kl[monadic variable]{monadic} \intro(logic){free variables};
	however, the parameter $\kl{\varN}$ is allowed to occur in a \kl{sentence}.
\AP	If we fix a value $n \in \Nat$ for $\kl{\varN}$, the semantics of $\kl{\size X <\varN}$
	is what one would expect:
	the predicate holds when $X$ has cardinality smaller than $n$. 
\AP	We say that a \kl{sentence} $\varphi$ \intro(logic){$n$-accepts} a \kl{tree} $T$ if
	it holds in $T$ when $n$ is used as a value of~$\kl{\varN}$;
\AP	it \intro(logic){accepts} $T$ if it \kl(logic){$n$-accepts} $T$ for some $n\in\Nat$.

\AP
	\intro{Weak cost monadic logic} (\reintro{WCMSO} for short) is the variant of \kl{CMSO} 
	where the second-order quantification is restricted to finite sets.
\AP
	Vanden Boom~\cite[Theorem 2]{weak-cost} proves that \kl{WCMSO} is effectively equivalent to a subclass of \kl{alternating B-automata}, called \intro{weak B-automata}.
	Thanks to \cref{thm:main-1}, we obtain the following corollary:

	\begin{corollaryOK}\label{cor:wcmso-to-baut}
		Given a \kl{WCMSO} formula $\varphi$ and a \kl{safe recursion scheme} $\Gg$,
		one can decide whether $\varphi$ \kl(logic){accepts} the \kl{tree generated by} $\Gg$.
	\QED\end{corollaryOK}

	\begin{remarkOK}
\AP		The same holds for a
		more expressive logic called
		\intro{quasi-weak cost monadic logic} (\intro{QWCMSO}) \cite{quasi-weak-logic},
		whose expressive power lies between \kl{WCMSO} and the \kl{CMSO}.
		Indeed, Blumensath et al.~\cite[Theorem 2]{quasi-weak-logic} prove that \kl{QWCMSO}
		is effectively equivalent to a subclass of \kl{alternating B-automata} called \intro{quasi-weak B-automata},
		and thus by \cref{thm:main-1} even model checking of \kl{safe recursion schemes} against \kl{QWCMSO} properties is decidable.
	\end{remarkOK}

\paragraph{Solving the diagonal problem.}

	\AP By \cref{thm:reduction}
	all we need to do in order to obtain \cref{thm:main-2} is to show that
	\begin{inparaenum}[(1)]
	\item the \kl{diagonal problem} is decidable for \kl(scheme){languages recognized} by \kl{safe recursion schemes}, and
	\item the class of these languages is effectively closed under \kl{linear FTT transductions}.
	\end{inparaenum}
	We start by proving the former:
	
	\begin{lemmaOK}\label{lem:diagonal-decidable}
		The \kl{diagonal problem} is decidable for the class of languages of finite \kl{trees} \kl(scheme){recognized} by \kl{safe recursion schemes}.
	\end{lemmaOK}
	
\AP	Recall that in the \kl{diagonal problem} we are given a \kl{safe recursion scheme} $\Gg$ and a set of \kl{letters} $\Sigma$,
	and we have to determine whether for every $n\in\Nat$ there is a \kl{tree} $T \in \lang{\Gg}$
	such that there are at least $n$ occurrences of every \kl{letter} $a\in\Sigma$ on every \kl{branch} of $T$
	(we say that such a \kl{tree} $T$ is \intro{$n$-large with respect to $\Sigma$}).
	In order to obtain decidability of this problem, given a set of \kl{letters} $\Sigma$,
	we write a \kl{WCMSO} \kl{sentence} $\varphi_\Sigma$ that \kl(logic){$n$-accepts} an (infinite) \kl{tree} $T$ if and only if no \kl{tree} in $\NT T$
	is \kl{$n$-large with respect to $\Sigma$}.
	Consequently, $\varphi_\Sigma$ \kl(logic){accepts} $T$ if for some $n$ no \kl{tree} in $\NT T$ is \kl{$n$-large with respect to $\Sigma$}, that is, if $\NT T$ is not \kl{$\Sigma$-diagonal}.
	Thus, in order to solve the \kl{diagonal problem}, it is enough to check whether $\varphi_\Sigma$ \kl(logic){accepts} $\BT{\Gg}$ (recall that $\lang{\Gg}$ is defined as $\NT{\BT{\Gg}}$),
	which is decidable by \cref{cor:wcmso-to-baut}.
	It remains to construct the aforementioned \kl{sentence} $\varphi_\Sigma$.

	First, observe that the process of producing a finite \kl{tree} \kl(scheme){recognized} by $\Gg$ from the infinite \kl{tree} $\BT \Gg$ \kl{generated by} $\Gg$ is expressible by a formula of \kl{WCMSO} 
	(actually, by a first-order formula):
	
	\begin{lemmaOK}
		There is a \kl{WCMSO} formula $\tree X$ that holds in a \kl{tree} $T$ if and only if $X$ is instantiated to a set of \kl{nodes} of a \kl{tree} $T'\in\NT T$, together with their $\ndd$-labeled ancestors.
	\end{lemmaOK}

	\begin{proof}
		The formula simply says that
		\begin{compactitem}
		\item	$X$ is finite,
		\item	the \kl{root} of the \kl{tree} belongs to $X$,
		\item	no \kl{node} $x\in X$ is $\bot$-labeled,
		\item	for every $\ndd$-labeled \kl{node} $x\in X$, exactly one among the \kl{children} of $x$ belongs to $X$,
		\item	for every \kl{node} $x\in X$ with label other than $\ndd$, all \kl{children} of $x$ belong to $X$, and
		\item	if $x\not\in X$, then no \kl{child} of $x$ belongs to $X$.
		\end{compactitem}
		All the above statements can easily be expressed in \kl{WCMSO}.
	\end{proof}
	
	Using $\tree X$ we construct the desired formula $\varphi_\Sigma$, and thus we finish the proof of \cref{lem:diagonal-decidable}:
	
	\begin{lemmaOK}
		Given a set of \kl{letters} $\Sigma$, one can compute a \kl{WCMSO} \kl{sentence} $\varphi_\Sigma$ that, for every $n\in\Nat$, 
		\kl(logic){$n$-accepts} a \kl{tree} $T$ if and only if no \kl{tree} in $\NT T$ is \kl{$n$-large with respect to $\Sigma$}.
	\end{lemmaOK}
	
	\proof
		We can reformulate the property as follows: for every \kl{tree} $T'\in \NT T$ there is a \kl{letter} $a\in\Sigma$, 
		and a \kl{leaf} $x$ that has less than $n$ $a$-labeled ancestors.
		This is expressed by the following formula of \kl{WCMSO} 
		(where $\mathsf{leaf}(x)$ states that the \kl{node} $x$ is a \kl{leaf}, $a(x)$ that $x$ has label $a$, and $z\leq x$ that $z$ is an ancestor of $x$, all being easily expressible):
		\begin{align*}
			\forall X.\Big(\!\tree X \to\!\bigvee_{a\in\Sigma}
				\exists x\exists Z.\big(x\in X\land\mathsf{leaf}(x)\land\forall z.(z\leq x\land a(z) \to z\in Z)\land|Z|<\varN\big)\!\Big).
			\tag*{\QED}
		\end{align*}

\paragraph{Closure under transductions.}

	Finally, we show closure under \kl{linear FTT transductions}, which allows us to apply the results of the previous section
	to \kl{safe recursion schemes}:
	
	\begin{lemmaOK}
		\label{thm:HORS:transd}
		The class of languages of finite \kl{trees} \kl(scheme){recognized} by \kl{safe recursion schemes}
		is effectively closed under \kl{linear FTT transductions}.
	\end{lemmaOK}

	Observe that \cref{thm:main-2} is a direct consequence of \cref{thm:reduction,lem:diagonal-decidable,thm:HORS:transd}.
	It thus remains to prove \cref{thm:HORS:transd}.
	A very similar result, albeit without the \kl(scheme){safety} assumption,
	has been proved by Clemente, Parys, Salvati, and Walukiewicz~\cite[Theorem 2.1]{ClementeParysSalvatiWalukiewicz:LICS:2016}:

	\begin{lemmaOK}\label{lem:unsafe-HORS:transd}
		The class of languages of finite \kl{trees} \kl(scheme){recognized} by \kl{recursion schemes} is effectively closed under \kl{linear FTT transductions}.
	\QED\end{lemmaOK}

	Notice that \cref{thm:HORS:transd} does not follow from \cref{lem:unsafe-HORS:transd},
	since we need to additionally show that applying a \kl{linear FTT transduction} to a language \kl(scheme){recognized} by a \kl{safe recursion scheme} preserves \kl(scheme){safety}.
	Essentially the same construction
	as in the proof of \cref{lem:unsafe-HORS:transd}~\cite[Appendix A]{DBLP:journals/corr/ClementePSW16}
	already achieves this, albeit some modifications are needed.
	We now argue how to modify the proof in three aspects:

	\begin{compactenum}
	\item	The proof uses the fact that \kl[recursion schemes]{higher-order recursion schemes} \emph{with states} (as introduced in that proof)
		are convertible to equivalent \kl[recursion schemes]{higher-order recursion schemes}
		by increasing the arity of \kl{nonterminals}.
		It is a simple observation that such a translation preserves \kl(scheme){safety}.

	\item	The proof of Clemente et al.~\cite[Appendix A]{DBLP:journals/corr/ClementePSW16} uses the notion of \emph{normalized} recursion schemes,
		wherein every \kl{rule} is assumed to be of the form
		\begin{align*}
			\Rr(X)=\lambda x_1\lamdots\lambda x_p. h\applyto (Y_1\applyto x_1\applyto \dots\applyto x_p) \applyto\dots\applyto (Y_r\applyto x_1\applyto \dots\applyto x_p),
		\end{align*}
		where $h$ is either a \kl(lambda){variable} $x_i$, or a \kl{nonterminal}, or a \kl{letter},
		and the $Y_j$'s are \kl{nonterminals}.
		This normal form is used only to simplify the presentation
		and is in no way essential.
		This is important since putting a recursion scheme in such a normal form does not preserve \kl(scheme){safety}.
		Indeed, a \kl{subterm} $Y_j\applyto x_1\applyto \dots\applyto x_p$ replaces some \kl{subterm} $M_j$ appearing originally in the \kl{rule} for $X$;
		if some \kl(lambda){variable} $x_i$ was not used in $M_j$, we could have $\ord(x_i)<\ord(M_j)$ (the latter equals $\ord(Y_j\applyto x_1\applyto \dots\applyto x_p)$),
		which violates \kl(lambda){safety} of the normalized \kl{rule}.
		On the other hand, by \kl(lambda){safety} we have $\ord(x_i)\geq\ord(M_j)$ if $x_i$ appeared in $M_j$.
		Therefore, we modify the definition of the normal form to allow removal of selected \kl(lambda){variables}, that is, to allow \kl{rules} of the form
		\begin{align*}
			\Rr(X)=\lambda x_1\lamdots\lambda x_p. h\applyto (Y_1\applyto x_{i_{1,1}}\applyto \dots\applyto x_{i_{1,k_1}}) \applyto\dots\applyto (Y_r\applyto x_{i_{r,1}}\applyto \dots\applyto x_{i_{r,{k_r}}}).
		\end{align*}
		By leaving in each \kl{subterm} $Y_j\applyto x_{i_{j,1}}\applyto \dots\applyto x_{i_{j,k_j}}$ only \kl(lambda){variables} used in the replaced \kl{subterm} $M_j$,
		we obtain a normalized recursion scheme that is \kl(scheme){safe}.

	\itemAP	The proof~\cite[Lemma A.3]{DBLP:journals/corr/ClementePSW16} uses also the \intro{MSO-reflection property} of recursion schemes.
		In order to define this property consider a \kl{tree} $T$ and an \kl{MSO} formula $\varphi(x)$ with one \kl(logic){free} \kl{first-order variable}.
		We define $T_\varphi$ to be the \kl{tree} obtained from $T$ by enhancing its labels:
		for every \kl{node} $v$ of $T$, we change its label from $a$ to $(a,b_{\varphi,v})$,
		where $b_{\varphi,v}\in\set{\mathbf{tt},\mathbf{ff}}$ says whether $\varphi(v)$ is true ($\mathbf{tt}$) or false ($\mathbf{ff}$) in $T$.
		The \kl{MSO-reflection property} says that given a recursion scheme $\Gg$ \kl{generating} a \kl{tree} $T$ and given an \kl{MSO} formula $\varphi(x)$
		one can compute a recursion scheme $\Gg'$ \kl{generating} the \kl{tree} $T_\varphi$.
		It was shown~\cite[Corollary 2]{reflection} that recursion schemes indeed have the \kl{MSO-reflection property}.
		
		While switching to \kl{safe recursion schemes} one needs a similar property, where both the input and the output recursion schemes are \kl(scheme){safe}
		(this way we have a stronger conclusion under stronger assumptions).
		It is a folklore result that such an \kl{MSO-reflection property} for \kl{safe recursion schemes} holds as well.
		Let us support this statement in three ways:
		\begin{compactitem}
		\item	It is remarked by Carayol and Serre~\cite[Remark 5]{reflection} that even a stronger property, called MSO-selection, holds for \kl{safe recursion schemes}.
		\item	To obtain a proof of the \kl{MSO-reflection property} for \kl{safe recursion schemes} one can take the original proof of this property for all schemes~\cite{reflection},
			and observe that the construction in this proof preserves \kl(scheme){safety}.
			The proof uses collapsible pushdown automata,
			where \kl(scheme){safety} corresponds to absence of collapse operations;
			it is thus enough to see that no collapse operations are introduced if no such operations were present on input.
		\item	Carayol and W\"ohrle~\cite{DBLP:conf/fsttcs/CarayolW03} prove
			that the class of \kl{trees} generated by deterministic higher-order pushdown automata
			is effectively closed under MSO-markings, which is essentially the same as \kl{MSO-reflection}.
			Although Carayol and W\"ohrle~\cite{DBLP:conf/fsttcs/CarayolW03} work with edge-labeled trees, it is a routine to transfer their results to our setting of node-labeled \kl{trees}
			(and to change MSO-markings into \kl{MSO-reflection}).
			Moreover, a \kl{tree} can be generated by a deterministic higher-order pushdown automaton
			if and only if it can be \kl{generated by} a \kl{safe recursion scheme}
			(see Knapik et al.~\cite[Theorems 5.1 and~5.3]{easy-trees}; note, however, that the authors use the word ``grammar'' for a recursion scheme).
			Thus, the result of Carayol and W\"ohrle~\cite{DBLP:conf/fsttcs/CarayolW03} implies the desired \kl{MSO-reflection property} for \kl{safe recursion schemes}.
		\end{compactitem}
	\end{compactenum}


\section{Conclusions}
\label{sec:conclusions}

A tantalising direction for further work is to drop the \kl(scheme){safety} assumption from \cref{thm:main-1},
that is, to establish decidability of the \kl{model-checking problem} against \kl{B-automata}
for \kl{trees generated by} (not necessarily safe) \kl{recursion schemes}.
We also leave open whether \kl{downward closures} are computable for this more expressive class.
Another direction for further work is to analyse the complexity of the considered \kl{diagonal problem}.
The related problem described in \cref{rem:diagonal} is $\EXP k$-complete for languages of finite \kl{trees}
\kl(scheme){recognized} by \kl{recursion schemes} of \kl(scheme){order} $k$~\cite{Parys:FSTTCS:2017},
and thus not harder than the nonemptiness problem \cite{Ong:LICS:2006}.
Does the same upper bound hold for the more general \kl{diagonal problem} that we consider in this paper?
Zetzsche~\cite{DBLP:conf/icalp/Zetzsche16} has shown that the \kl{downward closure} inclusion problem is $\coNEXP k$-hard
for languages of finite \kl{trees} \kl(scheme){recognized} by \kl{safe recursion schemes} of \kl(scheme){order} $k$.
Is it possible to obtain a matching upper bound?

\bibliographystyle{fundam}
\bibliography{bib}

\appendix
\section{Proof of Lemma~\ref{lem:new-scheme}}\label{appendix}

In this appendix, we provide a self-contained proof of \cref{lem:new-scheme}.
The proof in its essence comes from the papers of Knapik et al.~\cite{easy-trees,hyperalgebraic-trees}, up to some minor details.

\subsection{Preparatory steps}

\AP
A \kl{type} $\alpha_1\arr\dots\arr\alpha_k\arr\otyp$ is \intro(type){homogeneous} if $\ord(\alpha_1)\geq\dots\geq\ord(\alpha_k)$ and all $\alpha_1,\dots,\alpha_k$ are \kl(type){homogeneous}.
A \kl{recursion scheme} $\Gg=\tuple{\Alphabet,\Nn,X_0,\Rr}$ is \intro(scheme){homogeneous} if \kl{types of} all \kl{nonterminals} in $\Nn$ are \kl(type){homogeneous}.
Notice that then also the \kl{type of} every \kl{subterm} of $\Rr(X)$ is \kl(type){homogeneous}, for every \kl{nonterminal} $X\in\Nn$.
It is known that every (\kl(scheme){safe}) \kl{recursion scheme} can be made \kl(scheme){homogeneous}:

\begin{lemmaOK}[{\cite[Theorems 8 and 9]{homogeneous}}]\label{lem:homogeneous}
	For every \kl{safe recursion scheme} $\Gg$ one can construct a \kl(scheme){homogeneous} \kl{safe recursion scheme} $\Hh$ of the same \kl(scheme){order},
	such that $\BT{\Hh}=\BT{\Gg}$.
\QED\end{lemmaOK}

Thanks to \cref{lem:homogeneous} we may assume that the \kl{recursion scheme} $\Gg$ given in \cref{lem:new-scheme} is \kl(scheme){homogeneous}.
We remark that \kl(scheme){homogeneity} of $\Gg$ is not at all essential in the remainder of the proof; this assumption is just for technical convenience.
Namely, thanks to this assumption, the notion of \kl(lambda){order}-0 arguments coincides with the notion of arguments occurring after the last argument of positive \kl(lambda){order}.

It is also convenient to assume that every \kl{nonterminal} of positive \kl(lambda){order} takes some parameter of \kl(lambda){order} $0$.
Again, this can be achieved without loss of generality:

\begin{lemmaOK}
	For every \kl(scheme){homogeneous} \kl{safe recursion scheme} $\Gg$ one can construct a \kl(scheme){homogeneous} \kl{safe recursion scheme} $\Hh$ of the same \kl(scheme){order},
	such that $\BT{\Hh}=\BT{\Gg}$, and such that every \kl{nonterminal} of $\Hh$ having positive \kl(lambda){order} takes some parameter of \kl(lambda){order} $0$
	(i.e., there are no \kl{nonterminals} \kl{of type} $\alpha_1\arr\dots\arr\alpha_k\arr\otyp$ with $k\geq 1$ and $\ord(\alpha_k)\geq 1$).
\end{lemmaOK}


\begin{proof}
\AP
	We say that a \kl{type} is \intro{bad} if it is of the form $\alpha_1\arr\dots\arr\alpha_k\arr\otyp$ with $k\geq 1$ and $\ord(\alpha_k)\geq 1$.
	We add one additional \kl(lambda){order}-0 parameter to every \kl{lambda-term} of \kl{bad type}.
	More formally, recall that \kl{rules} of $\Gg$ are of the form $\Rr(X)=\lambda x_1\lamdots\lambda x_k.K$, where $K$ is an applicative term \kl{of type} $\otyp$.
	If the \kl{type of} $X$ is \kl{bad}, we replace this \kl{rule} by $\lambda x_1\lamdots\lambda x_k.\lambda \vary.M$ for a fresh \kl(lambda){variable} $\vary$ \kl{of type} $\otyp$.
	Note that in $\Lambda(\Gg)$ this inserts an additional \kl{lambda-binder} $\lambda \vary$
	between $\lambda x$ and $K$ in every \kl{subterm} of the form $\lambda x.K$ with $\ord(x)\geq 1$ and $\ord(K)=0$.
	Simultaneously, we replace every \kl{application} $K\applyto L$ with $\ord(L)\geq 1$ and $\ord(K\,L)=0$ by $K\applyto L\applyto\bot$:
	whenever the last argument is applied to a \kl{lambda-term} having a \kl{bad type}, we apply an additional \kl(lambda){order}-0 argument, which is chosen to be $\bot$
	(but can be any \kl{lambda-term of type} $\otyp$).
	This changes the \kl{types of} \kl{lambda-terms} as follows:
	every \kl{type} $\alpha=(\alpha_1\arr\dots\arr\alpha_k\arr\otyp)$ changes
	\begin{compactitem}
	\item	to $\alpha_1'\arr\dots\arr\alpha_k'\arr\otyp\arr\otyp$ if $\alpha$ was \kl{bad}, and
	\item	to $\alpha_1'\arr\dots\arr\alpha_k'\arr\otyp$ otherwise,
	\end{compactitem}
	where $\alpha_1',\dots,\alpha_k'$ are obtained by the same transformation applied to the \kl{types} $\alpha_1,\dots,\alpha_k$.
	It is tedious but straightforward to formally check that this way we obtain a valid \kl{recursion scheme} $\Hh$,
	and that it \kl{generates} the same \kl{tree} as $\Gg$.
\end{proof}

\subsection{Reification: Defining \texorpdfstring{$\reify\Gg$}{G*}}
\label{sec:reification}

\renewrobustcmd{\Xx}{{\kl[\Xx]{\mathcal{X}}}}
\knowledge\Xx{math notion}

\AP
Fix some \kl(scheme){normalizing} \kl(scheme){homogeneous} \kl{safe recursion scheme} $\Gg=\tuple{\Alphabet,\Nn,X_0,\Rr}$,
where every \kl{nonterminal} of positive \kl(lambda){order} takes some parameter of \kl(lambda){order} $0$.
Let $\intro*\Xx$ be the (finite) set of \kl(lambda){order}-$0$ \kl(lambda){variables} used for parameters in $\Gg$.

\AP
The \intro{maximal arity} of a \kl{type} $\alpha$, denoted $\intro*\mar(\alpha)$ is defined by induction:
\begin{align*}
	\mar(\alpha_1\arr\dots\arr\alpha_k\arr\otyp)=\max\big(\set{k}\cup\set{\mar(\alpha_i)\mid 1\leq i\leq k}\big).
\end{align*}
The \reintro{maximal arity} of a \kl{lambda-term} $M$, denoted $\mar(M)$, equals
\begin{align*}
	\mar(M)=\sup\set{\mar(\alpha)\mid\alpha\mbox{ is a \kl{type of} a \kl{subterm} of }M}.
\end{align*}
Finally, the \reintro{maximal arity} of a \kl{recursion scheme} $\Gg=\tuple{\Alphabet,\Nn,X_0,\Rr}$, denoted $\mar(\Gg)$, equals
\begin{align*}
	\mar(\Gg)=\max\set{\mar(X)\mid X\in\Nn}.
\end{align*}

Observe that $\mar(\Rr(X))\leq\mar(\Gg)$ for every \kl{nonterminal} $X$ of $\Gg$
(because the only \kl(lambda){variables} occurring in $\Rr(X)$ are \kl{nonterminals} of $\Gg$ and parameters of $X$).
It follows that $\mar(\Lambda(\Gg))\leq\mar(\Gg)$.

\AP
We say that $M$ is an \intro{input lambda-term} if
\begin{compactitem}
\item	$M$ uses \kl{letters} from the alphabet $\Alphabet$;
\item	all \kl(lambda){order}-$0$ \kl(lambda){variables} used in $M$, other than \kl{nonterminals} from $\Nn$, belong to $\Xx$;
\item	\kl{nonterminals} from $\Nn$ are not used in \kl{lambda-binders} in $M$;
\item	\kl{types of} all \kl{subterms} of $M$ are \kl(type){homogeneous};
\item	for every lambda-abstraction \kl{subterm} $\lambda x_1\lamdots\lambda x_k.K$ of $M$,
	where $k\geq 1$ and $K$ is not a lambda-abstraction,
	we have $\ord(x_k)=\ord(K)=0$;
\item	$\mar(M)\leq\mar(\Gg)$;
\itemAP	no \kl{subterm} of $M$ is an \intro{infinite application} $\cdots\applyto M_3\applyto M_2\applyto M_1$.
\end{compactitem}
Note that the above conditions are satisfied by $M=\Rr(X)$ for all \kl{nonterminals} $X\in\Nn$, as well as by $M=\Lambda(\Gg)$.
Moreover, every \kl{subterm} of an \kl{input lambda-term} is an \kl{input lambda-term}.

We additionally require that in a \kl{first-order} \kl{input lambda-term} $M$
every \kl(lambda){free variable} of $M$ belongs to $\Xx$ (i.e., no \kl{nonterminals}, even of \kl(lambda){order} $0$, may occur in $M$).

\AP
We define a function $\reify {(\_)}$, called \intro{reification};
it maps an \kl{input lambda-term} $M$ of a \kl{homogeneous type} $\alpha$
to a corresponding \kl{lambda-term} $\reify M$ of a \kl{homogeneous type} $\reify \alpha$,
using \kl{letters} from the alphabet $\Alphabet_\Xx$, defined in \cref{sec:model-checking}.
We also say that the \kl{lambda-term} $\reify M$
\intro{represents} the \kl{lambda-term} $M$.
Moreover, if $M$ is \kl{first-order},
then $\reify M$ is in fact a \kl{lambda-tree}, that is, it does not contain \kl(lambda){variables} nor \kl{lambda-binders} (cf.~\cref{lem:reification:basic}).

\AP
For every \kl{homogeneous type} $\alpha$, the \kl{type} $\intro*\reify \alpha$ is defined by induction on the structure of $\alpha$:
if $$\alpha \ =\ \alpha_1\arr\dots\arr\alpha_k\arr\otyp\arr\dots\arr\otyp\arr\otyp,$$
where $k=0$ or $\alpha_k\neq\otyp$,
then we take
$$\reify \alpha \ =\ \reify{\alpha_1}\arr\dots\arr\reify{\alpha_k}\arr\otyp.$$
In other words, \kl(lambda){order}-0 arguments are discarded
and the transformation is applied recursively to higher-order arguments.
For instance, $\reify \otyp = \otyp$,
$\reify {(\otyp \arr \otyp)} = \otyp$,
and $\reify {((\otyp \arr \otyp) \arr \otyp \arr \otyp)} = \otyp \arr \otyp$.
It is easy to see (by induction on the structure of $\alpha$) that $\ord(\reify \alpha)=\max(0,\ord(\alpha)-1)$.

We now define \kl{reification} of an \kl{input lambda-term} $M$.
First, to every \kl{nonterminal} $X\in\Nn$ \kl{of type} $\alpha$ we assign a unique \kl{nonterminal} $\reify X$ \kl{of type} $\reify \alpha$.
Likewise, to every \kl(lambda){variable} $x\not\in(\Xx\cup\Nn)$ \kl{of type} $\alpha$ we assign a unique \kl(lambda){variable} $\reify x$ \kl{of type} $\reify \alpha$.
Next, we proceed by coinduction on the structure of $M$:
\begin{compactenum}
\item	$\reify{(a)}=\tconst{a}$;
\item	$\reify{(X)}=\reify X$ if $X\in\Nn$ (i.e., the result of the $\reify {(\_)}$ operation for a \kl{nonterminal} $X$ is the \kl{nonterminal} denoted $\reify X$);
\item	$\reify{(x)}=\tvar{x}$ if $x\not\in\Nn$ and $\ord(x)=0$ (i.e., if $x\in\Xx$);
\item	$\reify{(x)}=\reify x$ if $x\not\in\Nn$ and $\ord(x)>0$;
\item	$\reify{(\lambda x.K)}=\tlambda{x}\applyto \reify K$ if $\ord(x)=0$ (i.e., if $x\in\Xx$);
\item	$\reify{(\lambda x.K)} = \lambda \reify x.\reify K$ if $\ord(x)>0$;
\item	$\reify{(K\applyto L)}=\tat\applyto \reify K\applyto \reify L$ if $\ord(L)=0$;
\item	$\reify{(K\applyto L)} =\reify K \applyto \reify L$ if $\ord(L)>0$.
\end{compactenum}
Observe (by coinduction) that if $M$ \kl{has type} $\alpha$ then $\reify M$ is a \kl{lambda-term of type} $\reify \alpha$.
This is immediate in Cases 2, 3, and 4.
In Case 1, a \kl{letter} $a$ \kl{has type} of the form $\alpha=(\otyp\arr\dots\arr\otyp\arr\otyp)$, while $\tconst{a}$ \kl{has type} $\reify \alpha=\otyp$.
In Case 5 we use the assumption that the \kl[type of]{type} $\otyp\arr\beta$ of $\lambda x.K$ is \kl(type){homogeneous}, which implies $\ord(K)=\ord(\beta)\leq 1$, that is, $\ord(\reify K)=0$.
Likewise in Case 7 we use the assumption that the \kl[type of]{type} $\otyp\arr\beta$ of $K$ is \kl(type){homogeneous}, which implies $\ord(K)=\ord(\otyp\arr\beta)\leq 1$, that is, $\ord(\reify K)=0$.
This is necessary, because the \kl{lambda-terms} $\tlambda{x}\applyto \reify K$ and $\tat\applyto \reify K\applyto \reify L$ make sense only if $\ord(\reify K)=0$.
In Cases 6 and 8 we observe that $\reify {(\beta\arr\gamma)}=\reify \beta \arr \reify \gamma$ if $\ord(\beta)>0$.

There is one delicate point of the definition above.
Namely, \kl{lambda-terms} are usually identified up to renaming bound \kl(lambda){variables} (\kl{alpha-conversion}).
The result of the \kl{reification operation} $\reify {(\_)}$, however, depends on particular names given to bound \kl(lambda){order}-$0$ \kl(lambda){variables}
(these names become written explicitly in the \kl{letters} (constants) $\tvar{x}$ and $\tlambda{x}$ ).
Thus, it is understood that no implicit renaming of bound \kl(lambda){order}-0 \kl(lambda){variables} is performed for \kl{lambda-terms} to which the $\reify {(\_)}$ operation is going to be applied.

When starting from a \kl{lambda-term} that is \kl{first-order}
(defined on \cpageref{page:first-order}),
we can see that Cases 2, 4, 6, and 8 can never occur.
In such a circumstance,
\kl{reification} produces a \kl{lambda-tree}.

\begin{lemmaOK}
	\label{lem:reification:basic}
	If an \kl{input lambda-term} $M$ is \kl{first-order} then $\reify M$ is a \kl{lambda-tree}.
\QED\end{lemmaOK}

Using the \kl{reification operation} $\reify {(\_)}$ for \kl{lambda-terms},
we can define the resulting \kl{recursion scheme} $\reify\Gg$:
we take
\begin{align}
	\label[equality]{eq:reified scheme}
	\reify \Gg=\tuple{\Alphabet_\Xx,\reify \Nn,\reify {X_0},\reify \Rr},
\end{align}
where $\reify \Nn=\set{\reify X\mid X\in\Nn}$
and $\reify \Rr(\reify X)=\reify {(\Rr(X))}$ for all $X\in\Nn$.

It is easy to see that $\reify \Gg$ is of \kl(scheme){order} $m-1$
if $\Gg$ was of \kl(scheme){order} $m\geq 1$: the \kl(lambda){order} of every \kl{nonterminal}, if positive, drops by one.
Let us now observe that $\reify\Gg$ is \kl(scheme){safe}:

\begin{lemma}
	\label{lem:represent:safety}
	If $\Gg$ is \kl(scheme){safe},
	then $\reify \Gg$ is \kl(scheme){safe}.
\end{lemma}

\begin{proof}
	Recall that, by definition, $\Gg$ is \kl(scheme){safe} when the \kl{lambda-term} $\Lambda(\Gg)$ is \kl(lambda){safe}; likewise for $\reify\Gg$ and $\Lambda(\reify\Gg)$.
	First, it is easy to see that $\Lambda(\reify \Gg)=\reify {(\Lambda(\Gg))}$.
	In order to ensure that $\reify \Gg$ is \kl(scheme){safe}, we thus need to ensure that every \kl{subterm} of $\reify {(\Lambda(\Gg))}$
	occurring in argument position of some \kl{application} is \kl{superficially safe}.
	Subterms occurring in argument position of an \kl{application} in $\reify {(\Lambda(\Gg))}$ are
	\begin{compactitem}
	\item	$\reify K$ in $\tlambda{x}\applyto \reify K$,
	\item 	$\reify K$ and $\reify L$ in $\tat\applyto \reify K\applyto \reify L$, and
	\item	$\reify L$ in $\reify K\applyto \reify L$.
	\end{compactitem}
	In the first two cases, the \kl{subterms} are of \kl(lambda){order} $0$, so they are automatically \kl{superficially safe}.
	In the last case, $L$ occurs in argument position of the \kl{application} $K\applyto L$ in $\Lambda(\Gg)$, which means that $L$ is \kl{superficially safe};
	we have $\ord(x)\geq\ord(L)$ for every \kl(lambda){free variable} $x$ of $L$.
	Every \kl(lambda){free variable} of $\reify L$ is of the form $\reify x$ for $x$ being a \kl(lambda){free variable} of $L$;
	we then have $\ord(\reify x)=\max(0,\ord(x)-1)\geq\max(0,\ord(L)-1)=\ord(\reify L)$, as required.
\end{proof}

The relation between $\reify \Gg$ and $\Gg$ is described by the following lemma:

\begin{lemmaOK}\label{lem:1}
	There exists a \kl{closed} \kl{first-order} \kl{input lambda-term} $M$ \kl{of type} $\otyp$ such that
	\begin{align*}
		\BT{\reify \Gg} = \reify M
			&&\text{and}&&
				\BT M = \BT\Gg.
	\end{align*}
\end{lemmaOK}

Notice that there is at most one \kl{lambda-term} $M$
such that $\BT{\reify \Gg}=\reify M$.
So the \kl{lambda-term} $M$ in the lemma above is in fact unique.
The remaining part of this subsection is devoted to the proof of \cref{lem:1}.

First, let us see that \kl(lambda){safety} is preserved by \kl{beta-reductions}:\footnote{%
	Blum and Ong~\cite{safe} write that safety is not preserved by arbitrary \kl{beta-reductions}, only by \kl{beta-reductions} of a special kind.
	Note, however, that they consider a slightly different definition of \kl(lambda){safe} \kl{lambda-terms} (leading to the same definition of \kl{safe recursion schemes}).}

\begin{lemmaOK}\label{lem:safety-preserved}
	If $M$ is \kl(lambda){safe} and $M\reducesto N$, then $N$ is \kl(lambda){safe}.
\end{lemmaOK}

\begin{proof}
	Before starting, let us state two inductive properties of \kl(lambda){safety}, following directly from its definition:
	\begin{compactenum}
	\item[\quad Inductive Property 1:] $\lambda x.P$ is \kl(lambda){safe} if, and only if, $P$ is \kl(lambda){safe};
	\item[\quad Inductive Property 2:] $P\applyto Q$ is \kl(lambda){safe} if, and only if, $P$ and $Q$ are \kl(lambda){safe}, and $Q$ is \kl{superficially safe}.
	\end{compactenum}
	Next, let us prove an auxiliary claim concerning \kl{substitution}:
	
	\begin{claimOK}\label{claim:safety-preserved}
		If $K$ and $L$ are \kl(lambda){safe}, and $L$ is \kl{superficially safe}, then $\subst K L x$ is \kl(lambda){safe}.
	\end{claimOK}
	
	We prove this claim by structural coinduction.
	When $x$ is not \kl(lambda){free} in $K$, or when $K=x$, then $\subst K L x$ equals $K$ or $L$, respectively, and the thesis holds by assumption.
	When $K=\lambda y.P$, the thesis is an immediate consequence of the coinduction hypothesis and Inductive Property 1.
	The only remaining case is that $K=P\applyto Q$.
	By the coinduction hypothesis we obtain that $\subst P L x$ and $\subst Q L x$ are safe.
	To conclude, we also need to know that $\subst Q L x$ is \kl{superficially safe} (cf.~Inductive Property 2).
	If $x$ is not \kl(lambda){free} in $Q$,
	then this is immediate: $\subst Q L x=Q$
	and the latter is superficially safe by assumption.
	Otherwise, every \kl(lambda){free variable} $y$ of $\subst Q L x$ is \kl(lambda){free} either in $Q$ or in $L$.
	In the former case we simply have that $\ord(y)\geq\ord(Q)=\ord(\subst Q L x)$, because $Q$ is \kl{superficially safe};
	in the latter case we have $\ord(y)\geq\ord(L)=\ord(x)\geq\ord(Q)={}\ord(\subst Q L x)$, because $L$ and $Q$ are \kl{superficially safe} and $x$ is \kl(lambda){free} in $Q$.
	It follows that $\subst Q L x$ is \kl{superficially safe} and thus $\subst K L x$ is \kl(lambda){safe}, as required.
	
	We can now come back to the proof of \cref{lem:safety-preserved},
	which we perform by induction on the depth of the considered redex.
	The base case, when $M=(\lambda x.K)\applyto L$ and $N=\subst K L x$, is provided directly by \cref{claim:safety-preserved}
	(note that $L$ occurs in argument position in $M$, so it is \kl{superficially safe} by \kl(lambda){safety} of $M$).
	For the induction step, we have three cases:
	\begin{compactenum}
	\item $M=\lambda x.P$ and $N=\lambda x.P'$, where $P\reducesto P'$;
	\item $M=P\applyto Q$ and $N=P'\applyto Q$, where $P\reducesto P'$;
	\item $M=P\applyto Q$ and $N=P\applyto Q'$, where $Q\reducesto Q'$.
	\end{compactenum}
	In the first two cases, we simply use the induction hypothesis for $P\reducesto P'$.
	In the last case, we also need to observe that $Q'$ is \kl{superficially safe},
	which holds because $Q$ is \kl{superficially safe}, and every \kl(lambda){free variable} of $Q'$ is \kl(lambda){free} already in $Q$.
\end{proof}

\paragraph{Beta-reductions of positive order.}

\AP
It is useful to consider \intro{beta-reductions of positive order}, denoted ``$\reducesplus$'':
We have $M\reducesplus N$ if $N$ is obtained from $M$ by replacing some \kl{subterm} $(\lambda x.K)\applyto L$ thereof with $\subst K L x$,
where we additionally require $\ord(x)\geq 1$.

We use the ``$\reducesplus$'' relation only for \kl(lambda){safe} \kl{lambda-terms},
and when writing $M\reducesplus N$ we implicitly assume that names of bound \kl(lambda){order}-0 \kl(lambda){variables} do not change.
Note that if $M$ is \kl(lambda){safe}, then the argument $L$ of the redex is {superficially safe}.
It follows that every \kl(lambda){free variable} $y$ of $L$ satisfies $\ord(y)\geq{}\ord(L)=\ord(x)\geq 1$.
In other words, $L$ has no \kl(lambda){free variables} of \kl(lambda){order} $0$.
Thus there is no danger that these \kl(lambda){variables} will conflict with bound \kl(lambda){order}-0 \kl(lambda){variables} in $K$;
there is never the need to rename bound \kl(lambda){order}-0 \kl(lambda){variables}.

Recall that the $(\_)^\bullet$ operation is defined only for \kl{input lambda-terms}, as defined at the beginning of the subsection.
With the above assumption in hand, we have that if $M$ is a \kl(lambda){safe} \kl{input lambda-term} and $M\reducesplus N$, then $N$ is also an \kl{input lambda-term}
(most importantly, all \kl(lambda){order}-$0$ \kl(lambda){variables} used in $N$, other than \kl{nonterminals}, belong to $\Xx$);
in particular, it makes sense to write $N^\bullet$.

Our next lemma connects the ``$\reducesplus$'' relation
with the ``$\reducesto$'' relation and \kl{reification}:

\begin{lemmaOK}\label{lem:rei:beta}
	Let $M$ be a \kl(lambda){safe} \kl{input lambda-term}.
	\begin{compactenum}
	\item	If $M\reducesplus N$, then $\reify M\reducesto\reify N$.
	\item	If $\reify M \reducesto O$, then $O = \reify N$ for a \kl{lambda-term} $N$ such that $M \reducesplus N$.
	\end{compactenum}
\end{lemmaOK}

In order to prove \cref{lem:rei:beta}, we first need to see that higher-order \kl{substitution} commutes with \kl{reification}:

\begin{lemmaOK}\label{lem:rei:subst}
	For every \kl{input lambda-term} of the form $\subst K L x$, where $\ord(x)\geq 1$,
	we have
	\begin{align*}
		\reify{(\subst K L x)} = \subst {\reify K} {\reify L} {\reify x}.
	\end{align*}
\end{lemmaOK}

\begin{proof}
	Follows directly from the definition of \kl{reification}.
\end{proof}

In \cref{lem:rei:subst} we implicitly assume that the \kl{substitution} $\subst K L x$ does not change names of bound \kl(lambda){order}-0 \kl(lambda){variables} in $K$.
As already said, it is never needed to rename them if $L$ does not have \kl(lambda)[free variables]{free \kl(lambda)[order]{order-0} variables}, that is,
when $(\lambda x.K)\applyto L$ is a \kl{subterm} of a \kl(lambda){safe} \kl{lambda-term}.
Note also that \cref{lem:rei:subst} does not make sense when $x$ has \kl(lambda){order} zero, because in that case there is no \kl(lambda){variable} $\reify x$
(the \kl(lambda){variable} $x$ is \kl{reified} to $\tconst x$, which is a \kl{letter}).

\begin{proofof}{\cref{lem:rei:beta}}
	For the first item, suppose that $N$ is obtained from $M$ by replacing a redex $(\lambda x.K)\applyto L$ with $\subst K L x$, where $\ord(x)=\ord(L)\geq 1$.
	Then in $\reify M$ we have a redex
	\begin{align*}
		\reify{((\lambda x.K)\applyto L)}=(\lambda \reify x.\reify K)\applyto \reify L,
	\end{align*}
	which \kl{beta-reduces} to $\subst {\reify K} {\reify L} {\reify x}=\reify{(\subst K L x)}$ (equality by \cref{lem:rei:subst}).
	We thus have $\reify M\reducesto\reify N$.
	
	For the second item, observe that the definition of \kl{reification} produces a \kl{lambda-binder} only in Case 6,
	and an \kl{application} whose operator is not a \kl{letter} only in Case 8.
	Thus the redex of $\reify M$ reduced in $\reify M\reducesto O$
	is necessarily of the form
	\begin{align}
		\label[formula]{eq:redex:base}
		\reify{((\lambda x.K)\applyto L)}
		= (\lambda \reify x.\reify K) \applyto \reify L
		\reducesto \subst {\reify K} {\reify L} {\reify x}
		= \reify{(\subst K L x)},
	\end{align}
	where $\ord(x)=\ord(L)\geq 1$,
	and where the second equality follows from \cref{lem:rei:subst}.
	Let $N$ be obtained from $M$ by reducing (the corresponding occurrence of) $(\lambda x.K)\applyto L$ to $\subst K L x$,
	and thus $M \reducesplus N$.
	A structural induction on the \kl{subterms} of $O$
	(the base case being provided by \cref{eq:redex:base})
	shows $O=\reify N$,
	as required.
\end{proofof}

One of consequences of \cref{lem:rei:beta} is the Church-Rosser property for $\reducesplus$:

\begin{lemmaOK}\label{lem:church-rosser}
	If $M\reducesplus^*N_1$ and $M\reducesplus^*N_2$ for a \kl(lambda){safe} \kl{input lambda-term} $M$,
	then $N_1\reducesplus^* P$ and $N_2\reducesplus^*P$ for some \kl{lambda-term} $P$.
\end{lemmaOK}

\begin{proof}
	By Item 1 of \cref{lem:rei:beta} (and using also \cref{lem:safety-preserved} to ensure that \kl{lambda-terms} under consideration are \kl(lambda){safe})
	we have $\reify M\reducesto^*\reify{N_1}$ and $\reify M\reducesto^*\reify{N_2}$.
	The Church-Rosser property for $\reducesto$ gives us a \kl{lambda-term} $O$ such that $\reify {N_1}\reducesto^*O$ and $\reify{N_2}\reducesto^*O$.
	Then, by Item 2 of \cref{lem:rei:beta} (and again by \cref{lem:safety-preserved})
	we obtain \kl{lambda-terms} $P_1$ and $P_2$ such that $O=\reify{P_1}=\reify{P_2}$, and $N_1\reducesplus^* P_1$ and $N_2\reducesplus^*P_2$.
	Observing that the \kl{reification operation} $\reify{(\_)}$ is injective, we actually have $P_1=P_2$, so this \kl{lambda-term} can be taken as $P$ in the thesis.
\end{proof}

Let $M$ be a (possibly infinite) \kl(lambda){safe} \kl{input lambda-term} of \kl(lambda){order} at most $1$ (we mean here the \kl(type){order} of the \kl{type of} $M$;
\kl{subterms} of $M$ may have higher \kl(lambda){order})
such that all \kl(lambda){free variables} thereof belong to $\Xx$.
\AP
We define the \kl{first-order} lambda-term obtained as the \intro{limit} of applying $\reducesplus$ reductions to $M$, denoted $\limbp M$,
analogously to how $\BT P$ is defined as the limit of applying the $\reducesto$ reductions to a \kl{closed} \kl{lambda-term} $P$ \kl{of type} $\otyp$.
The definition is coinductive:
\begin{compactitem}
\item	if $M\reducesplus^* a$ (for a \kl{letter} $a$), then $\limbp M=a$,
\item	if $M\reducesplus^* x$ (for a \kl(lambda){variable} $x\in\Xx$), then $\limbp M=x$,
\item	if $M\reducesplus^* \lambda x.N$ with $x\in\Xx$, then $\limbp M=\lambda x.(\limbp N)$, and
\item	if $M\reducesplus^* K\applyto L$ with $\ord(L)=0$, then $\limbp M=(\limbp K)\applyto(\limbp L)$.
\end{compactitem}
Clearly $\limbp M$ is a \kl{first-order} \kl{input lambda-term} of the same \kl[of type]{type} as $M$.


Observe that the above definition covers all possibilities (i.e., some of the above conditions holds for every $M$).
To this end consider all possible forms of $M$.
If $M=a$ or $M=K\applyto L$ with $\ord(L)=0$, we have the first or the last case of the definition, respectively.
If $M=x$, then $x\in\Xx$ by the assumption that all \kl(lambda){free variables} of $M$ belong to $\Xx$; we have the second case.
If $M=\lambda x.N$, then $\ord(x)=0$ by the assumption that $\ord(M)\leq 1$, hence $x\in\Xx$ (because $M$ is an \kl{input lambda-term});
we have the third case.
The only remaining case is that $M$ is an \kl{application} with an argument of positive \kl(lambda){order}.
Because $M$ is an \kl{input lambda-term}, it cannot be an \kl{infinite application}.
Thus, $M$ can be written as $H\applyto M_1\applyto\dots\applyto M_r$, where $H$ is not an \kl{application}, $r\geq 1$, and $\ord(M_r)\geq 1$.
Then $H$ cannot be a \kl{letter} (arguments of a \kl{letter} are all \kl{of type} $\otyp$) nor a \kl(lambda){variable} (all \kl(lambda){free variables} of $M$ are \kl{of type} $\otyp$, because they belong to $\Xx$);
$H$ has to start with a sequence of \kl{lambda-binders}: $H=\lambda x_1\lamdots\lambda x_k.K$, where $k\geq 1$ and $K$ does not start with a \kl{lambda-binder}. 
One of the assumptions for being an \kl{input lambda-term} implies that $\ord(x_k)=\ord(K)=0$.
Then necessarily $k\geq r$
(each of the provided arguments corresponds to some \kl{lambda-binder}),
and $\ord(x_r)={}\ord(M_r)\geq 1$ implies that $k>r$.
Moreover, because $M$ is of \kl(lambda){order} (at most) $1$, the \kl(lambda){variables} $x_{r+1},\dots,x_k$ are \kl{of type} $\otyp$.
On the other hand, $x_1,\dots,x_r$ are of \kl(lambda){order} at least $1$, by \kl(type){homogeneity}.
Thus $M\reducesplus^*\lambda x_{r+1}\lamdots\lambda x_k.K[M_1/x_1,\dots,M_r/x_r]$; we obtain the third case.

Moreover, thanks to \cref{lem:church-rosser}, the resulting \kl{lambda-term} $\limbp M$ is uniquely defined.

We now use \cref{lem:rei:beta} to show
a kind of commutativity property between \kl{reification} and \kl{Böhm trees}:

\begin{lemmaOK}\label{lem:reify-bt}
	Let $M$ be a \kl(lambda){safe} \kl{input lambda-term} of \kl(lambda){order} at most $1$,
	all \kl(lambda){free variables} of which belong to $\Xx$.
	Then $\BT{\reify M}=\reify{(\limbp M)}$.
\end{lemmaOK}

\proof 
	We proceed by coinduction.
	At every step we use \cref{lem:rei:beta} (and \cref{lem:safety-preserved} to obtain \kl(lambda){safety} of intermediate \kl{lambda-terms})
	to deduce $\reify M\reducesto^*\reify N$ from $M\reducesplus^* N$.
	According to the definition of $\limbp M$ we have four cases:
	\begin{compactitem}
	\item	If $M\reducesplus^* a$, then $\reify M\reducesto^*\reify{a}=\overline a$, so $\reify{(\limbp M)}=\reify a=\overline a=\BT{\reify M}$.
	\item	If $M\reducesplus^* x$ with $\ord(x)=0$, then $\reify M\reducesto^*\reify{x}=\overline x$, so $\reify{(\limbp M)}=\reify x=\overline x=\BT{\reify M}$.
	\item	If $M\reducesplus^* \lambda x.N$ with $\ord(x)=0$,
		then $\reify M\reducesto^*\reify{(\lambda x.N)}=\tlambda{x}\applyto N$,
		so $\reify{(\limbp M)}
			=\reify{(\lambda x.(\limbp N))}
			=\tlambda{x}\applyto\reify{(\limbp N)}
			=\tlambda{x}\applyto(\BT{\reify N})
			=\BT{\reify M}$,
		where the third equality is by the coinductive hypothesis.
	\item	If $M\reducesplus^* K\applyto L$ with $\ord(L)=0$,
		then $\reify M\reducesto^*\reify{(K\applyto L)}=\tat\applyto\reify K\applyto\reify L$,
		so
		\begin{align*}
		\reify{(\limbp M)}
			&=\reify{((\limbp K)\applyto(\limbp L))}
			=\tat\applyto\reify{(\limbp K)}\applyto\reify{(\limbp L)}\\
			&=\tat\applyto(\BT{\reify K})\applyto(\BT{\reify L})
			=\BT{\reify M},
		\end{align*}
		where the third equality is by the coinductive hypothesis.
	\QED\end{compactitem}
\endtrivlist

The other important property of $\limbp{\cdot}$
is that all higher-order reductions can be performed first,
followed by all (necessarily) order-zero reductions.
This is formally stated in the next lemma:%
\footnote{
	We remark that \cref{lem:pos-middle} can be generalized to say that $\BT{\limbp M}=\BT M$ for any \kl{closed} \kl{input lambda-term} $M$ of \kl(lambda){order} 0, not necessarily for $M=\Lambda(\Gg)$.
	The lemma can even be further generalized to say that $\BT N=\BT M$ whenever $N$ is obtained as an (appropriately defined) limit of applying any finite or infinite sequence of \kl{beta-reductions} to $M$.
	Nevertheless, we prove only the specific statement written above---in \cref{lem:order-0-cut} we explicitly use the fact that the \kl{lambda-term} is of the form $\Lambda(\Gg)$.
}

\begin{lemmaOK}\label{lem:pos-middle}
	$\BT{\limbp{\Lambda(\Gg)}}=\BT\Gg$.
\end{lemmaOK}

Before proving \cref{lem:pos-middle},
let us see how \cref{lem:1} follows from \cref{lem:reify-bt,lem:pos-middle}:

\begin{proofof}{\cref{lem:1}}
	We take $N = \Lambda(\Gg)$ and $M=\limbp{N}$.
	It is easy to check that $M$ is a \kl{closed} \kl{first-order} \kl{input lambda-term} \kl{of type} $\otyp$.
	We have $\BT{\reify\Gg}=\BT{\Lambda(\reify\Gg)}=\BT{\reify N}=\reify{(\limbp N)}$ by \cref{lem:reify-bt},
	and $\BT M=\BT{\limbp N}=\BT N=\BT{\Gg}$ by \cref{lem:pos-middle}.
\end{proofof}

It remains to prove \cref{lem:pos-middle}.
Our proof strategy is to show that the two \kl{Böhm trees} mentioned in the lemma are equal by showing that they agree on every finite prefix.
To this end, we have to define finite \kl{cuts} of a \kl{lambda-term}.

\paragraph{Finite \kl{cuts}.}

\AP
For every \kl{type} $\alpha$ let us fix a fresh \kl(lambda){variable} $\varx_\bot^\alpha$ \kl{of type} $\alpha$, not occurring anywhere in $\Lambda(\Gg)$, and called a \intro{cut variable}.
We say that $F$ is a \intro{cut} of $M$ if $F$ is obtained from $M$ by replacing some of its \kl{subterms} with \kl{cut variables} (of appropriate \kl[of type]{type}).
For example, $\lambda\vary.\varx_\bot^{\otyp\arr\otyp}$ is a \kl{cut} of $\lambda\vary.\lambda\varz.\leta\applyto\vary\applyto\varz$:
we have replaced the \kl{subterm} $\lambda\varz.\leta\applyto\vary\applyto\varz$
\kl{of type} $\otyp\arr\otyp$ with the \kl(lambda){variable} $\varx_\bot^{\otyp\arr\otyp}$.
We are particularly interested in finite \kl{cuts}, that is, \kl{cuts} that are finite \kl{lambda-terms}.

\AP
We say that a \kl{cut} $F$ is an \intro{order-0 cut} if the only \kl{cut variable} occurring in $F$ is $\varx_\bot^{\otyp}$ (i.e., only \kl{subterms} \kl{of type} $\otyp$ are cut off).
We have the following nice property of the \kl{lambda-term} $\Lambda(\Gg)$:

\begin{lemmaOK}\label{lem:order-0-cut}
	For every finite \kl{cut} $F$ of $\Lambda(\Gg)$ there exists a finite \kl{order-0 cut} $F_0$ of $\Lambda(\Gg)$ such that $F$ is a \kl{cut} of $F_0$.
\end{lemmaOK}

\begin{proof}
	Consider a \kl{subterm} of $\Lambda(\Gg)$ that was replaced by $\varx_\bot^\alpha$.
	It is necessarily of the form $K[M_1/X_1,\dots,\allowbreak M_k/X_k]$, where $K$ is a \kl{subterm} of $\Rr(X)$ for some \kl{nonterminal} $X$, hence $K$ is finite.
	Every \kl{lambda-term} $M_i$, substituted for the \kl{nonterminal} $X_i$,
	is obtained by further substituting \kl{lambda-terms} in $\Rr(X_i)$, hence it is of the form $\lambda x_{i,1}\lamdots\lambda x_{i,n_i}.K_i$, where $K_i$ is \kl{of type} $\otyp$.
	Instead of cutting off the whole $K[M_1/X_1,\dots,M_k/X_k]$, we can rather cut off at every occurrence of $K_i$.
	Our \kl{cut} remains finite, but all \kl{cut variables} are \kl{of type} $\otyp$.
\end{proof}

The next lemma says that the relation of being a \kl{cut}
is a simulation with respect to ``$\reducesto$'' and ``$\reducesplus$'' reductions.


\begin{lemmaOK}\label{lem:cut-beta}
	Let $F$ be a \kl{cut} of $M$.
	\begin{compactenum}
		\item If $F\reducesto G$, then $G$ is a \kl{cut} of a \kl{lambda-term} $N$ such that $M\reducesto N$.
		\item Likewise, if $F\reducesplus G$, then $G$ is a \kl{cut} of a \kl{lambda-term} $N$ such that $M\reducesplus N$.
	\end{compactenum}
\end{lemmaOK}

\begin{proof}
	We just reduce the redex of $M$ whose \kl{cut} was reduced in $F\reducesto G$ (in $F\reducesplus G$, respectively).
	It is easy to check that $G$ is indeed a \kl{cut} of the resulting \kl{lambda-term} $N$.
\end{proof}

\AP
We also need to state formally in which sense a \kl{lambda-term} agrees with a finite prefix of a \kl{tree}.
Let $n\in\Nat$, let $M$ be a \kl{lambda-term}, and let $T$ be a \kl{tree}.
We define when $M$ \intro[agrees]{agrees with $T$ up to level $n$}, by induction on $n$:
\begin{compactitem}
\item	every $M$ agrees with every $T$ up to level $0$;
\item	$M$ agrees with $T$ up to level $n+1$ if $M=a\applyto M_1\applyto\dots\applyto M_r$, $T=a\applyto T_1\applyto\dots\applyto T_r$,
	and $M_i$ agrees with $T_i$ up to level $n$, for every $i\in\set{1,\dots,r}$.
\end{compactitem}
The next lemma says that every finite prefix of $\BT M$ depends only on some finite prefix of $M$:

\begin{lemmaOK}\label{lem:finite-is-enough}
	Let $M$ be a \kl{closed} \kl(scheme){normalizing} \kl{lambda-term of type} $\otyp$.
	For every $n\in\Nat$ there exists a finite \kl{cut} $F$ of $M$, and a \kl{lambda-term} $G$ such that $F\reducesto^*G$ and $G$ \kl[agrees]{agrees with $\BT M$ up to level $n$}.
\QED\end{lemmaOK}

We skip the proof of \cref{lem:finite-is-enough},
which is a standard fact.
A very similar lemma is shown for instance in Parys~\cite[Lemma 4.2]{Parys:FSTTCS:2017}.
In \cref{lem:finite-is-enough} it is important that $M$ is \kl(lambda){normalizing}, so that every \kl{node} of $\BT M$ is created after finitely many reductions from $M$.
When $M=\Lambda(\Gg)$, we can strengthen \cref{lem:finite-is-enough} as follows:

\begin{lemmaOK}\label{lem:finite-order-0-is-enough}
	For every $n\in\Nat$ there exists a finite \kl{order-0 cut} $F_0$ of $\Lambda(\Gg)$,
	and a \kl{lambda-term} $G_0$ such that $F_0\reducesto^*G_0$ and $G_0$ \kl[agrees]{agrees with $\BT\Gg$ up to level $n$}.
\end{lemmaOK}

\begin{proof}
	First, from \cref{lem:finite-is-enough} we obtain a finite \kl{cut} $F$ of $\Lambda(\Gg)$,
	and a \kl{lambda-term} $G$ such that $F\reducesto^*G$ and $G$ \kl[agrees]{agrees with $\BT\Gg=\BT{\Lambda(\Gg)}$ up to level $n$}.
	It is not necessarily an \kl{order-0 cut}, but by \cref{lem:order-0-cut} we can extend it to a finite \kl{order-0 cut} $F_0$ (such that $F$ is a \kl{cut} of $F_0$).
	Then, by \cref{lem:cut-beta} we know that $G$ is a \kl{cut} of some $G_0$ such that $F_0\reducesto^*G_0$.
	It is easy to see that if $G$ \kl[agrees]{agrees with} some \kl{tree} (in particular, with $\BT\Gg$) \kl[agrees]{up to some level $n$},
	and $G$ is a \kl{cut} of $G_0$, then also $G_0$ \kl[agrees]{agrees with this \kl{tree} up to the same level $n$}.
\end{proof}

\begin{lemmaOK}\label{lem:equality-by-cut}
	Let $H$ be a \kl{cut} of two \kl{trees}, $T_1$ and $T_2$.
	If $H$ \kl[agrees]{agrees with $T_1$ up to some level $n$}, then both $H$ and $T_1$ \kl[agrees]{agree with $T_2$ up to level $n$}.
\end{lemmaOK}

\begin{proof}
	Straightforward: if $H$ \kl[agrees]{agrees with $T_1$ up to some level $n$}, then \kl{cut variables} may appear in $H$ only below this level.
\end{proof}

Recall that a \kl{lambda-term} is in beta-normal form
if it does not contain any redex.

\begin{lemmaOK}\label{lem:cut-bt}
	Let $H$ be a finite \kl{order-0 cut} of a \kl{closed} \kl{lambda-term} $M$ \kl{of type} $\otyp$.
	If $H$ is in beta-normal form, then $H$ is also a \kl{cut} of $\BT M$.
\end{lemmaOK}

\begin{proof}
	By induction on the size of $H$.
	Let us write $H=H_0\applyto H_1\applyto\dots\applyto H_r$, where $H_0$ is not an \kl{application}.
	If $H_0$ is a \kl{letter} $a$, then $M=a\applyto M_1\applyto\dots\applyto M_r$, where for every $i\in\set{1,\dots,r}$ the \kl{lambda-term} $M_i$ is \kl{closed} and \kl{of type} $\otyp$,
	and $H_i$ is a finite \kl{order-0 cut} of $M_i$, and is in beta-normal form.
	By the induction hypothesis, every $H_i$ is also a \kl{cut} of $\BT{M_i}$, which gives the thesis due to $\BT{M}=a\applyto(\BT{M_1})\applyto\dots\applyto(\BT{M_r})$.
	If $H_0$ is a \kl(lambda){variable}, then necessarily $H_0=\varx_\bot^{\otyp}$ (because $M$ is \kl{closed}) and $r=0$; $\varx_\bot^{\otyp}$ is a \kl{cut} of every \kl{lambda-term}.
	Finally, if $H_0$ is a lambda-abstraction, then necessarily $r\geq 1$ (because the \kl{type of} the whole term $H$ is $\otyp$), which contradicts the assumption that $H$ is in beta-normal form.
\end{proof}

\begin{lemmaOK}\label{lem:cut-limbp}
	Let $M$ be a \kl(lambda){safe} \kl{input lambda-term} $M$ of \kl(lambda){order} at most $1$ such that all \kl(lambda){free variables} thereof belong to $\Xx$,
	and let $G$ be a finite \kl{order-0 cut} of $M$.
	If no ``$\reducesplus$'' reduction can be executed from $G$, then $G$ is also a \kl{cut} of $\limbp M$.
\end{lemmaOK}

\begin{proof}
	By induction on the size of $G$.
	If $G=\varx_\bot^{\otyp}$, then it is a \kl{cut} of every \kl{lambda-term}.
	If $G$ is a \kl(lambda){variable} $x$ other than $\varx_\bot^{\otyp}$, but necessarily from $\Xx$ (by assumption), then also $M=x=\limbp M$, and the thesis is clear.
	Likewise, if $G$ is a \kl{letter} $a$, then also $M=a=\limbp M$, and the thesis is clear.

	Suppose that $G=\lambda x.G'$.
	We have $M=\lambda x.M'$, where $G'$ is a finite \kl{order-0 cut} of $M'$.
	Then necessarily $x\in\Xx$ (because $M$ is an \kl{input lambda-term} and $\ord(M)\leq 1$).
	The induction hypothesis can be applied to $G'$ and $M'$, implying that $G'$ is a \kl{cut} of $\limbp{M'}$.
	Then $G$ is a \kl{cut} of $\limbp M={}\lambda x.(\limbp{M'})$.

	Next, suppose that $G=G_0\applyto G_1$ with $\ord(G_1)=0$.
	Then $M=M_0\applyto M_1$, where $G_0$ and $G_1$ are finite \kl{order-0 cuts} of $M_0$ and $M_1$, respectively.
	The induction hypothesis implies that $G_0$ and $G_1$ are also \kl{cuts} of $\limbp{M_0}$ and $\limbp{M_1}$, respectively.
	Then $G$ is a \kl{cut} of $\limbp M=(\limbp{M_1})\applyto(\limbp{M_2})$.
	
	Finally, suppose that $G=G_0\applyto G_1\applyto\dots\applyto G_r$, where $G_0$ is not an \kl{application}, $r\geq 1$, and $\ord(G_r)\geq 1$.
	Note that $G_0$ cannot be a \kl{letter} nor a \kl(lambda){variable} (\kl{of type} $\otyp$, by assumption), because they do not take arguments of positive \kl(lambda){order}.
	So $G_0$ is a lambda-abstraction.
	But $\ord(G_1)\geq\ord(G_r)\geq 1$ by \kl(type){homogeneity}, which means that ``$\reducesplus$'' can be applied to the redex $G_0\applyto G_1$, contrary to the assumption;
	thus this case is actually impossible.
\end{proof}

\begin{proofof}{\cref{lem:pos-middle}}
	In order to prove that $\BT{\limbp{\Lambda(\Gg)}}=\BT\Gg$,
	it is enough to prove that $\BT{\limbp{\Lambda(\Gg)}}$ \kl[agrees]{agrees with $\BT\Gg$ up to every level $n\in\Nat$}.
	Fix some $n\in\Nat$, and consider a finite \kl{order-0 cut} $F$ of $\Lambda(\Gg)$ such that $F\reducesto^*H$ and $H$ \kl[agrees]{agrees with $\BT\Gg$ up to level $n$}.
	The \kl{cut} $F$ exists by \cref{lem:finite-order-0-is-enough}.
	Observe also that if $H$ \kl[agrees]{agrees with $\BT\Gg$ up to level $n$},
	and $H\reducesto H'$, then $H'$ also \kl[agrees]{agrees with $\BT\Gg$ up to level $n$}.
	Recall that finite \kl[type]{simply-typed} \kl{lambda-terms} are strongly normalizing,
	which in particular means that no infinite sequence of \kl{beta-reductions} can start in $H$.
	We can thus assume from this point on, without loss of generality, that $H$ is in beta-normal form.
	
	Let also $G$ be a \kl{lambda-term} such that $F\reducesplus^*G$, but no further ``$\reducesplus$'' reductions can be executed from $G$
	(i.e., $G$ is in $\reducesplus$-normal form).
	Using strong normalization again, we have that $G\reducesto^*H$.
	Recall that $F$ is a (finite, order-0) \kl{cut} of $\Lambda(\Gg)$.
	Due to $F\reducesto^*H$, by \cref{lem:cut-beta} we know that $H$ is a \kl{cut} of a \kl{lambda-term} $M$ such that $\Lambda(\Gg)\reducesto^* M$;
	then \cref{lem:cut-bt} implies that $H$ is also a \kl{cut} of $\BT M=\BT{\Lambda(\Gg)}=\BT\Gg$.
	Likewise,
	due to $F\reducesplus^*G$, by \cref{lem:cut-beta} we know that $G$ is a (finite, order-0) \kl{cut} of a \kl{lambda-term} $P$ such that $\Lambda(\Gg)\reducesplus^* P$;
	then \cref{lem:cut-limbp} implies that $G$ is also a \kl{cut} of $\limbp P=\limbp{\Lambda(\Gg)}$.
	Having this, and due to $G\reducesto^*H$, by \cref{lem:cut-beta} we know that $H$ is a \kl{cut} of a \kl{lambda-term} $Q$ such that $\limbp{\Lambda(\Gg)}\reducesto^* Q$;
	then \cref{lem:cut-bt} implies that $H$ is also a \kl{cut} of $\BT Q=\BT{\limbp{\Lambda(\Gg)}}$.

	We thus know that $H$ is a \kl{cut} of both $\BT\Gg$ and $\BT{\limbp{\Lambda(\Gg)}}$, and that it \kl[agrees]{agrees with $\BT\Gg$ up to level $n$}.
	In such a situation \cref{lem:equality-by-cut} implies that the two \kl{trees} \kl[agrees]{agree up to level $n$}, as required.
\end{proofof}

\subsection{From the \kl{B\"ohm tree} to the \kl{derived tree}}

We have already defined a \kl{safe recursion scheme} $\reify\Gg$, being of \kl(scheme){order} smaller by one than the \kl(scheme){order} of $\Gg$,
and such that
\begin{align}\label[equalities]{eq1}
	\BT{\reify\Gg}=\reify M && \mbox{and} && \BT{M}=\BT{\Gg}
\end{align}
for some \kl{closed} \kl{first-order} \kl{input lambda-term} $M$ \kl{of type} $\otyp$ (cf.~\cref{lem:1}).
For \cref{lem:new-scheme} we rather need the equality
\begin{align}\label[equality]{eq2}
	\semdt{\Xx,\mar(\Gg)}{\BT {\reify\Gg}} = \BT \Gg.
\end{align}
Because $\Gg$ is \kl(scheme){normalizing}, $M$ is \kl(lambda){normalizing} as well
(recall that $\Gg$, resp., $M$, is \kl(lambda){normalizing} if $\BT{\Gg}$, resp., $\BT{M}$, does not contain the special \kl{letter} $\bot$).
Thus, \cref{eq2} follows immediately from \cref{eq1} and from the following lemma:

\begin{lemmaOK}
	\label{lem:reify:sucessor:BT}
	Let $M$ be a \kl{closed} \kl(lambda){normalizing} \kl{first-order} \kl{input lambda-term} \kl{of type} $\otyp$,
	and let $s=\mathit{mar}(\Gg)$.
	Then $\reify M$ is a \kl{lambda-tree} and moreover
	\begin{align*}
		\semdt{\Xx, s}{\reify M} = \BT M.
	\end{align*}
\end{lemmaOK}

While proving \cref{lem:reify:sucessor:BT},
we identify a \kl{node} with a finite sequence of numbers from $\set{1, 2, \dots}$,
which denote directions when going down the tree
($1$ for the first \kl{child}, $2$ for the second \kl{child}, and so on).
Thus, $\varepsilon$ is the \kl{root},
and the $i$-th \kl{child} of a \kl{node} $v \in \set{1, 2, \dots}^*$ is the \kl{node} $v\cdot i$.

\AP
It is convenient to consider a more restrictive notion of beta\nobreakdash-reduction,
namely head {beta\nobreakdash-reduction}.
We say that $M$ \intro{head beta-reduces} to $N$, written $M \headreducesto N$,
if $M$ can be written as
\begin{align*}
	M \;=\; (\lambda x.K)\applyto L\applyto L_1\applyto\dots\applyto L_j&& \text{(for some $j\geq 0$)},
\end{align*}
and $N = \subst K L x \applyto L_1\applyto\dots\applyto L_j$.%
\footnote{
	The usual definition of head beta-reduction allows additionally a sequence of \kl{lambda-binders} outside the two \kl{lambda-terms},
	i.e., $M$ is of the form
	$\lambda y_1\lamdots y_k. ((\lambda x.K)\applyto L\applyto L_1\applyto\dots\applyto L_j)$.
	In our case, we consider head beta-reductions only for \kl{lambda-terms of type} $\otyp$,
	and thus such a sequence of \kl{lambda-binders} does not exist, i.e., $k = 0$.	
}
When writing $M \headreducesto N$, we implicitly assume that names of bound \kl(lambda){variables} do not change.
Note that when $M$ as above is \kl{closed},
then $L$ is \kl{closed} as well,
and thus indeed there is no need to rename bound \kl(lambda){variables} in $K$ while performing \kl{head beta-reductions}.
The following is a known fact (c.f.~\cite[Paragraph 11.4.7,~``Standardization theorem'']{Barendregt:1984},
where it is attributed to Curry and Feys \cite{CurryFeys:1958}):

\begin{lemmaOK}[Standardization theorem]\label{standardization}
	The \kl{Böhm tree} can be constructed using only \kl{head beta-reductions}
	(instead of arbitrary \kl[beta-reductions]{beta-re\-duc\-tions}).
	In other words, for every \kl{closed} \kl(lambda){normalizing} \kl{lambda-term} $M$ \kl{of type} $\otyp$ we have
	\begin{align*}
		\BT M=a\applyto(\BT{M_1})\applyto\dots\applyto(\BT{M_r})
	\end{align*}
	for some \kl{lambda-term} $N=a\applyto M_1\applyto\dots\applyto M_r$ such that $M\headreducesto^* N$.
\end{lemmaOK}

The next lemma states that \kl{derived trees} are invariant under \kl{head beta-reductions}:

\begin{lemmaOK}
	\label{lem:reify:beta:derived}
	If $M \headreducesto N$,
	where $M,N$ are \kl{closed} \kl{first-order} \kl{input lambda-terms} \kl{of type} $\otyp$
	and $\reify N$ is \kl(lambda-tree){normalizing},%
	\footnote{%
		The lemma holds also when $\reify N$ is not \kl(lambda-tree){normalizing},
		but then some additional arguments are needed in the proof.
		In the following,
		we need only the version when $\reify N$ is \kl(lambda-tree){normalizing},
		for which we provide an easier argument.
	}
	then $\semdt{\Xx, s}{\reify M} = \semdt{\Xx, s}{\reify N}$.
\end{lemmaOK}

Before proving \cref{lem:reify:beta:derived}
we show immediately how it is used in the proof of \cref{lem:reify:sucessor:BT}:

\begin{proofof}{\cref{lem:reify:sucessor:BT}}
	We have already proved in \cref{lem:reification:basic}
	that \kl{reification} for a \kl{first-order} \kl{input lambda-term} results in a \kl{lambda-tree}.
	In order to prove that $\semdt{\Xx, s}{\reify M} = \BT M$, we proceed by coinduction on the \kl{Böhm tree}.
	By \cref{standardization}, we have
	\begin{align*}
		\BT M = a \applyto (\BT {M_1}) \applyto \dots\applyto(\BT {M_r})
	\end{align*}
	for some \kl{lambda-term} $N = a \applyto M_1 \applyto \dots\applyto M_r$ such that $M \headreducesto^* N$.
	By the definition of \kl{reification}, we have
	\begin{align*}
		\reify N = \tat\applyto(\dots\applyto(\tat \applyto (\tat \applyto \tconst a \applyto \reify {M_1}) \applyto \reify {M_2})\applyto\dots)\applyto\reify{M_r}.
	\end{align*}
	Let us now compute the \kl{derived tree} $\semdt{\Xx, s}{\reify N}$.
	Following the definition for several \kl{successor steps},
	we arrive at
	\begin{align*}
		\semdt{\Xx, s}{\reify N} = \semdt{\Xx, s}{\reify N,\Rdown,\varepsilon} = a \applyto
			\semdt{\Xx, s}{\reify N,\Rup_1,1^r} \applyto\dots\applyto
			\semdt{\Xx, s}{\reify N,\Rup_r,1^r},
	\end{align*}
	where $1^r$ is the \kl{node} labeled by $\overline a$
	(i.e., the \kl{node} reached by going $r$ times left from the \kl{root}).
	Performing a few more \kl{successor steps} from $(\Rup_i,1^r)$ we see, for every $i\in\set{1,\dots,r}$,
	that
	\begin{align*}
		\semdt{\Xx, s}{\reify N,\Rup_i,1^r}=\semdt{\Xx, s}{\reify N,\Rdown,1^{r-i}\cdot 2},
	\end{align*}
	where $1^{r-i}\cdot 2$ is the \kl{root} of the \kl{subtree} $\reify M_i$
	(i.e., the \kl{node} reached by going $r-i$ times left and then one time right from the \kl{root}).
	By the coinductive assumption applied to $M_i$
	we have $\semdt{\Xx, s}{\reify {M_i}} = \BT {M_i}$.
	%
	By assumption these \kl{trees} do not contain the special \kl{letter} $\bot$ (i.e., $M_i$ is \kl(lambda){normalizing}),
	so the sequence of \kl{successors} used to define $\semdt{\Xx, s}{\reify {M_i}}$ never tries to go up from the \kl{root} of $\reify {M_i}$.
	It follows that
	\begin{align*}
		\semdt{\Xx, s}{\reify N,\Rdown,1^{r-i}\cdot\nobreak 2}=\semdt{\Xx, s}{\reify {M_i}}.
	\end{align*}
	Putting the pieces together yields
	\begin{align*}
		\semdt{\Xx, s}{\reify N}
			= a \applyto
			\semdt{\Xx, s}{\reify {M_1}} \applyto\dots\applyto
			\semdt{\Xx, s}{\reify {M_r}}
			= a \applyto (\BT {M_1})\applyto\dots \applyto (\BT {M_r})
			= \BT M.
	\end{align*}
	In particular, we now know that the \kl{lambda-tree} $\reify N$ is \kl(lambda-tree){normalizing}.
	Recalling that $M \headreducesto^* N$,
	we can conclude with the equality $\semdt{\Xx, s}{\reify M} = \semdt{\Xx, s}{\reify N}$
	obtained by a repeated use of \cref{lem:reify:beta:derived}.
	More precisely, consider the sequence of \kl{head beta-reductions}
	\begin{align*}
		M = N_0 \headreducesto N_1 \headreducesto \cdots \headreducesto N_k = N
	\end{align*}
	leading from $M$ to $N$.
	We can prove by induction on $i \in \set{0, \dots, k}$
	that 
	$\semdt{\Xx, s}{\reify N_{k-i}} = \semdt{\Xx, s}{\reify N}$.
	The base case $i = 0$ holds trivially.
	For the inductive case $i > 0$,
	we apply \cref{lem:reify:beta:derived} to $N_{k-i}$ and $N_{k-i+1}$; 
	we know that $\reify N_{k-i+1}$ is \kl(lambda-tree){normalizing} due to the induction hypothesis $\semdt{\Xx, s}{\reify N_{k-i+1}} = {}\semdt{\Xx, s}{\reify N}$.
	The required equality $\semdt{\Xx, s}{\reify M} = \semdt{\Xx, s}{\reify N}$
	follows by taking $i = k$.
\end{proofof}

\AP
Heading towards the proof of \cref{lem:reify:beta:derived}, we introduce some notions.
Let $M$ be a \kl{closed} \kl{first-order} \kl{input lambda-term} \kl{of type} $\otyp$,
$d \in \Dirs_{\Xx,s}$ a direction,
and $v$ a \kl{node} in the \kl{reified} \kl{lambda-tree} $\reify M$.
We call a triple $\tuple{\reify M, d, v}$ a
\intro{configuration}.
For two \kl{configurations} $c, b$
let $c \successor {\Xx, s} b$
if $b$ is the \kl[successor]{$(\Xx, s)$\nobreakdash-successor} of $c$
(recall that the \kl{successor} is unique, if defined).


\AP
Consider a \kl{configuration} $\tuple{\reify M, d, v}$.
If $v$ has a \kl{child}, let $K$ be such that the \kl{subtree} of $\reify M$ starting in the \kl{node} $v\cdot 1$ equals $\reify K$ 
(checking the definition of $\reify M$, where $M$ is \kl{first-order}, we see that all \kl{subtrees} of $\reify M$ are of this form).
We say that $\tuple{\reify M, d, v}$ is \intro{valid} if either
\begin{compactitem}
\item $d=\Rdown$,
\item $d=\Rup_x$, $v$ has a \kl{child}, and $x$ is \kl(lambda){free} in $K$, or
\item $d=\Rup_i$, $v$ has a \kl{child}, and $K$ requires at least $i$ arguments (i.e., $K$ \kl{has type} $\otyp^k\arr\otyp$ with $k\geq i$).
\end{compactitem}
We have the following lemma:

\begin{lemma}
	All \kl{configurations} $\tuple{\reify M, d, v}$ reached while computing $\semdt{\Xx, s}{\reify M}$ are \kl{valid}.
\end{lemma}

\begin{proof}
	By a case-by-case analysis of the definition of \kl[successor]{$(\Xx, s)$\nobreakdash-successor},
	we immediately see that if $c \successor {\Xx, s} b$ and $c$ is \kl{valid}, then $b$ is \kl{valid} as well.
	Additionally, if a \kl{configuration} $\tuple{\reify M, \Rdown, v}$ is \kl{valid} and the \kl{node} $v$ is labelled by $\overline a$,
	then the \kl{configurations} $\tuple{\reify M, \Rup_i, v}$ for $i\in\set{1,\dots,r}$, where $r$ is the \kl{rank} of $a$, are \kl{valid} as well.
\end{proof}

\AP
Let ``$\simul$'' be a binary relation between \kl{valid configurations}.
We say that ``$\simul$'' is a \intro{weak simulation} if,
whenever $c \simul b$ holds for two \kl{valid configurations} $b, c$, we then have
\begin{compactenum}
	\item	if $b \successor {\Xx, s} b'$,
		then there exists a \kl{valid configuration} $c'$
		such that $c \tsuccessor {\Xx, s} c'$ and $c' \simul b'$, and
	\item	if $b=\tuple{\reify N,\Rdown,v}$ and $v$ has label $\overline a$, then $c=\tuple{\reify M,\Rdown,u}$ and $u$ has the same label $\overline a$,
		and $\tuple{\reify N,\Rup_i,v}\simul\tuple{\reify M,\Rup_i,u}$ for all $i\in\set{1,\dots,r}$, where $r$ is the \kl{rank} of $a$.
\end{compactenum}
The following lemma shows that \kl{weak simulation} preserves \kl{derived trees}:
\begin{lemmaOK}
	\label{lem:sim:der}
	If ``$\simul$'' is a \kl{weak simulation}, and $\semdt {\Xx, s} b$ does not contain the special \kl{letter} $\bot$, then
	\begin{align*}
		c \simul b
			\quad\text{implies}\quad
				\semdt {\Xx, s} c = \semdt {\Xx, s} b.
	\end{align*}
\end{lemmaOK}

\begin{proof}
	We proceed by coinduction on \kl{derived trees}.
	Let $c=\tuple{\reify M,d,u}$ and $b=\tuple{\reify N,e,v}$.
	By the definition of the \kl{derived tree} $\semdt {\Xx, s} {\reify N, e, v}$
	there is a maximal sequence of \kl{successors}
	\begin{align*}
		\tuple{\reify N, e, v}
			\tsuccessor {\Xx, s}
				\tuple{\reify N, \Rdown, v'},
	\end{align*}
	where the \kl{node} $v'$ in $\reify N$ is labelled with $\tconst a$,
	and such that
	\begin{align*}
		\semdt {\Xx, s} {\reify N, e, v}
			&= a \applyto \semdt {\Xx, s} {\reify N, \Rup_1, v'} \applyto \cdots \applyto \semdt {\Xx, s} {\reify N, \Rup_r, v'},
	\end{align*}
	where $r$ is the \kl{rank} of the \kl{letter} $a$.
	By assumption $\tuple{\reify M, d, u} \simul \tuple{\reify N, e, v}$.
	Recall that ``$\simul$'' is a \kl{weak simulation},
	thus a repeated \kl{application} of the first item in the definition of a \kl{weak simulation} shows
	\begin{align*}
		\tuple{\reify M, d, u}
			\tsuccessor {\Xx, s}
				\tuple{\reify M, d', u'}
	\end{align*}
	with $\tuple{\reify M, d', u'} \simul \tuple{\reify N, \Rdown, v'}$.
	By the second item in the definition of a \kl{weak simulation} we have that $d'=\Rdown$ and $u'$ is labeled by $\overline a$;
	in particular $\tuple{\reify M, d', u'}$ does not have a \kl{successor} (cf.~the definition of a \kl{successor}).
	By the definition of a \kl{derived tree} we thus have
	\begin{align*}
		\semdt {\Xx, s} {\reify M, d, u}
		 &= a \applyto \semdt {\Xx, s} {\reify M, \Rup_1, u'} \applyto \cdots \applyto \semdt {\Xx, s} {\reify M, \Rup_r, u'}.
	\end{align*}
	This shows that the \kl{derived trees} of
	$\tuple{\reify M, d, u}$ and $\tuple{\reify N, e, v}$
	agree on the label of their \kl{root}.
	Moreover, the second item in the definition of a \kl{weak simulation} also says that
	\begin{align*}
		\tuple{\reify M, \Rup_i, u'} \simul \tuple{\reify N, \Rup_i, v'}&&\mbox{for all }i\in\set{1,\dots,r}.
	\end{align*}
	By coinduction on \kl{derived trees} we thus have
	\begin{align*}
		\semdt {\Xx, s} {\reify M, \Rup_i, u'}  = \semdt {\Xx, s} {\reify N, \Rup_i, v'}&&\mbox{for all }i\in\set{1,\dots,r}.
	\end{align*}
	This shows that all relevant \kl{subtrees} also agree,
	thus concluding the proof.
\end{proof}

Recall that our goal is to prove \cref{lem:reify:beta:derived}, saying that \kl{derived trees} are invariant under \kl{head beta-reductions}.
Fix thus two \kl{closed} \kl{first-order} \kl{input lambda-terms} $M,N$ \kl{of type} $\otyp$, such that $M\headreducesto N$.
Then
\begin{align*}
	M = (\lambda x.K)\applyto L\applyto L_1\applyto\dots\applyto L_j
	&&\mbox{and}&&
	N = \subst K L x \applyto L_1\applyto\dots\applyto L_j.
\end{align*}
\AP
We now define a \emph{concrete} \kl{weak simulation},
denoted by overloading the same symbol ``$\tsimul$'',
between \kl{valid configurations} involving $\reify M$ and $\reify N$.
Define $\tuple{\reify M,d,u}\intro*\tsimul\tuple{\reify N,d,v}$ if either
(for $j$ being the same number as above, i.e., the number of arguments $L_1,\dots,L_j$ in $M$ and $N$)
\begin{compactenum}
\item	$v$ is not of the form $1^j\cdot v'$ and $u=v$
	(i.e., $v$ is outside of $\reify{(\subst K L x)}$, and $u$ is the same \kl{node} in the other \kl{lambda-tree}),
\item	$v=1^j\cdot v'$, and $u=1^{j+2}\cdot v'$, and $u$ is not labeled by $\overline x$ in $\reify M$
	(i.e., $v$ is inside the ``$\reify K$ part'' of $\reify{(\subst K L x)}$,
	and $u$ is the corresponding \kl{node} of $\reify K$ in $\tat\applyto(\tlambda{x}\applyto\reify K)\applyto\reify L$), or
\item	$v$ can be written as $v=1^j\cdot v'\cdot v''$, where $1^{j+2}\cdot v'$ has label $\overline x$ in $\reify M$, and $u=1^j\cdot 2\cdot v''$
	(i.e., $v$ is inside some $\reify L$ in $\reify{(\subst K L x)}$, 
	and $u$ is the corresponding \kl{node} of $\reify L$ in $\tat\applyto(\tlambda{x}\applyto\reify K)\applyto\reify L$).
\end{compactenum}
\begin{figure*}
	\centering
	\def\svgscale{0.5}\import{pics/}{simulation.pdf_tex_ok}\vspace{-1ex}
	\caption{An illustration of the \kl{weak simulation} relation ``$\tsimul$''}
	\label{fig:simulation}
\end{figure*}
See \cref{fig:simulation} for an illustration of the definition above for $j = 2$.
As a special case of the last condition,
the \kl{root} of $\reify L$ in $\reify M$ is in relation with the \kl{root} of some copy of $\reify L$ in $\reify N$.
Note that $\tuple{\reify M,d,u}\tsimul\tuple{\reify N,d,v}$ holds only for \kl{configurations} with the same direction $d$.
Note also that for every \kl{node} $v$ of $\reify N$ we can find a (unique) corresponding \kl{node} $u$ in $\reify M$,
but it is not the case that for every \kl{node} $u$ of $\reify M$ there is a corresponding \kl{node} $v$ in $\reify N$.
In particular, in ``$\tsimul$'' we do not have any pairs with $u=1^j$ nor $u=1^{j+1}$
(i.e.,~with $u$ pointing to the ``$\tat$'' or ``$\tlambda{x}$'' at the top of $\tat\applyto(\tlambda{x}\applyto\reify K)\applyto\reify L$);
also, \kl{nodes} labelled with $\tvar x$ in $\reify K$ from $\reify M$ (if any) are not in relation with any \kl{node} of $\reify N$;
finally, if $x$ does not occur in $K$,
then additionally no \kl{node} in $\reify L$
is in relation with a \kl{node} in $\reify N$.

\begin{lemmaOK}
	\label{lem:beta:successor}
	The binary relation ``$\tsimul$'' is a \kl{weak simulation}.
\end{lemmaOK}

Before proving \cref{lem:beta:successor}, let us see how \cref{lem:reify:beta:derived} follows from it: 

\begin{proofof}{\cref{lem:reify:beta:derived}}
	Recall that $\semdt {\Xx, s} {\reify N} = \semdt {\Xx, s} b$
	with \kl{configuration} $b = \tuple{\reify N, \Rdown, v}$,
	where $v=\varepsilon$ is the \kl{root} of the \kl{lambda-tree} $\reify N$.
	In order to ensure $c \tsimul b$, we have to be a bit careful while choosing the \kl{configuration} $c$ for $\reify M$.
	We take $c=\tuple{\reify M, \Rdown,u}$, where $u$ is chosen as follows:
	\begin{compactitem}
	\item	If $j\geq 1$, we can simply take $u=\varepsilon$ (where $j$ is the same number as previously, i.e., the number of arguments $L_1,\dots,L_j$ in $M$ and $N$).
	\item	If $j=0$ and $K\neq x$, we rather take $u=1\cdot 1$ (which is the \kl{node} where $\reify K$ starts).
		Note that $\tuple{\reify M, \Rdown, \varepsilon}\successor {\Xx, s}^*c$ in two steps.
	\item	Finally, if $j=0$ and $K=x$, we take $u=2$ (which is the \kl{node} where $\reify L$ starts).
		This time we also have $\tuple{\reify M, \Rdown, \varepsilon}\successor {\Xx, s}^*c$.
	\end{compactitem}
	Thus, in any case $\semdt {\Xx, s} {\reify M} = \semdt {\Xx, s} c$.
	Moreover, we have $c \tsimul b$ by definition.
	By \cref{lem:sim:der} we have $\semdt{\Xx, s} b = \semdt{\Xx, s}c$, as required.
\end{proofof}

What remains is to prove \cref{lem:beta:successor}:
\begin{proofof}{\cref{lem:beta:successor}}
	Consider two \kl{valid configurations} $c=\tuple{\reify M,d,u}$ and $b=\tuple{\reify N,d,v}$
	such that $c\tsimul b$.
	Observe first that then necessarily $u$ and $v$ have the same label,
	and that $\tuple{\reify M,e,u}\tsimul\tuple{\reify N,e,v}$ for every other direction $e$ for which the \kl{configurations} are \kl{valid}.
	This immediately implies the second item in the definition of a \kl{weak simulation}.

	Let us check the first item.
	To this end, consider $b'=\tuple{\reify N,e,v'}$ such that $b \successor {\Xx, s} b'$; we have to find $c'$ such that $c \tsuccessor {\Xx, s} c'$ and $c' \tsimul b'$.
	By the definition of a \kl{successor}, $v'$ is either a \kl{child} of $v$, or a \kl{parent} of $v$.
	A natural candidate for $c'$ is the unique \kl{$(\Xx,s)$-successor} of $c$.
	Because $c$ and $b$ have the same direction and the same \kl{node} label, 
	the \kl{successor} indeed exists, and it is of the form $\tuple{\reify M,e,u'}$,
	where $u'$ is in the same relation to $u$ as $v'$ to $v$ (i.e., $u'=\parent(u)$ if $v'=\parent(v)$, and $u'=\child_i(u)$ if $v'=\child_i(v)$).
	If $v$ and $v'$ are in the same ``part'' of $\reify N$, that is, both in $\reify L$, both in $\reify{(\subst K L x)}$ but outside of $\reify L$,
	or both outside of $\reify{(\subst K L x)}$, then we have $c' \tsimul b'$, and we are done.
	
	The situation is more complicated only when $v$ and $v'$ are in different parts.
	Let us first consider the border of $\reify{(\subst K L x)}$:
	\begin{compactenum}
		\begin{figure*}
			\centering
			\def\svgscale{0.5}\import{pics/}{simulation_1.pdf_tex_ok}\vspace{-1ex}
			\caption{Illustration of Case 1 in the proof of \cref{lem:beta:successor}}
			\label{fig:simulation:1}
		\end{figure*}
		
	\item
		Suppose that $v=1^{j-1}$, $v'=1^j$ (recall that $j$ is the number of arguments following the redex, i.e., $N=\subst K L x \applyto L_1\applyto\dots\applyto L_j$,
		so $v'$ is the \kl{root} of $\reify{(\subst K L x)}$, and $v$ its \kl{parent}).
		Then $u=1^{j-1}$ and $e=\Rdown$.
		If $K\neq x$, in $\reify M$ we can make three \kl{successor steps},
		going through the \kl{nodes} labeled by $\tat$ and $\tlambda{x}$ to the \kl{root} of $\reify K$ in $\tat\applyto(\tlambda{x}\applyto\reify K)\applyto\reify L$;
		for $c'=\tuple{\reify M,\Rdown,1^{j+2}}$ we have $c'\tsimul b'$;
		see \cref{fig:simulation:1},
		where the thick arrow on the right
		is simulated by the three dashed arrows on the left,
		denoting \kl{successor steps}.
		
		If $K=x$, we need three more \kl{successor steps}: from the $\overline{x}$-labelled \kl{root} of $\reify K$ we go up to the $\tlambda{x}$-labelled \kl{node} with direction $\Rup_x$,
		then to the $\tat$-labelled \kl{node} with direction $\Rup_1$, and finally we go down to the \kl{root} of $\reify L$ with direction $\Rdown$;
		for $c'=\tuple{\reify M,\Rdown,1^j\cdot 2}$ we have $c'\tsimul b'$.
		
		\begin{figure*}
			\centering
			\def\svgscale{0.5}\import{pics/}{simulation_2.pdf_tex_ok}\vspace{-1ex}
			\caption{Illustration of Cases 2 and 3 in the proof of \cref{lem:beta:successor}}
			\label{fig:simulation:2:3}
		\end{figure*}
		
	\item	Suppose that $v=1^j$, $v'=1^{j-1}$, and $e=\Rup_i$.
		Then $u=1^{j+2}$ is the \kl{root} of $\reify K$ in $\tat\applyto(\tlambda{x}\applyto\reify K)\applyto\reify L$.
		Note that $\subst K L x$ \kl{has type} $\otyp^j\arr\otyp$, so because $b'$ is \kl{valid}, we have $i\leq j$.
		It is important that $\lambda x.K$ is also a \kl{subterm} of $M$ and \kl{has type} $\otyp^{j+1}\arr\otyp$;
		because $M$ is an \kl{input lambda-term}, we have $j+1\leq\mar(M)\leq\mar(\Gg)=s$.
		We can thus make three \kl{successor steps} in $\reify M$
		(c.f.~\cref{fig:simulation:2:3}),
		going up through the \kl{nodes} labeled by $\tlambda x$ and $\tat$:
		\begin{align*}
			\tuple{\reify M,d,1^{j+2}}
			\successor{\Xx,s}\tuple{\reify M,\Rup_i,1^{j+1}}
			\successor{\Xx,s}\tuple{\reify M,\Rup_{i+1},1^j}
			\successor{\Xx,s}\tuple{\reify M,\Rup_i,1^{j-1}}=c'.
		\end{align*}
	\item	Finally, suppose that $v=1^j$, $v'=1^{j-1}$, and $e$ is not of the form $\Rup_i$.
		Then necessarily $e=\Rup_y$ and again $u=1^{j+2}$.
		Recall that $M$ is \kl{closed}, implying that $L$ is \kl{closed}.
		Thus $x$ is not \kl(lambda){free} in $\subst K L x$, so $y\neq x$ because $b'$ is \kl{valid}.
		This allows us to make three \kl{successor steps} in $\reify M$
		(c.f.~\cref{fig:simulation:2:3}),
		going up through the \kl{nodes} labeled by $\tlambda{x}$ and $\tat$; we take $c'=\tuple{\reify M,\Rup_y,1^{j-1}}$.
	\end{compactenum}
	
	\begin{figure*}
		\centering
		\def\svgscale{0.5}\import{pics/}{simulation_4.pdf_tex_ok}\vspace{-1ex}
		\caption{Illustration of the last case in the proof of \cref{lem:beta:successor}}
		\label{fig:simulation:4}
	\end{figure*}
	
	Next, we consider the border of $\reify L$.
	Recall that $M$ is a \kl{first-order} \kl{lambda-term}, implying that $L$ is \kl{of type} $\otyp$.
	Thus, because $b'$ is \kl{valid}, we can never leave $\reify L$ with direction $\Rup_i$.
	Likewise, because $L$ is \kl{closed} and $b'$ is \kl{valid}, we can never leave $\reify L$ with direction $\Rup_y$, for any \kl(lambda){variable} $y$.
	It remains to consider the case when we enter $L$ from above.
	This means that $v'=1^j\cdot w$ is the \kl{root} of some copy of $\reify L$ in $\reify N$, while the \kl{node} $1^{j+2}\cdot w$ in $\reify M$ is labelled by $\overline x$.
	We then have $e=\Rdown$.
	The case of $K=x$ is already covered by Item 1 above; we may assume that $K\neq x$.
	Then $u$ is the \kl{parent} of the $\overline x$-labelled \kl{node} $1^{j+2}\cdot w$,
	and the \kl{successor} of $c$ is $\tuple{\reify M,\Rdown,1^{j+2}\cdot w}$.
	Although $x$ may occur in some \kl{lambda-binders} in $K$, 
	we know that for the considered occurrence of $x$ we have \kl{substituted} $L$, so it is not under the scope of $\lambda x$ inside $K$
	(i.e., no ancestor of the \kl{node} $1^{j+2}\cdot w$ inside $\reify K$ is labelled by $\tlambda{x}$).
	Thus the sequence of \kl{successors} from $\tuple{\reify M,\Rdown,1^{j+2}\cdot w}$ goes up with direction $\Rup_x$ until it reaches the $\tlambda{x}$-labelled \kl{node} $1^{j+1}$.
	The \kl{successor} of $\tuple{\reify M,\Rup_x,1^{j+1}}$ is $\tuple{\reify M,\Rup_1,1^j}$, and its \kl{successor} is $\tuple{\reify M,\Rdown,1^j\cdot 2}$
	(whose \kl{node} is the \kl{root} of $\reify L$;
	c.f.~\cref{fig:simulation:4});
	taking this \kl{configuration} as $c'$, we have $c'\tsimul b'$, as required.
\end{proofof}

\end{document}